\documentclass[11pt]{article}

\usepackage{amsthm,amsmath,amssymb}
\usepackage{xcolor}
\usepackage{graphicx}
\usepackage{booktabs}
\usepackage{fullpage}
\usepackage{caption}
\usepackage{hyperref}
\usepackage{enumerate}
\usepackage{enumitem}
\usepackage{lineno}

\newcommand{\REMOVE}[1]{}

\newcommand{\eps}{\varepsilon}
\newcommand{\vol}{{\rm vol}}
\newcommand{\lca}{{\rm lca}}

\newcommand{\BB}{{\mathcal B}}
\newcommand{\PP}{{\mathcal P}}
\newcommand{\Pin}{{\mathcal Pin}}

\newcommand{\VV}{{\mathcal V}}

\newcommand{\calS}{{\mathcal S}}

\newcommand{\GG}{{\mathcal G}}

\newcommand{\nodeexpansion}{{\sf node-expansion}}
\newcommand{\oldbarexp}{{\sf old-bar-expansion}}
\newcommand{\newbarexp}{{\sf new-bar-expansion}}

\newcommand{\crimpreduction}{{\sf crimp-reduction}}
\newcommand{\spurreduction}{{\sf spur-reduction}}
\newcommand{\nodesplit}{{\sf node-split}}
\newcommand{\merge}{{\sf merge}}
\newcommand{\subdivision}{{\sf subdivision}}

\newcommand{\pinextraction}{{\sf pin-extraction}}
\newcommand{\Vshortcut}{{\sf V-shortcut}}
\newcommand{\Lshortcut}{{\sf L-shortcut}}

\newcommand{\spurshortcut}{{\sf spur-shortcut}}
\def\Lhat{{\smash{\widehat{L}}\vphantom{L}}}

\newtheorem{lemma}{Lemma}

\newtheorem{theorem}{Theorem}

\newtheorem{corollary}{Corollary}

\newtheorem{remark}{Remark}

\date{}
\begin{document}

\title{Recognizing Weakly Simple Polygons\footnote{A preliminary version of this paper appeared in the
\emph{Proceedings of the 32nd International Symposium on Computational Geometry (SoCG 2016)},
\url{doi:10.4230/LIPIcs.SoCG.2016.8}.}}
\author{
Hugo A. Akitaya\thanks{Department of Computer Science, Tufts University, Medford, MA.
Email: \texttt{hugo.alves\_akitaya@tufts.edu}, \texttt{aloupis.greg@gmail.com}, \texttt{cdtoth@eecs.tufts.edu}}
\and
Greg Aloupis\footnotemark[2]
\and
Jeff Erickson\thanks{Department of Computer Science, University of Illinois, Urbana-Champaign, IL.
Email: \texttt{jeffe@illinois.edu}}
\and
Csaba D. T\'oth\footnotemark[2]~\thanks{Department of Mathematics, California State University Northridge, Los Angeles, CA.}
}

\maketitle

\begin{abstract}
\noindent We present an $O(n\log n)$-time algorithm that determines whether a given $n$-gon in the plane is weakly simple. This improves upon an $O(n^2\log n)$-time algorithm by Chang, Erickson, and Xu~\cite{CEX15}. Weakly simple polygons are required as input for several geometric algorithms. As such, recognizing simple or weakly simple polygons is a fundamental problem.
\end{abstract}

\paragraph{Keywords:} simple polygon, combinatorial embedding, perturbation
\paragraph{MSC:}
05C10, 
05C38, 
52C45, 
68R10. 

\section{Introduction}
\label{sec:intro}

A polygon is \emph{simple} if it has distinct vertices and interior-disjoint edges that do not pass through vertices. Geometric algorithms are often designed for simple polygons, but many also work for degenerate polygons that do not ``self-cross.'' A polygon with at least three vertices is \emph{weakly simple} if for every $\eps>0$, the vertices can be perturbed within a ball of radius $\eps$ to obtain a simple polygon. Such polygons arise naturally in numerous applications, e.g., for modeling planar networks or as the geodesic hull of points within a simple polygon (Figure~\ref{fig:intro}).

\begin{figure}[htb]
		\centering
		\def\svgwidth{.7\textwidth}
\begingroup%
  \makeatletter%
  \providecommand\color[2][]{%
    \errmessage{(Inkscape) Color is used for the text in Inkscape, but the package 'color.sty' is not loaded}%
    \renewcommand\color[2][]{}%
  }%
  \providecommand\transparent[1]{%
    \errmessage{(Inkscape) Transparency is used (non-zero) for the text in Inkscape, but the package 'transparent.sty' is not loaded}%
    \renewcommand\transparent[1]{}%
  }%
  \providecommand\rotatebox[2]{#2}%
  \ifx\svgwidth\undefined%
    \setlength{\unitlength}{497.43540039bp}%
    \ifx\svgscale\undefined%
      \relax%
    \else%
      \setlength{\unitlength}{\unitlength * \real{\svgscale}}%
    \fi%
  \else%
    \setlength{\unitlength}{\svgwidth}%
  \fi%
  \global\let\svgwidth\undefined%
  \global\let\svgscale\undefined%
  \makeatother%
  \begin{picture}(1,0.26614848)%
    \put(0,0){\includegraphics[width=\unitlength,page=1]{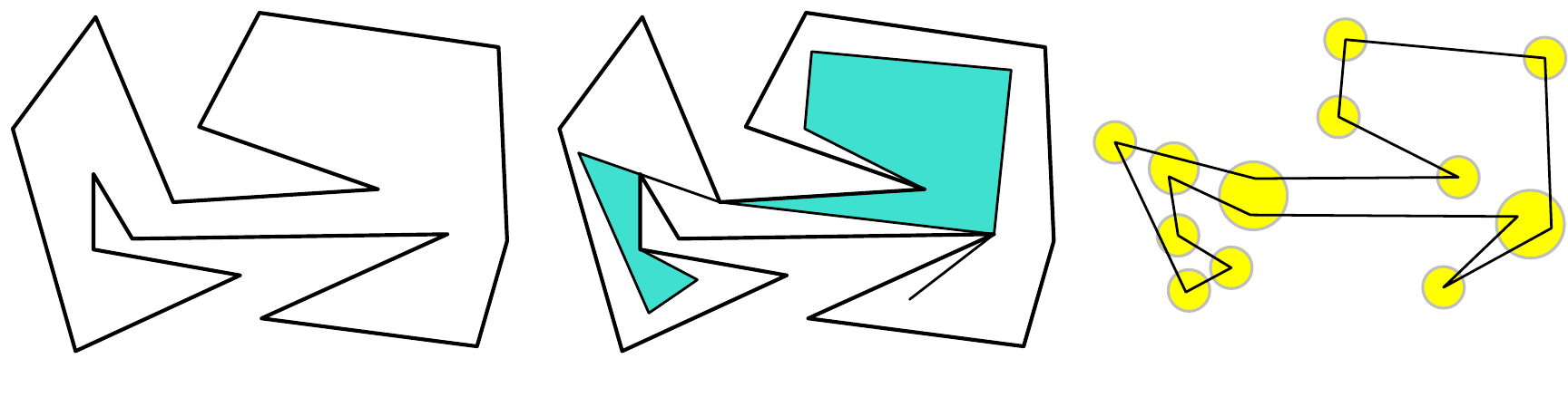}}%
    \put(0.14027634,0.01199398){\color[rgb]{0,0,0}\makebox(0,0)[lb]{\smash{(a)}}}%
    \put(0.49699588,0.01199398){\color[rgb]{0,0,0}\makebox(0,0)[lb]{\smash{(b)}}}%
    \put(0.8387838,0.01199398){\color[rgb]{0,0,0}\makebox(0,0)[lb]{\smash{(c)}}}%
    \put(0,0){\includegraphics[width=\unitlength,page=2]{fig-intro.pdf}}%
  \end{picture}%
\endgroup%
		\caption{\small
		(a) A simple polygon $P$ with 16 vertices.
		(b) Eight points in the interior of $P$ (solid dots); their geodesic hull is a weakly simple polygon $P'$ with 14 vertices.
		(c) A perturbation of $P'$ into a simple polygon.
		}
		\label{fig:intro}
\end{figure}

Several alternative definitions have been proposed for weakly simple polygons, formalizing the intuition that such polygons do not self-cross. Some of these definitions were unnecessarily restrictive or incorrect; see~\cite{CEX15} for a detailed discussion and five equivalent definitions for weak simplicity of a polygon. Among others, a result by {Rib\'o Mor}~\cite[Theorem 3.1]{Mor06} implies an equivalent definition in terms of Fr\'echet distance, in which a polygon is perturbed into a simple closed curve (see Section~\ref{sec:preliminaries}).
This definition is particularly useful for recognizing weakly simple polygons, since it allows transforming edges into polylines (by subdividing the edges with Steiner points, which may be perturbed). With suitable Steiner points, the perturbation of a vertex incurs only local changes. (In other words, we do not need to worry about stretchability of the perturbed configuration.)

We can decide whether an $n$-gon in the plane is simple in $O(n\log n)$ time by a sweepline algorithm~\cite{SH76}.
Chazelle's polygon triangulation algorithm also recognizes simple polygons (in $O(n)$ time),
because it only produces a triangulation if the input is simple~\cite{Cha91}.
Recognizing weakly simple polygons, however, is more subtle.
Skopenkov~\cite{Sko03} gave a combinatorial characterization of the topological obstructions to weak simplicity in terms of line graphs. Cortese et al.~\cite{CDP+09} gave an $O(n^6)$-time algorithm to recognize weakly simple $n$-gons.
Chang et al.~\cite{CEX15} improved the running time to $O(n^2\log n)$ in general; and to $O(n\log n)$ in several special cases.
They identified two features that are difficult to handle: A \emph{spur} is a vertex whose incident edges overlap, and a \emph{fork} is a vertex that lies in the interior of an edge. (A vertex may be both a fork and a spur.)
They gave an easy algorithm for polygons that have neither forks nor spurs, and two more involved ones for polygons with spurs but no forks and for polygons with forks but no spurs, all three running in $O(n\log n)$ time.
In the presence of both forks and spurs, they presented an $O(n^2 \log n)$ time algorithm that eliminates forks by subdividing all edges
that contain vertices in their interiors, potentially creating a quadratic number of vertices.

We show how to manage both forks and spurs efficiently, while building on ideas from~\cite{CEX15,CDP+09} and from
Arkin et al.~\cite{ABD+09}, and obtain the following main results.

\begin{theorem}\label{thm:main}
Deciding whether a polygon $P$ with $n$ vertices in the plane is weakly simple takes $O(n\log n)$ time.
\end{theorem}
\begin{theorem}\label{thm:perturb}
Given a weakly simple polygon $P$ with $n$ vertices and a constant $\eps>0$, a simple polygon with $2n$ vertices within Fr\'echet distance $\eps$ from $P$ can be computed in $O(n\log n)$ time.
\end{theorem}

Our decision algorithm is detailed in Sections~\ref{sec:preprocess}--\ref{sec:tree-exp}.
It consists of three phases, simplifying the input polygon by a sequence of reduction steps.
First, the \emph{preprocessing} phase rules out edge crossings in $O(n\log n)$ time
and applies known reduction steps such as \emph{crimp reductions} and \emph{node expansions} (Section~\ref{sec:preprocess}).
Second, the \emph{bar simplification} phase successively eliminates all forks (Section~\ref{sec:bars}).
Third, the \emph{spur elimination} phase eliminates all spurs (Section~\ref{sec:tree-exp}).
When neither forks nor spurs are present, we can decide weak simplicity in $O(n)$ time~\cite{CDP+09}.
Finally, by reversing the sequence of operations, we can also perturb any weakly simple polygon into
a simple polygon in $O(n\log n)$ time (Section~\ref{sec:perturb}).

\section{Preliminaries}
\label{sec:preliminaries}
In this section, we review previously established definitions and known methods from~\cite{CEX15} and~\cite{CDP+09}.

\medskip\noindent{\bf Polygons and weak simplicity.}
An \emph{arc} in $\mathbb{R}^2$ is a continuous function $\gamma:[0,1]\rightarrow \mathbb{R}^2$.
A \emph{closed curve} is a continuous function (map) $\gamma:\mathbb{S}^1\rightarrow \mathbb{R}^2$.
A closed curve $\gamma$ is \emph{simple} (also known as a \emph{Jordan curve}) if it is injective.
A (\emph{simple}) \emph{polygon} is the image of a piecewise linear ({\em simple}) closed curve.
Thus a polygon $P$ can be represented by a cyclic sequence of points $(p_0,\ldots , p_{n-1})$, called \emph{vertices}, where the image of $\gamma$ consists of line segments $p_0p_1,\ldots , p_{n-2}p_{n-1}$, and $p_{n-1}p_0$ in this cyclic order. Note that a nonsimple polygon may have repeated vertices and overlapping edges~\cite{Gru12}. Similarly, a \emph{polygonal chain} (alternatively, \emph{path}) is the image of a piecewise linear arc, and can be represented by a sequence of points $[p_0,\ldots , p_{n-1}]$.

A polygon $P=(p_0,\ldots , p_{n-1})$ is \emph{weakly simple} if $n=2$,
or if $n>2$ and for every $\eps>0$ there is a simple polygon $(p'_0,\ldots , p'_{n-1})$ such that $|p_i,p'_i|<\eps$ for all $i=0,\ldots , n-1$. This definition is difficult to work with because a small perturbation of a vertex modifies the two incident edges, which may be long, and the effect of a perturbation is not localized. Combining earlier results from~\cite{CDR02}, \cite{CDP+09}, and \cite[Theorem 3.1]{Mor06}, an equivalent definition was formulated by Chang et al.~\cite{CEX15} in terms of Fr\'echet distance: A polygon given by $\gamma: \mathbb{S}^1\rightarrow \mathbb{R}^2$ is weakly simple if for every $\eps>0$ there is a simple closed curve $\gamma':\mathbb{S}^1\rightarrow \mathbb{R}^2$ such that ${\rm dist}_F(\gamma, \gamma')<\eps$, where ${\rm dist}_F$ denotes the Fr\'echet distance between two closed curves.
The curve $\gamma'$ can approximate an edge of the polygon by a polyline, and any perturbation of a vertex can be restricted to a small neighborhood. With this definition, recognizing weakly simple polygons becomes a combinatorial problem, as explained below. Note that in topology, the broader question of \emph{isotopic embeddability} has been considered~\cite{Min97,Sko03}: Given a continuous map $f:A\rightarrow \mathbb{R}^d$ for a simplicial complex $A$, is it isotopic to some \emph{injective} continuous map (i.e., \emph{embedding}) $g: A\rightarrow \mathbb{R}^d$?

\medskip\noindent{\bf Bar decomposition and image graph.}
Two edges of a polygon $P$ \emph{cross} if their interiors intersect at precisely one point; we call this an \emph{edge crossing}. Weakly simple polygons cannot have edge crossings. In the remainder of this section,
we assume that such crossings have been ruled out.
Two edges of $P$ \emph{overlap} if their intersection is a (nondegenerate) line segment.
The transitive closure of the overlap relation is an equivalence relation on the edges of $P$; see Figure~\ref{fig:bars}(a) where equivalence classes are represented by purple regions.
The union of all edges in an equivalence class is called a \emph{bar}.\footnote{We adopt terminology from~\cite{CEX15}.} All bars of a polygon can be computed in $O(n\log n)$ time~\cite{CEX15}.
The bars are open line segments that are pairwise disjoint.
There are at most $n$ bars, since the bars are unions of disjoint subsets of edges.

The vertices and bars of $P$ define a planar straight-line graph $G$, called the \emph{image graph} of $P$. We call the vertices and edges of $G$ \emph{nodes} and \emph{segments}\footnotemark[1] to distinguish them from the vertices and edges of $P$. Every node that is not in the interior of a bar is called \emph{sober}\footnotemark[1]. The set of nodes in $G$ is $\{p_0,\ldots , p_{n-1}\}$ (note that $P$ may have repeated vertices that correspond to the same node); two nodes are connected by a segment in $G$ if they are consecutive nodes along a bar; see Figure~\ref{fig:bars}(b).
Hence $G$ has $O(n)$ nodes and segments, and it can be computed in $O(n\log n)$ time~\cite{CEX15}. Note, however, that up to $O(n)$ edges of $P$ may pass through a node of $G$, and there may be $O(n^2)$ edge-node pairs such that an edge of $P$ passes through a node of $G$.
An $O(n \log n)$-time algorithm cannot afford to compute these pairs explicitly.

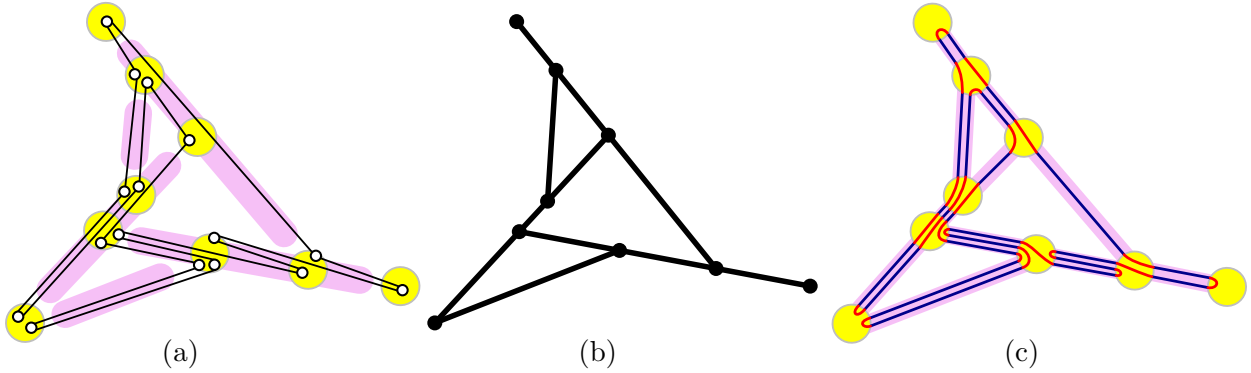
\begin{figure}[h!tbp]
\centering
		\def\svgwidth{\textwidth}
\begingroup%
  \makeatletter%
  \providecommand\color[2][]{%
    \errmessage{(Inkscape) Color is used for the text in Inkscape, but the package 'color.sty' is not loaded}%
    \renewcommand\color[2][]{}%
  }%
  \providecommand\transparent[1]{%
    \errmessage{(Inkscape) Transparency is used (non-zero) for the text in Inkscape, but the package 'transparent.sty' is not loaded}%
    \renewcommand\transparent[1]{}%
  }%
  \providecommand\rotatebox[2]{#2}%
  \ifx\svgwidth\undefined%
    \setlength{\unitlength}{540.55043945bp}%
    \ifx\svgscale\undefined%
      \relax%
    \else%
      \setlength{\unitlength}{\unitlength * \real{\svgscale}}%
    \fi%
  \else%
    \setlength{\unitlength}{\svgwidth}%
  \fi%
  \global\let\svgwidth\undefined%
  \global\let\svgscale\undefined%
  \makeatother%
  \begin{picture}(1,0.30037502)%
    \put(0,0){\includegraphics[width=\unitlength,page=1]{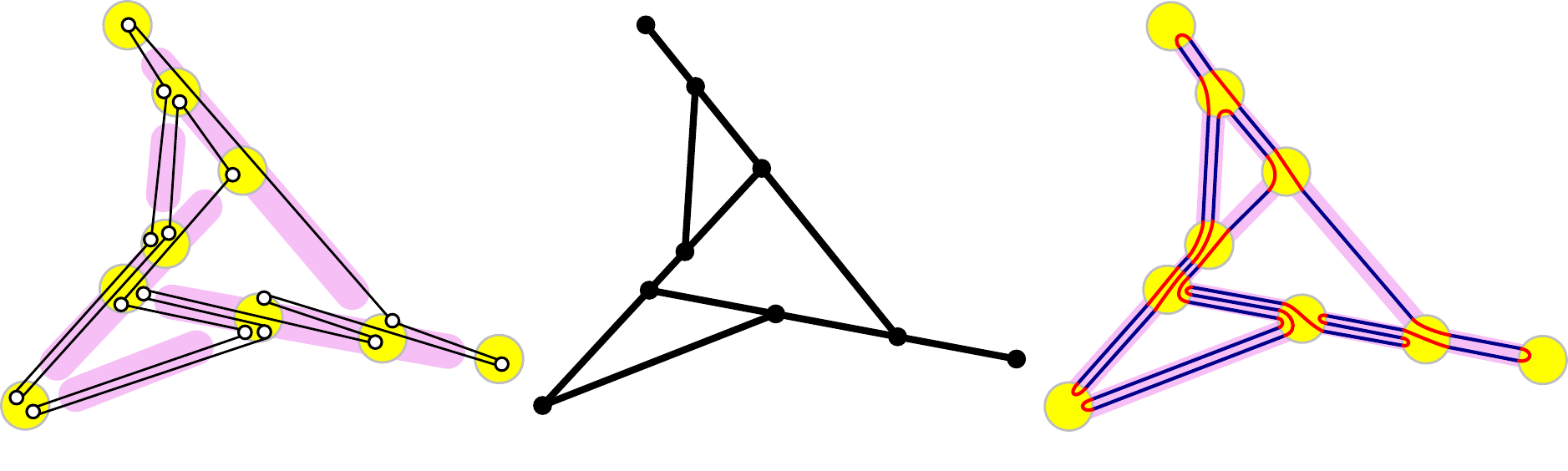}}%
    \put(0.12635056,0.01062479){\color[rgb]{0,0,0}\makebox(0,0)[lb]{\smash{(a)}}}%
    \put(0.46192666,0.01062479){\color[rgb]{0,0,0}\makebox(0,0)[lb]{\smash{(b)}}}%
    \put(0.80449391,0.01062479){\color[rgb]{0,0,0}\makebox(0,0)[lb]{\smash{(c)}}}%
  \end{picture}%
\endgroup%
\caption{\small
(a) The bar decomposition for a weakly simple polygon $P$ with 16 vertices
    ($P$ is perturbed into a simple polygon for clarity).
(b) The image graph of $P$.
(c) A perturbation in a strip system of $P$.
\label{fig:bars}}
\end{figure}

\noindent{\bf Operations.}
We use certain elementary operations that successively
modify a polygon and ultimately eliminate forks and spurs.
An operation that produces a weakly simple polygon if and only if it is performed on a weakly simple polygon is called \emph{ws-equivalent}.
Several such operations are already known
(e.g., crimp reduction, node expansion, bar expansion).
We shall use these and introduce several new operations in Sections~\ref{ssec:bar-expansion}--\ref{sec:tree-exp}.

\medskip\noindent{\bf Combinatorial characterization of weak simplicity.}
To show that an operation is ws-equivalent, it suffices to provide suitable simple $\eps$-perturbations for all $\eps>0$.
We use a combinatorial representation of an $\eps$-perturbation (independent of $\eps$ or any specific embedding).
When a weakly simple polygon $P$ is perturbed into a simple polygon, overlapping edges in $P$ are perturbed into interior-disjoint near-parallel edges, which define an ordering. It turns out that these orderings over all segments of the image graph are sufficient to encode an $\eps$-perturbation and to (re)construct an $\eps$-perturbation.

We rely on the notion of ``strip system'' introduced in \cite[Appendix B]{CEX15}.
Similar concepts have previously been used in \cite{CDR02,CDP+09,FT14,Min97,Sko03}.
Let $P$ be a polygon and $G$ its image graph. Without loss of generality, we assume that no bar is vertical
(so that the above-below relationship is defined between disjoint segments parallel to a bar).
For every $\eps>0$, the \emph{$\eps$-strip-system} of $P$ consists of the following regions:
\begin{itemize}\itemsep -2pt
\item For every node $u$ of $G$, let $D_u$ be a disk of radius $\eps$ centered at $u$.
\item For every segment $uv$, let the \emph{corridor} $N_{uv}$ be the set of points at distance at most
$\eps^2$ from $uv$, outside of the disks $D_u$ and $D_v$, that is,
$N_{uv}=\{p\in \mathbb{R}^2: {\rm dist}(p,uv)\leq \eps^2, p\not\in D_u\cup D_v\}$.
\end{itemize}
Denote by $U_\eps$ the union of all these disks and corridors. There is a sufficiently small $\eps_0=\eps_0(P)>0$, depending on $P$,
such that the disks $D_u$ are pairwise disjoint, the corridors $N_{uv}$ are pairwise disjoint, and
every corridor $N_{uv}$ of a segment intersects only the disks at its endpoints $D_u$ and $D_v$.
These properties hold for all $\eps$, $0<\eps<\eps_0$.

A polygon is \emph{in the $\eps$-strip-system} of $P$ if its edges alternate between an edge that connects the boundaries of two disks $D_u$ and $D_v$ and whose interior is contained in $N_{uv}$; and an edge between two points on the boundary of a disk. In particular, the edges of $P$ that lie in a disk $D_u$ or a corridor $N_{uv}$ form a perfect matching. See Figure~\ref{fig:bars}(c) for an example, where the edges within the disk $D_u$ are drawn with circular arcs for clarity.
Let $\Phi(P)$ be the set of simple polygons in the $\eps$-strip-system of $P$ that cross the disks and corridors in the same order as $P$ traverses the corresponding nodes and segments of $G$. It is clear that every $Q\in \Phi(P)$ is within Fr\'echet distance $\eps$ from $P$. By \cite[Theorem~B.2]{CEX15}, $P$ is weakly simple if and only if $\Phi(P)\neq \emptyset$.

\medskip\noindent{\bf Combinatorial representation by signatures.}
Let $Q$ be a polygon in the strip system of $P$. For each segment $uv$, the above-below relationship of the edges of $Q$ in $N_{uv}$ is a total order. We define the \emph{signature} of $Q\in \Phi(P)$, denoted $\sigma(Q)$, as the collection of these total orders for all segments of $G$.

Given the signature $\sigma(Q)$ of a polygon $Q$ in the strip system of $P$, we can easily (re)construct a simple polygon $Q'$ with the same signature in the $\eps$-strip-system of $P$ for any $0<\eps<\eps_0$. For every segment $uv$ of $G$, let the \emph{volume} $\vol(uv)$ be the number of edges of $P$ that lie on $uv$. Place $\vol(uv)$ parallel line segments between $\partial D_u$ and $\partial D_v$ in $N_{uv}$ of the $\eps$-strip-system of $P$. Finally, for every disk $D_u$, construct a straight-line perfect matching between the endpoints of these edges that lie in $\partial D_u$: connect the endpoints of two edges if they correspond to adjacent edges of $P$. It is easily verified that the Fr\'echet distance between $Q$ and $Q'$ is at most $2\eps$. Furthermore, $Q\in \Phi(P)$ implies $Q'\in \Phi(P)$, since $Q$ and $Q'$ determine the same perfect matching between corresponding endpoints on $\partial D_u$ at every node $u$.

\begin{remark}\label{rem:combin}
{\rm The construction above has two consequences: (1) To prove weak simplicity, it is enough to find a signature that defines a simple perturbation. In other words, the signature can witness weak simplicity (independent of the value of $\eps$). (2) Weak simplicity of a polygon depends only on the \emph{combinatorial embedding} of the image graph $G$ (i.e., the counterclockwise order of edges incident to each vertex), as long as $G$ is a planar graph. Consequently, when an operation modifies the image graph, it is enough to maintain the combinatorial embedding of $G$ (the  precise coordinates of the nodes do not matter).}
\end{remark}

In the presence of spurs, the size of a signature is $O(n^2)$, and this bound is the best possible.
We use this simple combinatorial representation in our proofs of correctness, but our algorithm does not maintain it explicitly. In Section~\ref{sec:perturb}, we introduce another combinatorial representation of $O(n)$ size that uses
the ordering of the edges in each bar (rather than each segment) of the image graph.

\medskip\noindent{\bf Combinatorially different perturbations.}
In the absence of spurs, a polygon $P$ determines a unique noncrossing perfect matching in each disk $D_u$, hence a unique noncrossing 2-regular graph in the $\eps$-strip-system of $P$~\cite[Section~3.3]{CEX15}.
Consequently, to decide whether $P$ is weakly simple it is enough to check whether this graph is connected. The uniqueness no longer holds in the presence of spurs. In fact, it is not difficult to construct weakly simple $n$-gons that admit $2^{\Theta(n)}$ perturbations into simple polygons that are combinatorially different (i.e., have different bar-signatures); see Figure~\ref{fig:ambig}.

\begin{figure}[htb]
\centering
\includegraphics[width=.7\textwidth]{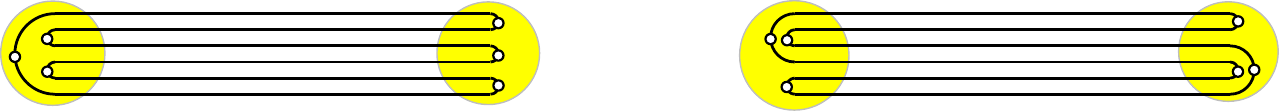}
\caption{\small
Two perturbations of a weakly simple polygon on 6 vertices (all of them spurs) that
alternate between two distinct points in the plane.
\label{fig:ambig}}
\end{figure}

\section{Preprocessing}
\label{sec:preprocess}

We are given a polygon $P=(p_0,\ldots , p_{n-1})$ in the plane. By a standard line sweep~\cite{SH76},
we can test whether any two edges properly cross;
if they do, the algorithm halts and reports that $P$ is not weakly simple.
We then simplify the polygon, using some known steps from~\cite{ABD+09,CEX15},
and some new ones. All of this takes $O(n\log n)$ time.

\subsection{Crimp reduction}
\label{ssec:crimp}

Arkin et al.~\cite{ABD+09} gave an $O(n)$-time algorithm for recognizing  weakly simple $n$-gons
in the special case where all edges are collinear (in the context of flat foldability of a polygonal linkage). They defined the ws-equivalent \crimpreduction\/ operation. A \emph{crimp} is a chain of three consecutive collinear edges $[a,b,c,d]$ such that both the first edge $[a,b]$ and the last edge $[c,d]$ contain the middle edge $[b,c]$ (the containment need not be proper).
The operation \crimpreduction$(a,b,c,d)$ replaces the crimp $[a,b,c,d]$ with edge $[a,d]$; see Figure~\ref{fig:crimp}.

\begin{figure}[htbp]
\centering
	\def\svgwidth{.8\textwidth}
\begingroup%
  \makeatletter%
  \providecommand\color[2][]{%
    \errmessage{(Inkscape) Color is used for the text in Inkscape, but the package 'color.sty' is not loaded}%
    \renewcommand\color[2][]{}%
  }%
  \providecommand\transparent[1]{%
    \errmessage{(Inkscape) Transparency is used (non-zero) for the text in Inkscape, but the package 'transparent.sty' is not loaded}%
    \renewcommand\transparent[1]{}%
  }%
  \providecommand\rotatebox[2]{#2}%
  \ifx\svgwidth\undefined%
    \setlength{\unitlength}{326.48100586bp}%
    \ifx\svgscale\undefined%
      \relax%
    \else%
      \setlength{\unitlength}{\unitlength * \real{\svgscale}}%
    \fi%
  \else%
    \setlength{\unitlength}{\svgwidth}%
  \fi%
  \global\let\svgwidth\undefined%
  \global\let\svgscale\undefined%
  \makeatother%
  \begin{picture}(1,0.13956355)%
    \put(0,0){\includegraphics[width=\unitlength,page=1]{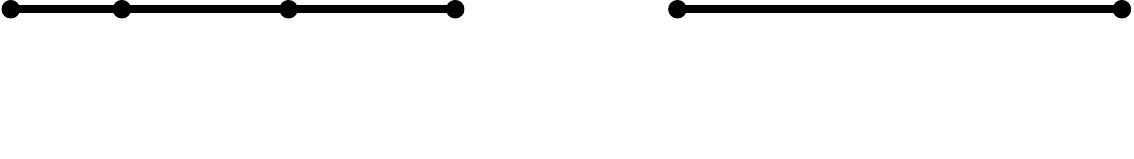}}%
    \put(0.49064619,0.1187872){\color[rgb]{0,0,0}\makebox(0,0)[lb]{\smash{$\Rightarrow$}}}%
    \put(0,0){\includegraphics[width=\unitlength,page=2]{fig-crimp-v2.pdf}}%
    \put(0.48889,0.03197444){\color[rgb]{0,0,0}\makebox(0,0)[lb]{\smash{$\Rightarrow$}}}%
    \put(0.00269418,0.07893944){\color[rgb]{0,0,0}\makebox(0,0)[lb]{\smash{$a$}}}%
    \put(0.10220535,0.07893944){\color[rgb]{0,0,0}\makebox(0,0)[lb]{\smash{$c$}}}%
    \put(0.2505992,0.07893944){\color[rgb]{0,0,0}\makebox(0,0)[lb]{\smash{$b$}}}%
    \put(0.3958506,0.07893944){\color[rgb]{0,0,0}\makebox(0,0)[lb]{\smash{$d$}}}%
    \put(0.59030389,0.07893944){\color[rgb]{0,0,0}\makebox(0,0)[lb]{\smash{$a$}}}%
    \put(0.98346023,0.07893944){\color[rgb]{0,0,0}\makebox(0,0)[lb]{\smash{$d$}}}%
  \end{picture}%
\endgroup%
\caption{\small A crimp reduction replaces
$[a,b,c,d]$ with $[a,d]$.
Top: image graph. Bottom: polygon.
\label{fig:crimp}}
\end{figure}

\begin{lemma}\label{lem:crimp}
The \crimpreduction\/ operation is ws-equivalent.
\end{lemma}
\begin{proof}
Let $P_1$ and $P_2$ be two polygons such that $P_2$ is obtained from $P_1$ by the operation \crimpreduction$(a,b,c,d)$. Without loss of generality, assume that $ad$ is horizontal with $a$ on the left and $d$ on the right.

First assume that $P_1$ is weakly simple. Then there exists a simple polygon $Q_1\in \Phi(P_1)$.
We modify $Q_1$ to obtain a simple polygon $Q_2\in \Phi(P_2)$. Without loss of generality, assume that edge $[a,b]$ is above $[b,c]$ (consequently, $[c,d]$ is below $[b,c]$) in $Q_1$. The modification involves the perfect matchings at the disks $D_b$ and $D_c$, and all disks and corridors along the line segment $bc$. Denote by $W_{top}$ the set of maximal paths that lie in the convex hull of $D_b\cup D_c$, below $[a,b]$ and above $[b,c]$; similarly, let $W_{bot}$ be the set of maximal paths that lie in the convex hull of $D_b\cup D_c$, below $[b,c]$ and above $[c,d]$. We proceed in two steps; refer to Figure~\ref{fig:crimp1}.
First, replace the path $[a,b,c,d]$ with the path $[a,c,b,d]$ such that the new edge $[a,c]$ replaces the old $[a,b]$ in the edge ordering of segment $ac$, the new $[c,b]$ replaces $[b,c]$ in the segments contained in $bc$, and finally the new $[b,d]$ replaces $[c,d]$ in $bd$.
Second, exchange $W_{top}$ and $W_{bot}$ such that the top-to-bottom order within each set of paths remains the same. Since the top-to-bottom order within $W_{top}$ and $W_{bot}$ is preserved, and the paths in $W_{top}$ (resp., $W_{bot}$) lie below (resp., above) the new path $[a,c,b,d]$, no edge crossings have been introduced. We obtain a simple polygon $Q_2\in \Phi(P_2)$, which shows that $P_2$ is weakly simple.

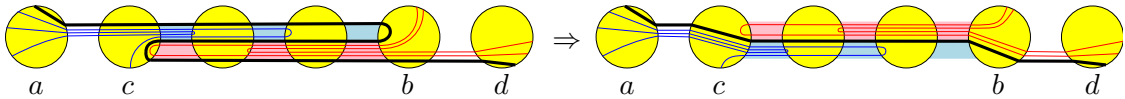
\begin{figure}[htb]
\centering
	\def\svgwidth{.9\textwidth}
\begingroup%
  \makeatletter%
  \providecommand\color[2][]{%
    \errmessage{(Inkscape) Color is used for the text in Inkscape, but the package 'color.sty' is not loaded}%
    \renewcommand\color[2][]{}%
  }%
  \providecommand\transparent[1]{%
    \errmessage{(Inkscape) Transparency is used (non-zero) for the text in Inkscape, but the package 'transparent.sty' is not loaded}%
    \renewcommand\transparent[1]{}%
  }%
  \providecommand\rotatebox[2]{#2}%
  \ifx\svgwidth\undefined%
    \setlength{\unitlength}{576.4bp}%
    \ifx\svgscale\undefined%
      \relax%
    \else%
      \setlength{\unitlength}{\unitlength * \real{\svgscale}}%
    \fi%
  \else%
    \setlength{\unitlength}{\svgwidth}%
  \fi%
  \global\let\svgwidth\undefined%
  \global\let\svgscale\undefined%
  \makeatother%
  \begin{picture}(1,0.07730382)%
    \put(0.10417988,-0.00233975){\color[rgb]{0,0,0}\makebox(0,0)[lb]{\smash{$c$}}}%
    \put(0,0){\includegraphics[width=\unitlength,page=1]{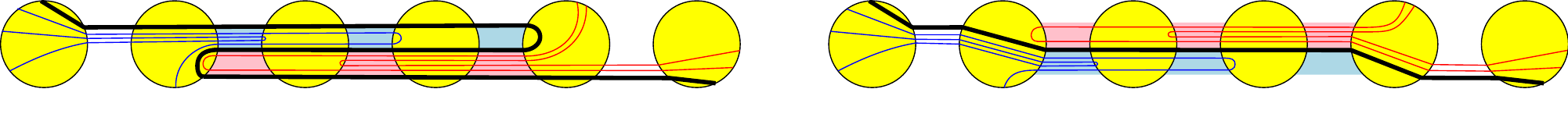}}%
    \put(0.02093038,-0.00233975){\color[rgb]{0,0,0}\makebox(0,0)[lb]{\smash{$a$}}}%
    \put(0.3541853,-0.00233975){\color[rgb]{0,0,0}\makebox(0,0)[lb]{\smash{$b$}}}%
    \put(0.43717783,-0.00233975){\color[rgb]{0,0,0}\makebox(0,0)[lb]{\smash{$d$}}}%
    \put(0.63296835,-0.00233975){\color[rgb]{0,0,0}\makebox(0,0)[lb]{\smash{$c$}}}%
    \put(0.54971884,-0.00233975){\color[rgb]{0,0,0}\makebox(0,0)[lb]{\smash{$a$}}}%
    \put(0.88297376,-0.00233975){\color[rgb]{0,0,0}\makebox(0,0)[lb]{\smash{$b$}}}%
    \put(0.96596627,-0.00233975){\color[rgb]{0,0,0}\makebox(0,0)[lb]{\smash{$d$}}}%
    \put(0.48970526,0.03936875){\color[rgb]{0,0,0}\makebox(0,0)[lb]{\smash{$\Rightarrow$}}}%
  \end{picture}%
\endgroup%
\caption{\small The operation \crimpreduction\ replaces a crimp $[a,b,c,d]$ with an edge $[ad]$.
\label{fig:crimp1}}
\end{figure}

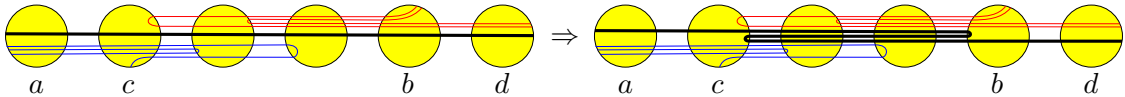
\begin{figure}[htb]
\centering
	\def\svgwidth{.9\textwidth}
\begingroup%
  \makeatletter%
  \providecommand\color[2][]{%
    \errmessage{(Inkscape) Color is used for the text in Inkscape, but the package 'color.sty' is not loaded}%
    \renewcommand\color[2][]{}%
  }%
  \providecommand\transparent[1]{%
    \errmessage{(Inkscape) Transparency is used (non-zero) for the text in Inkscape, but the package 'transparent.sty' is not loaded}%
    \renewcommand\transparent[1]{}%
  }%
  \providecommand\rotatebox[2]{#2}%
  \ifx\svgwidth\undefined%
    \setlength{\unitlength}{575.77016602bp}%
    \ifx\svgscale\undefined%
      \relax%
    \else%
      \setlength{\unitlength}{\unitlength * \real{\svgscale}}%
    \fi%
  \else%
    \setlength{\unitlength}{\svgwidth}%
  \fi%
  \global\let\svgwidth\undefined%
  \global\let\svgscale\undefined%
  \makeatother%
  \begin{picture}(1,0.07925906)%
    \put(0,0){\includegraphics[width=\unitlength,page=1]{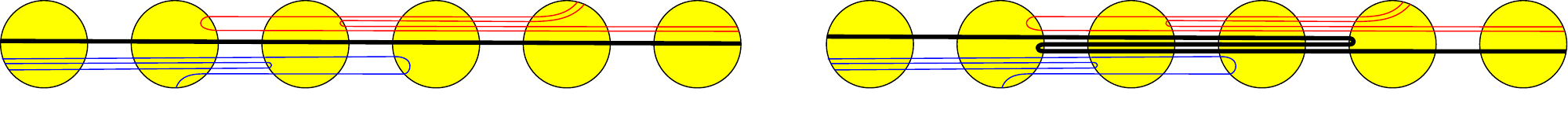}}%
    \put(0.1044071,0.00049815){\color[rgb]{0,0,0}\makebox(0,0)[lb]{\smash{$c$}}}%
    \put(0.02106655,0.00049815){\color[rgb]{0,0,0}\makebox(0,0)[lb]{\smash{$a$}}}%
    \put(0.35468598,0.00049815){\color[rgb]{0,0,0}\makebox(0,0)[lb]{\smash{$b$}}}%
    \put(0.43776928,0.00049815){\color[rgb]{0,0,0}\makebox(0,0)[lb]{\smash{$d$}}}%
    \put(0.63160524,0.00049815){\color[rgb]{0,0,0}\makebox(0,0)[lb]{\smash{$c$}}}%
    \put(0.54826467,0.00049815){\color[rgb]{0,0,0}\makebox(0,0)[lb]{\smash{$a$}}}%
    \put(0.88188408,0.00049815){\color[rgb]{0,0,0}\makebox(0,0)[lb]{\smash{$b$}}}%
    \put(0.96496738,0.00049815){\color[rgb]{0,0,0}\makebox(0,0)[lb]{\smash{$d$}}}%
    \put(0.48796277,0.04370109){\color[rgb]{0,0,0}\makebox(0,0)[lb]{\smash{$\Rightarrow$}}}%
  \end{picture}%
\endgroup%
\caption{\small The reversal of \crimpreduction\ replaces edge $[ad]$ with a crimp $[a,b,c,d]$.
\label{fig:crimp2}}
\end{figure}

Next assume that $P_2$ is weakly simple. Then, there exists a simple polygon $Q_2\in \Phi(P_2)$.
We modify $Q_2$ to obtain a simple polygon $Q_1\in \Phi(P_1)$; refer to Figure~\ref{fig:crimp2}.
Replace edge $[a,d]$ by $[a,b,c,d]$ also replacing $[a,d]$ in the ordering of the affected segments by $[c,d]$, $[b,c]$, and $[a,b]$, in this order.
The new ordering produces a polygon $Q_1$ in the strip system of $P$.
Because $Q_2$ is simple, by construction the new matchings do not interact with the preexisting edges in the disks. Hence, $Q_1\in \Phi(P_1)$, which shows that $P_1$ is weakly simple.
\end{proof}

Given a chain of two edges $[a,b,c]$ such that $[a,b]$ and $[b,c]$ are collinear but do not overlap, the {\sf merge} operation replaces $[a,b,c]$ with a single edge $[a,c]$. The merge operation (as well as its inverse, {\sf subdivision}) is ws-equivalent by the definition of weak simplicity in terms of  Fr\'echet distance~\cite{CEX15}.
If we greedily apply \crimpreduction\ and {\sf merge} operations,
in linear time we obtain a polygon with the following two properties:

\begin{enumerate}[label=(A\arabic*)]\itemsep -1pt
\item \label{cond:A1} Every two consecutive collinear edges overlap (i.e., form a spur).
\item \label{cond:A2} No three consecutive collinear edges form a crimp.
\end{enumerate}

Assuming properties \ref{cond:A1} and \ref{cond:A2}, we can characterize a chain of collinear edges with the sequence of their edge lengths.

\begin{lemma}\label{lem:irreducible}
Let $C = [e_i,\ldots ,e_{k}]$ be a chain of collinear edges in a polygon with properties \ref{cond:A1} and \ref{cond:A2}.
Then the sequence of edge lengths $(|e_i|,\ldots , |e_{k}|)$ is unimodal (all local maxima are consecutive); and no two consecutive edges have the same length, except possibly the maximal edge length that can occur at most twice.
\end{lemma}
\begin{proof}
For every $j$ such that $i<j<k$, consider $|e_j|$. If $|e_{j{-}1}|$ and  $|e_{j+1}|$ are at least as large as $|e_j|$,
then the three edges form a crimp, by \ref{cond:A1}. However, this contradicts \ref{cond:A2}.
This proves unimodality, and that no three consecutive edges can have the same length.
In fact if $|e_j|$ is not maximal, one neighbor must be strictly smaller, to avoid the same contradiction.
\end{proof}

The operations introduced in Section~\ref{sec:bars} maintain properties \ref{cond:A1}--\ref{cond:A2} for all maximal paths inside an elliptical disk $D_b$.

\subsection{Node expansion}
Compute the bar decomposition of $P$ and its image graph $G$ (defined in Section~\ref{sec:preliminaries}, see Figure~\ref{fig:bars}).
For every sober node of the image graph, we perform the ws-equivalent \nodeexpansion\ operation, described in~\cite[Section~3]{CEX15} (Cortese et al.~\cite{CDP+09} call this a \emph{cluster expansion}).
Let $u$ be a sober node of the image graph. Let $D_u$ be the disk centered at $u$ with radius $\delta>0$  sufficiently small so that $D_u$ intersects only the segments incident to $u$. For each segment $ux$ incident to $u$, create a new node $u^x$ at the intersection point $ux\cap \partial D_u$. Then modify $P$ by replacing each subpath $[x,u,y]$ passing through $u$ by $[x,u^x,u^y,y]$; see Figure~\ref{fig:node-expansion}.
If a node expansion produces an edge crossing, report that $P$ is not weakly simple.
\begin{figure}[htb]
\centering
	\def\svgwidth{.75\textwidth}
\begingroup%
  \makeatletter%
  \providecommand\color[2][]{%
    \errmessage{(Inkscape) Color is used for the text in Inkscape, but the package 'color.sty' is not loaded}%
    \renewcommand\color[2][]{}%
  }%
  \providecommand\transparent[1]{%
    \errmessage{(Inkscape) Transparency is used (non-zero) for the text in Inkscape, but the package 'transparent.sty' is not loaded}%
    \renewcommand\transparent[1]{}%
  }%
  \providecommand\rotatebox[2]{#2}%
  \ifx\svgwidth\undefined%
    \setlength{\unitlength}{269.43876953bp}%
    \ifx\svgscale\undefined%
      \relax%
    \else%
      \setlength{\unitlength}{\unitlength * \real{\svgscale}}%
    \fi%
  \else%
    \setlength{\unitlength}{\svgwidth}%
  \fi%
  \global\let\svgwidth\undefined%
  \global\let\svgscale\undefined%
  \makeatother%
  \begin{picture}(1,0.17891604)%
    \put(0,0){\includegraphics[width=\unitlength,page=1]{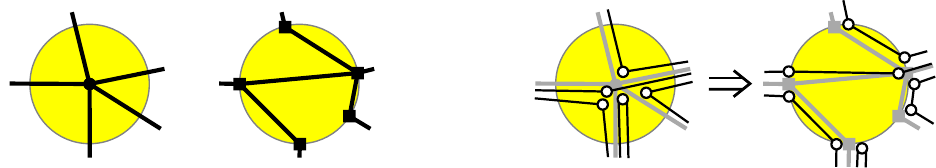}}%
    \put(-0.00254991,0.12631645){\color[rgb]{0,0,0}\makebox(0,0)[lb]{\smash{$D_u$}}}%
    \put(0,0){\includegraphics[width=\unitlength,page=2]{fig-nodeexpansion.pdf}}%
  \end{picture}%
\endgroup%
\caption{Node expansion. (Left) Changes in the image graph. (Right) Changes in $P$ (the vertices are perturbed for clarity).
New nodes are shown as squares.}
\label{fig:node-expansion}
\end{figure}

\subsection{Bar expansion}
\label{ssec:bar-expansion}
Chang et al.~\cite[Section 4]{CEX15} define a bar expansion operation.
In this paper, we refer to it as \oldbarexp.
For a bar $b$ of the image graph, draw a long and narrow ellipse $D_b$ around the interior nodes of $b$,
create subdivision vertices at the intersection of $\partial D_b$ with the edges, and replace each maximal path in $D_b$ by a straight-line edge. If $b$ contains no spurs, \oldbarexp\ is known to be ws-equivalent~\cite{CEX15}.
Otherwise, it can produce false positives, hence it is not ws-equivalent; see Figure~\ref{fig:normal-expansion} for an example.
\begin{figure}[htbp]
\centering
	\def\svgwidth{\textwidth}
\begingroup%
  \makeatletter%
  \providecommand\color[2][]{%
    \errmessage{(Inkscape) Color is used for the text in Inkscape, but the package 'color.sty' is not loaded}%
    \renewcommand\color[2][]{}%
  }%
  \providecommand\transparent[1]{%
    \errmessage{(Inkscape) Transparency is used (non-zero) for the text in Inkscape, but the package 'transparent.sty' is not loaded}%
    \renewcommand\transparent[1]{}%
  }%
  \providecommand\rotatebox[2]{#2}%
  \ifx\svgwidth\undefined%
    \setlength{\unitlength}{357.16848145bp}%
    \ifx\svgscale\undefined%
      \relax%
    \else%
      \setlength{\unitlength}{\unitlength * \real{\svgscale}}%
    \fi%
  \else%
    \setlength{\unitlength}{\svgwidth}%
  \fi%
  \global\let\svgwidth\undefined%
  \global\let\svgscale\undefined%
  \makeatother%
  \begin{picture}(1,0.16755263)%
    \put(0,0){\includegraphics[width=\unitlength,page=1]{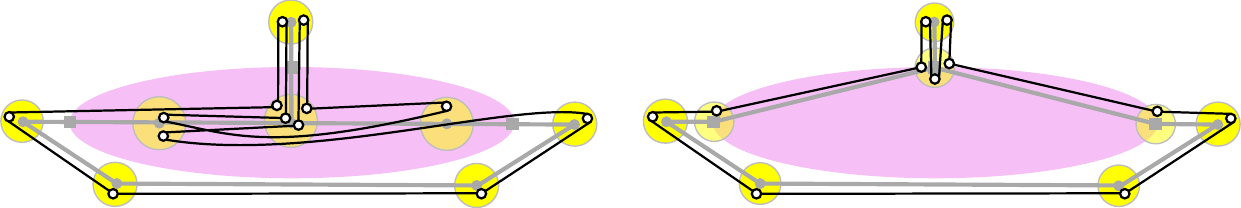}}%
    \put(0.09165214,0.10714039){\color[rgb]{0,0,0}\makebox(0,0)[lb]{\smash{$D_b$}}}%
    \put(0.49213378,0.02960972){\color[rgb]{0,0,0}\makebox(0,0)[lb]{\smash{$\Rightarrow$}}}%
  \end{picture}%
\endgroup%
\caption{
The \oldbarexp\ converts a non-weakly simple polygon to a weakly simple one.
}
\label{fig:normal-expansion}
\end{figure}

\medskip\noindent{\bf New bar expansion operation.}
Let $b$ be a bar in the image graph with at least one interior node; see Figure~\ref{fig:bar-expansion}.
Without loss of generality, assume that $b$ is horizontal.
Let $D_b$ be an ellipse whose major axis is in $b$ such that $D_b$ contains all interior nodes of $b$ (nodes in $b$ except its endpoints), but does not contain any other node of the image graph and does not intersect any segment that is not incident to some node inside $D_b$.

Similar to \oldbarexp, the operation \newbarexp\ introduces subdivision vertices on $\partial D_b$, however we keep all interior vertices of a bar at their original positions.
{In Section~\ref{sec:bars}, we apply a sequence of new operations to eliminate all vertices on $b$ sequentially while creating new nodes in the vicinity of $D_b$.
Our bar expansion operation can be considered as a preprocessing step for this subroutine.}

For each segment $ux$ between a node $u\in b\cap D_b$ and a node $x\not\in b$, create a new node $u^x$ at the intersection point $ux\cap \partial D_b$ and subdivide every edge $[u,x]$  to a path $[u,u^x,x]$. For each endpoint $v$ of $b$, create two new nodes, $v'$ and $v''$, as follows.
Node $v$ is adjacent to a unique segment $vw\subset b$, where $w\in b\cap D_b$.
Create a new node $v'\in \partial D_b$ sufficiently close to the intersection point $vw\cap \partial D_b$, but strictly above $b$; and create a new node $v''$ in the interior of segment $vw\cap D_b$. Subdivide every edge $[v,y]$, where $y\in b$, into a path $[v,v',v'',y]$.
Since the \newbarexp\ operation consists of only subdivisions (and slight perturbations of the edges passing through the end-segments of the bars), it is ws-equivalent.

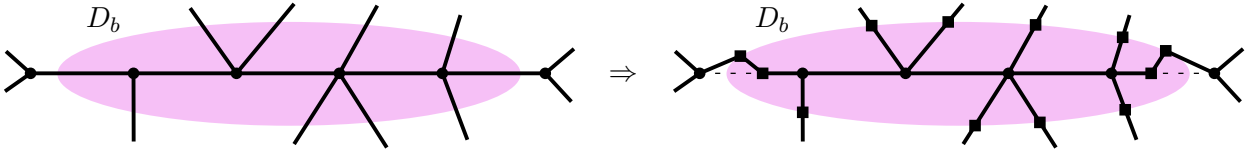
\begin{figure}[htb]
\centering
	\def\svgwidth{\textwidth}
\begingroup%
  \makeatletter%
  \providecommand\color[2][]{%
    \errmessage{(Inkscape) Color is used for the text in Inkscape, but the package 'color.sty' is not loaded}%
    \renewcommand\color[2][]{}%
  }%
  \providecommand\transparent[1]{%
    \errmessage{(Inkscape) Transparency is used (non-zero) for the text in Inkscape, but the package 'transparent.sty' is not loaded}%
    \renewcommand\transparent[1]{}%
  }%
  \providecommand\rotatebox[2]{#2}%
  \ifx\svgwidth\undefined%
    \setlength{\unitlength}{386.07536621bp}%
    \ifx\svgscale\undefined%
      \relax%
    \else%
      \setlength{\unitlength}{\unitlength * \real{\svgscale}}%
    \fi%
  \else%
    \setlength{\unitlength}{\svgwidth}%
  \fi%
  \global\let\svgwidth\undefined%
  \global\let\svgscale\undefined%
  \makeatother%
  \begin{picture}(1,0.11891201)%
    \put(0,0){\includegraphics[width=\unitlength,page=1]{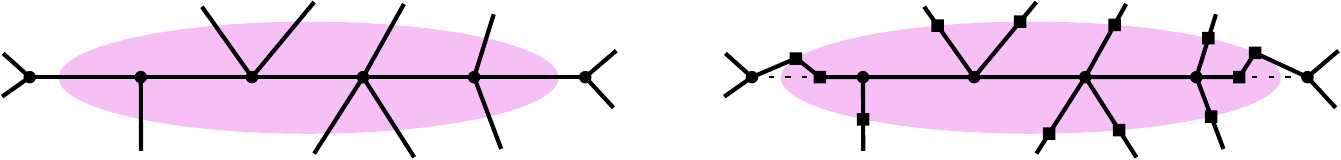}}%
    \put(0.06695963,0.09834062){\color[rgb]{0,0,0}\makebox(0,0)[lb]{\smash{$D_b$}}}%
    \put(0.60527625,0.09834062){\color[rgb]{0,0,0}\makebox(0,0)[lb]{\smash{$D_b$}}}%
    \put(0.48784422,0.05516142){\color[rgb]{0,0,0}\makebox(0,0)[lb]{\smash{$\Rightarrow$}}}%
  \end{picture}%
\endgroup%
\caption{The changes in the image graph caused by \newbarexp.}
\label{fig:bar-expansion}
\end{figure}

\medskip\noindent{\bf Crossing paths.}
Apart from \nodeexpansion\ and \oldbarexp, none of our operations creates edge crossings. In some cases, our bar simplification algorithm (Section~\ref{sec:bars}) detects whether two subpaths cross. Crossings between overlapping paths are not easy to identify (see~\cite[Section~2]{CEX15} for a discussion).
We rely on the following simple condition to detect some (but not all) crossings.

\begin{lemma}\label{lem:cross}
Let $P$ be a weakly simple polygon parameterized by a curve $\gamma_1:\mathbb{S}^1\rightarrow \mathbb{R}^2$; and let $\gamma_2:\mathbb{S}^1\rightarrow \mathbb{R}^2$ be a closed Jordan curve that does not pass through any vertices of $P$ and intersects every edge of $P$ transversely. Suppose that $q_1, \ldots , q_4$ are distinct points in $\gamma_2(\mathbb{S}^1)$ in counterclockwise order. Then there are no two disjoint arcs $I_1,I_2\subset \mathbb{S}^1$ such that $\gamma_1(I_1)$ and $\gamma_1(I_2)$ connect $q_1$ to $q_3$ and $q_2$ to $q_4$, each passing through the interior of $\gamma_2(\mathbb{S}^1)$.
\end{lemma}
\begin{proof}
Suppose, to the contrary, that there exist two disjoint arcs $I_1,I_2\subset \mathbb{S}_1$ such that $\gamma_1(I_1)$ and $\gamma_1(I_2)$ respectively connect $q_1$ to $q_3$ and $q_2$ to $q_4$, passing through the interior of $\gamma_2(\mathbb{S}^1)$. (See Figure~\ref{fig:crossing}.)
Since $P$ is weakly simple, then $\gamma_1$ can be perturbed to a closed Jordan curve $\gamma_1'$ with the same properties as $\gamma_1$. Let $U$ denote the interior of $\gamma_2(\mathbb{S}^1)$, and note that $U$ is simply connected. Consequently, $U\setminus \gamma_1'(I_1)$ has two components, which are incident to $q_2$ and $q_4$, respectively. The Jordan arc $\gamma_1'(I_2)$ connects $q_2$ to $q_4$ via $U$, so it must intersect $\gamma_1'(I_1)$,
contradicting the assumption that $\gamma_1'$ is a Jordan curve.
\end{proof}

\begin{figure}[h!tbp]
\centering
	\def\svgwidth{.4\textwidth}
\begingroup%
  \makeatletter%
  \providecommand\color[2][]{%
    \errmessage{(Inkscape) Color is used for the text in Inkscape, but the package 'color.sty' is not loaded}%
    \renewcommand\color[2][]{}%
  }%
  \providecommand\transparent[1]{%
    \errmessage{(Inkscape) Transparency is used (non-zero) for the text in Inkscape, but the package 'transparent.sty' is not loaded}%
    \renewcommand\transparent[1]{}%
  }%
  \providecommand\rotatebox[2]{#2}%
  \ifx\svgwidth\undefined%
    \setlength{\unitlength}{181.96272482bp}%
    \ifx\svgscale\undefined%
      \relax%
    \else%
      \setlength{\unitlength}{\unitlength * \real{\svgscale}}%
    \fi%
  \else%
    \setlength{\unitlength}{\svgwidth}%
  \fi%
  \global\let\svgwidth\undefined%
  \global\let\svgscale\undefined%
  \makeatother%
  \begin{picture}(1,0.37815236)%
    \put(0,0){\includegraphics[width=\unitlength,page=1]{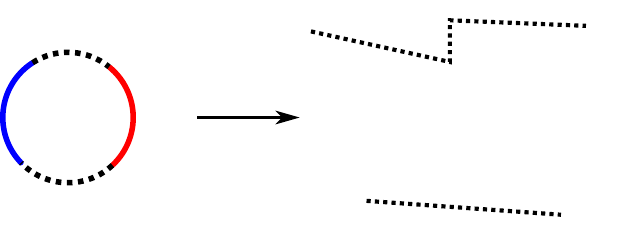}}%
    \put(0.38768361,0.20967619){\color[rgb]{0,0,0}\makebox(0,0)[b]{\smash{$\gamma_1'$}}}%
    \put(0,0){\includegraphics[width=\unitlength,page=2]{fig-crossing-lemma.pdf}}%
    \put(1.02535063,0.18769374){\color[rgb]{0,0,0}\makebox(0,0)[lb]{\smash{$\textcolor{darkgray}{\gamma_2}(\mathbb{S}^1)$}}}%
    \put(-0.0573594,0.49992056){\color[rgb]{0,0,0}\makebox(0,0)[lt]{\begin{minipage}{1.07714369\unitlength}\raggedright \end{minipage}}}%
    \put(0.10756951,0.00304046){\color[rgb]{0,0,0}\makebox(0,0)[b]{\smash{$\mathbb{S}^1$}}}%
    \put(-0.03751515,0.18769367){\color[rgb]{0,0,0}\makebox(0,0)[b]{\smash{$\textcolor{blue}{I_1}$}}}%
    \put(0.25474978,0.18769367){\color[rgb]{0,0,0}\makebox(0,0)[b]{\smash{$\textcolor{red}{I_2}$}}}%
    \put(0.50974706,0.25799521){\color[rgb]{0,0,0}\makebox(0,0)[b]{\smash{$q_1$}}}%
    \put(0.61965968,0.10411748){\color[rgb]{0,0,0}\makebox(0,0)[b]{\smash{$q_2$}}}%
    \put(0.94060454,0.07334194){\color[rgb]{0,0,0}\makebox(0,0)[b]{\smash{$q_3$}}}%
    \put(0.87026046,0.24920214){\color[rgb]{0,0,0}\makebox(0,0)[b]{\smash{$q_4$}}}%
    \put(0,0){\includegraphics[width=\unitlength,page=3]{fig-crossing-lemma.pdf}}%
  \end{picture}%
\endgroup%
\caption{Forbidden configuration described by Lemma~\ref{lem:cross}.}
\label{fig:crossing}
\end{figure}

We show that a weakly simple polygon cannot contain certain configurations, outlined below.
\begin{corollary}\label{cor:forbidden}
A weakly simple polygon cannot contain a pair of paths of the following types:
\begin{enumerate}\itemsep -1pt
\item\label{cross:1}
 $[u_1,u_2,u_3]$ and $[v,u_2,w]$, where $u_2u_1$, $u_2v$, $u_2u_3$, and $u_2w$ are nonoverlapping segments in this cyclic order around $u_2$ (\emph{node crossing}; see Figure~\ref{fig:forbidden}(a)).
\item\label{cross:3}
 $[u_1,u_3,w]$ and $[v,u_2,u_4]$, where $u_1$, $u_2$, $u_3$, and $u_4$ are on a line in this order, and nodes $v$ and $w$ lie in an open halfplane bounded by this line (Figure~\ref{fig:forbidden}(b)).
\item\label{cross:4}
 $[u_1,u_2,u_3]$ and $[v_1,v_2,\ldots ,v_{k-1},v_k]$ where $v_2\in {\rm int}(u_2u_3)$, $v_3,\ldots , v_{k-1}\in \{u_2\}\cup {\rm int}(u_2u_3)$, nodes $u_1$ and $v_1$ lie in an open halfplane bounded by the supporting line of $u_2u_3$, and node $v_k$ lies on the other open halfplane bounded by this line (Figure~\ref{fig:forbidden}(c)).
\end{enumerate}
\end{corollary}
\begin{proof}
In all four cases, Lemma~\ref{lem:cross} with a suitable Jordan curve $\gamma_2$ completes the proof.
In case~\ref{cross:1}, let $\gamma_2$ be a small circle around $u_2$.
In case~\ref{cross:3}, let $\gamma_2$ be a small neighborhood of segment $u_1u_2$.
In case~\ref{cross:4}, let $\gamma_2$ be a small neighborhood of the convex hull of $\{v_2,\ldots, v_{k-1}\}$.
\end{proof}

\begin{figure}[htbp]
\centering
	\def\svgwidth{.75\textwidth}
\begingroup%
  \makeatletter%
  \providecommand\color[2][]{%
    \errmessage{(Inkscape) Color is used for the text in Inkscape, but the package 'color.sty' is not loaded}%
    \renewcommand\color[2][]{}%
  }%
  \providecommand\transparent[1]{%
    \errmessage{(Inkscape) Transparency is used (non-zero) for the text in Inkscape, but the package 'transparent.sty' is not loaded}%
    \renewcommand\transparent[1]{}%
  }%
  \providecommand\rotatebox[2]{#2}%
  \ifx\svgwidth\undefined%
    \setlength{\unitlength}{396.46938477bp}%
    \ifx\svgscale\undefined%
      \relax%
    \else%
      \setlength{\unitlength}{\unitlength * \real{\svgscale}}%
    \fi%
  \else%
    \setlength{\unitlength}{\svgwidth}%
  \fi%
  \global\let\svgwidth\undefined%
  \global\let\svgscale\undefined%
  \makeatother%
  \begin{picture}(1,0.22404272)%
    \put(0,0){\includegraphics[width=\unitlength,page=1]{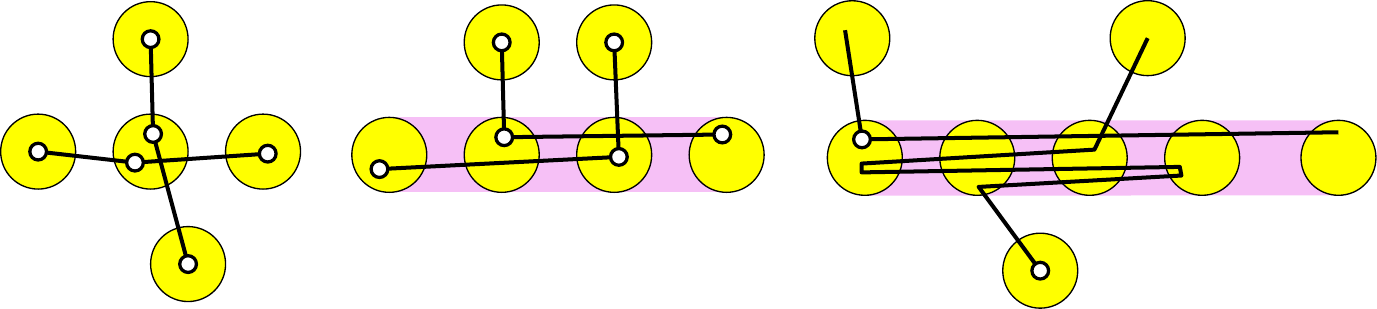}}%
    \put(0.00893843,0.00903658){\color[rgb]{0,0,0}\makebox(0,0)[lb]{\smash{(a)}}}%
    \put(0.25588036,0.00903658){\color[rgb]{0,0,0}\makebox(0,0)[lb]{\smash{(b)}}}%
    \put(0.60357926,0.00903658){\color[rgb]{0,0,0}\makebox(0,0)[lb]{\smash{(c)}}}%
    \put(0.05999997,0.18743589){\color[rgb]{0,0,0}\makebox(0,0)[lb]{\smash{$v$}}}%
    \put(0.01817817,0.06716575){\color[rgb]{0,0,0}\makebox(0,0)[lb]{\smash{$u_1$}}}%
    \put(0.0880336,0.06820995){\color[rgb]{0,0,0}\makebox(0,0)[lb]{\smash{$u_2$}}}%
    \put(0.18034562,0.06845119){\color[rgb]{0,0,0}\makebox(0,0)[lb]{\smash{$u_3$}}}%
    \put(0.08526715,0.02158787){\color[rgb]{0,0,0}\makebox(0,0)[lb]{\smash{$w$}}}%
    \put(0.31881364,0.18815718){\color[rgb]{0,0,0}\makebox(0,0)[lb]{\smash{$v$}}}%
    \put(0.48319893,0.18859399){\color[rgb]{0,0,0}\makebox(0,0)[lb]{\smash{$w$}}}%
    \put(0.26833465,0.0672297){\color[rgb]{0,0,0}\makebox(0,0)[lb]{\smash{$u_1$}}}%
    \put(0.35506867,0.06745696){\color[rgb]{0,0,0}\makebox(0,0)[lb]{\smash{$u_2$}}}%
    \put(0.43529568,0.06464483){\color[rgb]{0,0,0}\makebox(0,0)[lb]{\smash{$u_3$}}}%
    \put(0.51654191,0.06539354){\color[rgb]{0,0,0}\makebox(0,0)[lb]{\smash{$u_4$}}}%
    \put(0.65282592,0.19470929){\color[rgb]{0,0,0}\makebox(0,0)[lb]{\smash{$u_1$}}}%
    \put(0.77695415,0.19382583){\color[rgb]{0,0,0}\makebox(0,0)[lb]{\smash{$v_1$}}}%
    \put(0.58681215,0.06318452){\color[rgb]{0,0,0}\makebox(0,0)[lb]{\smash{$v_3=u_2$}}}%
    \put(0.69667125,0.14836691){\color[rgb]{0,0,0}\makebox(0,0)[lb]{\smash{$v_5$}}}%
    \put(0.69933222,0.01762016){\color[rgb]{0,0,0}\makebox(0,0)[lb]{\smash{$v_6$}}}%
    \put(0.86758966,0.06422872){\color[rgb]{0,0,0}\makebox(0,0)[lb]{\smash{$v_4$}}}%
    \put(0,0){\includegraphics[width=\unitlength,page=2]{fig-forbidden+v2.pdf}}%
    \put(0.96250519,0.06395821){\color[rgb]{0,0,0}\makebox(0,0)[lb]{\smash{$u_3$}}}%
    \put(0.781779,0.0648617){\color[rgb]{0,0,0}\makebox(0,0)[lb]{\smash{$v_2$}}}%
  \end{picture}%
\endgroup%
\caption{Three pairs of incompatible paths.}
\label{fig:forbidden}
\end{figure}

\medskip\noindent{\bf Terminology.}
We classify the maximal paths in $D_b$.
All nodes $u\in \partial D_b$ lie either above or below  $b$. We call them \emph{top} and \emph{bottom} nodes, respectively. Let $\PP$ denote the set of maximal paths $p=[u_1^x,u_1,\ldots ,u_{k},u_k^y]$ in $D_b$. The paths in $\PP$ are classified based on the position of their endpoints.
A path $p$ can be labeled as follows:
\begin{itemize}\itemsep -2pt
\item \emph{cross-chain} if $u_1^x$ and $u_k^y$ are top and bottom nodes respectively,
\item \emph{top chain} (resp., \emph{bottom chain}) if both $u_1^x$ and $u_k^y$ are top nodes (resp., bottom nodes),
\item \emph{pin} if $p=[u_1^x,u_1,u_1^x]$ (note that every pin is a top or a bottom chain),
\item \emph{V-chain} if $p=[u_1^x,u_1,u_1^y]$, where $x\neq y$ and $p$ is a top or a bottom chain.
\end{itemize}
Finally, let $\Pin\subset \PP$ be the set of pins, and $\VV\subset \PP$ the set of V-chains.

\subsection{Clusters}
\label{ssec:clusters}
As a preprocessing step for spur elimination (Section~\ref{sec:tree-exp}), we group all nodes that do not lie inside a bar into \emph{clusters}.
After \nodeexpansion\ and \newbarexp, all such nodes lie on a boundary of a disk (circular or elliptical).
For every sober node $u$, we create $\deg(u)$ clusters as follows. Refer to Figure~\ref{fig:newcluster}.
The node expansion has replaced $u$ with new nodes on $\partial D_u$. Subdivide each segment in $D_u$ with two new nodes.
For each node $v\in \partial D_u$, form a cluster $C(v)$ that consists of $v$ and all adjacent (subdivision) nodes inside $D_u$.
For each node $u$ on the boundary of an elliptical disk $D_b$, subdivide the unique edge outside $D_b$ incident to $u$ with a node $u^*$. Form a cluster $C(u^*)$ containing $u$ and $u^*$.
Every cluster maintains the following invariants.
\begin{enumerate}[label=(I\arabic*)]\itemsep -2pt
\item[] {\bf
\noindent Cluster Invariants.} For every cluster $C(u)$:
\item \label{inv:tree} $C(u)$ induces a tree $T[u]$ in the image graph rooted at $u$.
\item \label{inv:max-path}Every maximal path of $P$ in $C(u)$ is of one of the following two types:
        \begin{enumerate}\itemsep -2pt
        \item\label{inv:mp-a} both endpoints are at the root of $T[u]$ and the path contains a single spur;
        \item\label{inv:mp-b} one endpoint is at the root,
         the other is at a leaf, and the path contains no spurs.
        \end{enumerate}
%
%
%
\item \label{inv:deg} Every leaf node $\ell$ satisfies one of the following conditions:
        \begin{enumerate}\itemsep -2pt
        \item\label{inv:deg-a} $\ell$ has degree one in the image graph of $P$ (and every vertex at $\ell$ is a spur);
        \item\label{inv:deg-b}\label{inv:no-spur} $\ell$ has degree two in the image graph of $P$
        and there is no spur at $\ell$.
        \end{enumerate}
\item \label{inv:subdiv} No edge passes through a leaf $\ell$ (i.e., there is no edge $[a,b]$ such that $\ell\in ab$ but $\ell\not \in \{a,b\}$).
\end{enumerate}

Initially, every cluster trivially satisfies \ref{inv:tree}--\ref{inv:max-path} and every leaf node satisfies \ref{inv:deg}--\ref{inv:subdiv} since it was created by a subdivision.

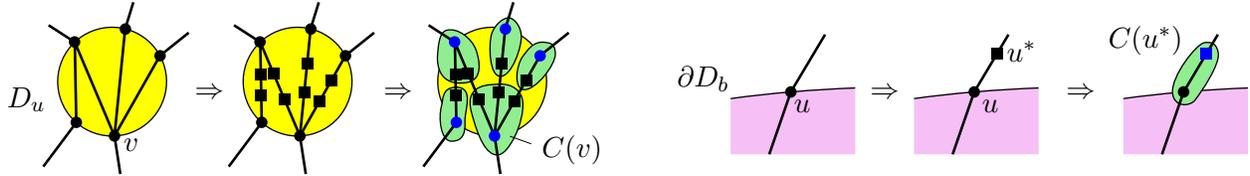
\begin{figure}[htbp]
\centering
	\def\svgwidth{\textwidth}
\begingroup%
  \makeatletter%
  \providecommand\color[2][]{%
    \errmessage{(Inkscape) Color is used for the text in Inkscape, but the package 'color.sty' is not loaded}%
    \renewcommand\color[2][]{}%
  }%
  \providecommand\transparent[1]{%
    \errmessage{(Inkscape) Transparency is used (non-zero) for the text in Inkscape, but the package 'transparent.sty' is not loaded}%
    \renewcommand\transparent[1]{}%
  }%
  \providecommand\rotatebox[2]{#2}%
  \ifx\svgwidth\undefined%
    \setlength{\unitlength}{380.09963379bp}%
    \ifx\svgscale\undefined%
      \relax%
    \else%
      \setlength{\unitlength}{\unitlength * \real{\svgscale}}%
    \fi%
  \else%
    \setlength{\unitlength}{\svgwidth}%
  \fi%
  \global\let\svgwidth\undefined%
  \global\let\svgscale\undefined%
  \makeatother%
  \begin{picture}(1,0.13888302)%
    \put(0,0){\includegraphics[width=\unitlength,page=1]{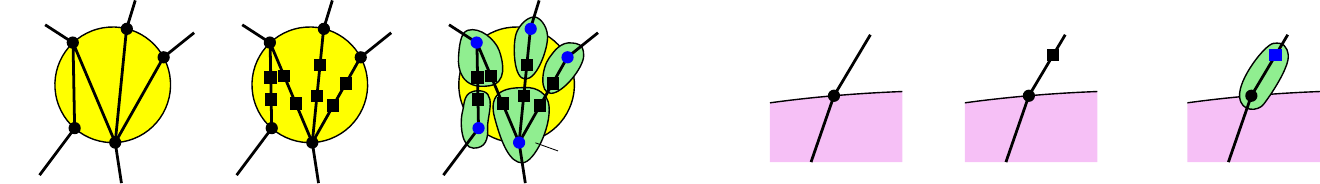}}%
    \put(0.00097235,0.05385054){\color[rgb]{0,0,0}\makebox(0,0)[lb]{\smash{$D_u$}}}%
    \put(0.09470191,0.01860652){\color[rgb]{0,0,0}\makebox(0,0)[lb]{\smash{$v$}}}%
    \put(0.43157292,0.0142823){\color[rgb]{0,0,0}\makebox(0,0)[lb]{\smash{$C(v)$}}}%
    \put(0.54040151,0.07053551){\color[rgb]{0,0,0}\makebox(0,0)[lb]{\smash{$\partial D_b$}}}%
    \put(0.63452176,0.05032638){\color[rgb]{0,0,0}\makebox(0,0)[lb]{\smash{$u$}}}%
    \put(0.78581507,0.04986723){\color[rgb]{0,0,0}\makebox(0,0)[lb]{\smash{$u$}}}%
    \put(0.80595653,0.09119164){\color[rgb]{0,0,0}\makebox(0,0)[lb]{\smash{$u^*$}}}%
    \put(0.88759879,0.10232425){\color[rgb]{0,0,0}\makebox(0,0)[lb]{\smash{$C(u^*)$}}}%
    \put(0.15217932,0.06076185){\color[rgb]{0,0,0}\makebox(0,0)[lb]{\smash{$\Rightarrow$}}}%
    \put(0.30412737,0.06076185){\color[rgb]{0,0,0}\makebox(0,0)[lb]{\smash{$\Rightarrow$}}}%
    \put(0.69601052,0.06076185){\color[rgb]{0,0,0}\makebox(0,0)[lb]{\smash{$\Rightarrow$}}}%
    \put(0.85380269,0.06076185){\color[rgb]{0,0,0}\makebox(0,0)[lb]{\smash{$\Rightarrow$}}}%
  \end{picture}%
\endgroup%
\caption{Formation of new clusters around (left) a sober node and (right) a node on the boundary of an elliptical disk. The roots of the induced trees are colored blue.}
\label{fig:newcluster}
\end{figure}
\medskip\noindent{\bf Dummy vertices.}
Although the operations described in Sections \ref{sec:bars} and \ref{sec:tree-exp} introduce new nodes in the clusters, the image graph will always have $O(n)$ nodes and segments.
A vertex at a cluster node is called a \emph{benchmark} if it is a spur or if it is at a leaf node;
otherwise it is called a \emph{dummy vertex}.
Paths traversing clusters may jointly contain $\Theta(n^2)$ dummy vertices in the worst case, however
we do not store these explicitly.
By \ref{inv:tree}, \ref{inv:max-path}, and \ref{inv:deg} a maximal path in a cluster can be uniquely encoded by one benchmark vertex:
if it goes from a root to a spur at an interior node $s$ and back, we record only $[s]$;
and if it traverses $T[u]$ from the root to a leaf $\ell$, we record only $[\ell]$.

\section{Bar simplification}
\label{sec:bars}

In this section we introduce three new ws-equivalent operations and show that they can eliminate all vertices from each bar independently (thus eliminating all forks).
The bar decomposition is pre-computed, and the bars remain fixed
during this phase (even though all edges along each bar are  eliminated).

We give an overview of the overall effect of the operations (Section~\ref{ssec:overview}),
define them and show that they are ws-equivalent (Sections~\ref{ssec:primitives}--\ref{ssec:operations}), and then show how to use these operations to eliminate all vertices from a bar (Section~\ref{ssec:phases}).

\subsection{Overview}
\label{ssec:overview}

After preprocessing in Section~\ref{sec:preprocess}, we may assume that $P$ has no edge crossings and satisfies \ref{cond:A1}--\ref{cond:A2}.
%
We summarize the overall effect of the bar simplification subroutine for a given expanded bar.

\medskip
\noindent{\bf Changes in the image graph {\em G}.} Refer to Figure~\ref{fig:overview-G}.
All nodes in the interior of the ellipse $D_b$ are eliminated.
Some spurs on $b$  are moved to new nodes in the clusters along  $\partial D_b$.
Segments inside $D_b$ connect two leaves of trees induced by clusters.

\begin{figure}[h!tb]
\centering
	\def\svgwidth{.95\textwidth}
\begingroup%
  \makeatletter%
  \providecommand\color[2][]{%
    \errmessage{(Inkscape) Color is used for the text in Inkscape, but the package 'color.sty' is not loaded}%
    \renewcommand\color[2][]{}%
  }%
  \providecommand\transparent[1]{%
    \errmessage{(Inkscape) Transparency is used (non-zero) for the text in Inkscape, but the package 'transparent.sty' is not loaded}%
    \renewcommand\transparent[1]{}%
  }%
  \providecommand\rotatebox[2]{#2}%
  \ifx\svgwidth\undefined%
    \setlength{\unitlength}{401.03430176bp}%
    \ifx\svgscale\undefined%
      \relax%
    \else%
      \setlength{\unitlength}{\unitlength * \real{\svgscale}}%
    \fi%
  \else%
    \setlength{\unitlength}{\svgwidth}%
  \fi%
  \global\let\svgwidth\undefined%
  \global\let\svgscale\undefined%
  \makeatother%
  \begin{picture}(1,0.16686103)%
    \put(0,0){\includegraphics[width=\unitlength,page=1]{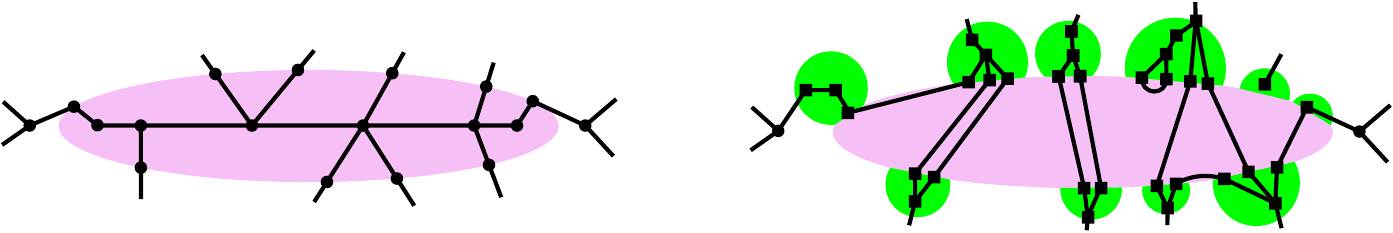}}%
    \put(0.07479119,0.1136075){\color[rgb]{0,0,0}\makebox(0,0)[lb]{\smash{$D_b$}}}%
    \put(0.63194474,0.11022495){\color[rgb]{0,0,0}\makebox(0,0)[lb]{\smash{$D_b$}}}%
    \put(0.48393463,0.07224731){\color[rgb]{0,0,0}\makebox(0,0)[lb]{\smash{$\Rightarrow$}}}%
  \end{picture}%
\endgroup%
\caption{The changes in the image graph caused by a bar simplification.}
\label{fig:overview-G}
\end{figure}
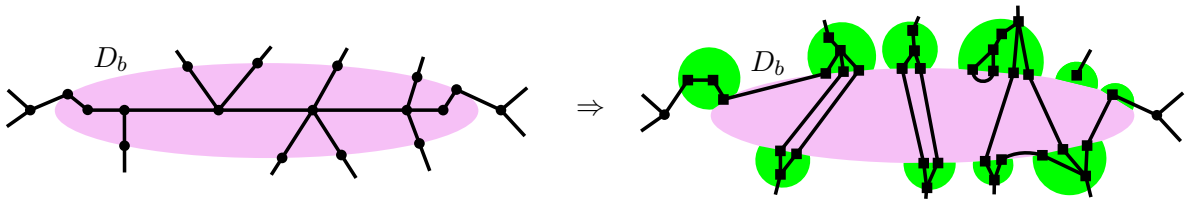

\medskip \noindent{\bf Changes in the polygon {\em P}.}
Refer to Figure~\ref{fig:overview-P}.
Consider a maximal path $p$ in $P$ that lies in  $D_b$.
The bar simplification replaces $p=[u,\ldots,v]$ with a new path $p'$.
By \ref{inv:deg}-\ref{inv:subdiv}, only nodes $u$ and $v$ in $p$ lie on $\partial D_b$.
If $p$ is the concatenation of a path $p_1$ and $p_1^{-1}$ (the path formed by the vertices of $p_1$ in reverse order),
then $p'$ is a spur in the cluster containing $u$ (Figure~\ref{fig:overview-P} (a)).
If $p$ has no such decomposition, but its two endpoints are at the same node, $u=v$,
then $p'$ is a single edge connecting two leaves in the cluster containing $u$ (Figure~\ref{fig:overview-P} (b)).
If the endpoints of $p$ are at two different nodes, $p'$ is an edge between two leaves of the clusters containing $u$ and $v$ respectively (Figure~\ref{fig:overview-P} (c) and (d)).

\begin{figure}[h!tb]
\centering
	\def\svgwidth{.9\textwidth}
\begingroup%
  \makeatletter%
  \providecommand\color[2][]{%
    \errmessage{(Inkscape) Color is used for the text in Inkscape, but the package 'color.sty' is not loaded}%
    \renewcommand\color[2][]{}%
  }%
  \providecommand\transparent[1]{%
    \errmessage{(Inkscape) Transparency is used (non-zero) for the text in Inkscape, but the package 'transparent.sty' is not loaded}%
    \renewcommand\transparent[1]{}%
  }%
  \providecommand\rotatebox[2]{#2}%
  \ifx\svgwidth\undefined%
    \setlength{\unitlength}{396.39069824bp}%
    \ifx\svgscale\undefined%
      \relax%
    \else%
      \setlength{\unitlength}{\unitlength * \real{\svgscale}}%
    \fi%
  \else%
    \setlength{\unitlength}{\svgwidth}%
  \fi%
  \global\let\svgwidth\undefined%
  \global\let\svgscale\undefined%
  \makeatother%
  \begin{picture}(1,0.81478592)%
    \put(0,0){\includegraphics[width=\unitlength,page=1]{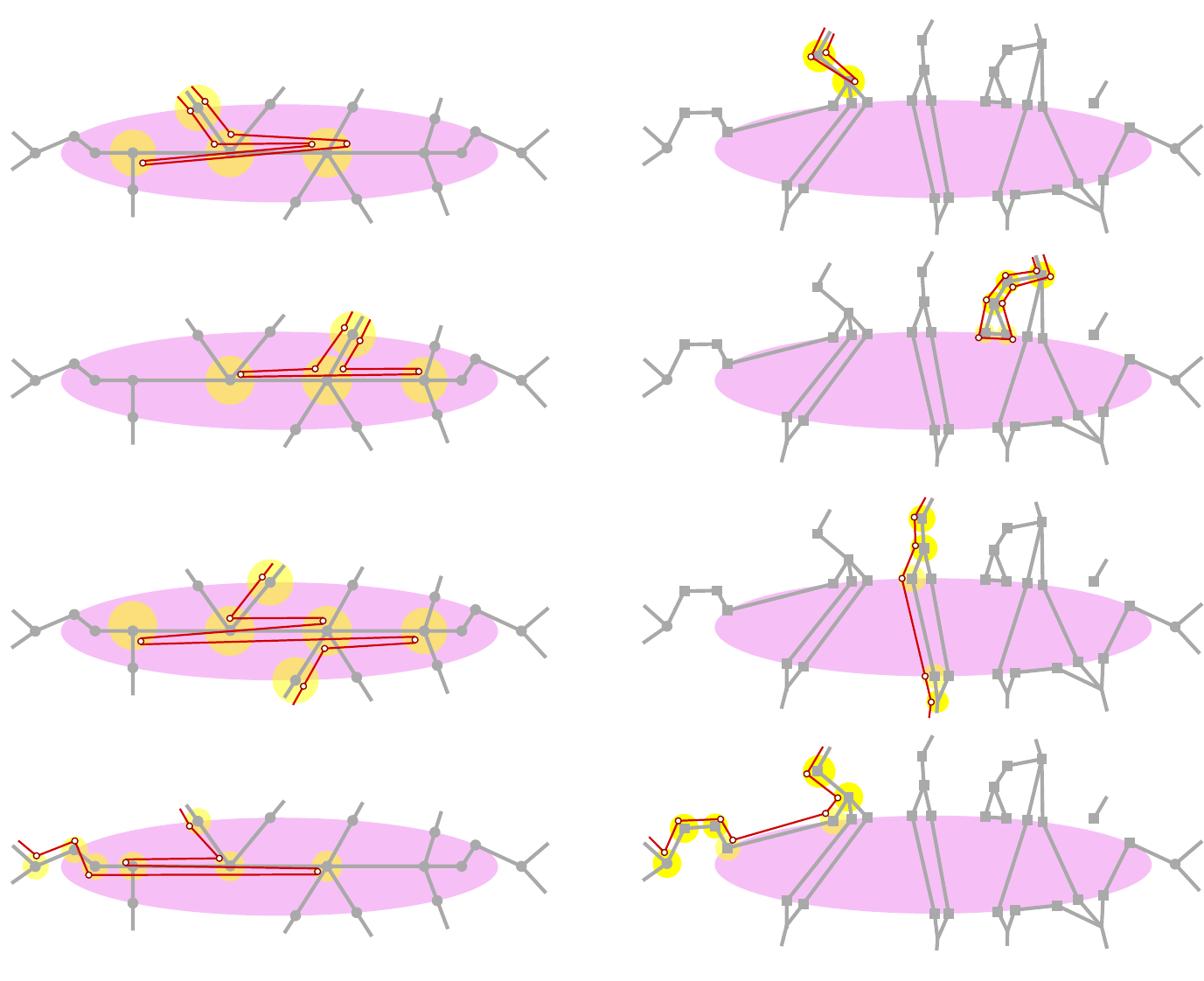}}%
    \put(0.4677143,0.60685907){\color[rgb]{0,0,0}\makebox(0,0)[lb]{\smash{(a)}}}%
    \put(0.4677143,0.39496295){\color[rgb]{0,0,0}\makebox(0,0)[lb]{\smash{(b)}}}%
    \put(0.4677143,0.19916632){\color[rgb]{0,0,0}\makebox(0,0)[lb]{\smash{(c)}}}%
    \put(0.4677143,0.01118823){\color[rgb]{0,0,0}\makebox(0,0)[lb]{\smash{(d)}}}%
    \put(0.48553243,0.68863713){\color[rgb]{0,0,0}\makebox(0,0)[lb]{\smash{$\Rightarrow$}}}%
    \put(0.48553243,0.49638681){\color[rgb]{0,0,0}\makebox(0,0)[lb]{\smash{$\Rightarrow$}}}%
    \put(0.48553243,0.29193012){\color[rgb]{0,0,0}\makebox(0,0)[lb]{\smash{$\Rightarrow$}}}%
    \put(0.48553243,0.09174567){\color[rgb]{0,0,0}\makebox(0,0)[lb]{\smash{$\Rightarrow$}}}%
  \end{picture}%
\endgroup%
\caption{The changes in the polygon caused by a bar simplification.}
\label{fig:overview-P}
\end{figure}
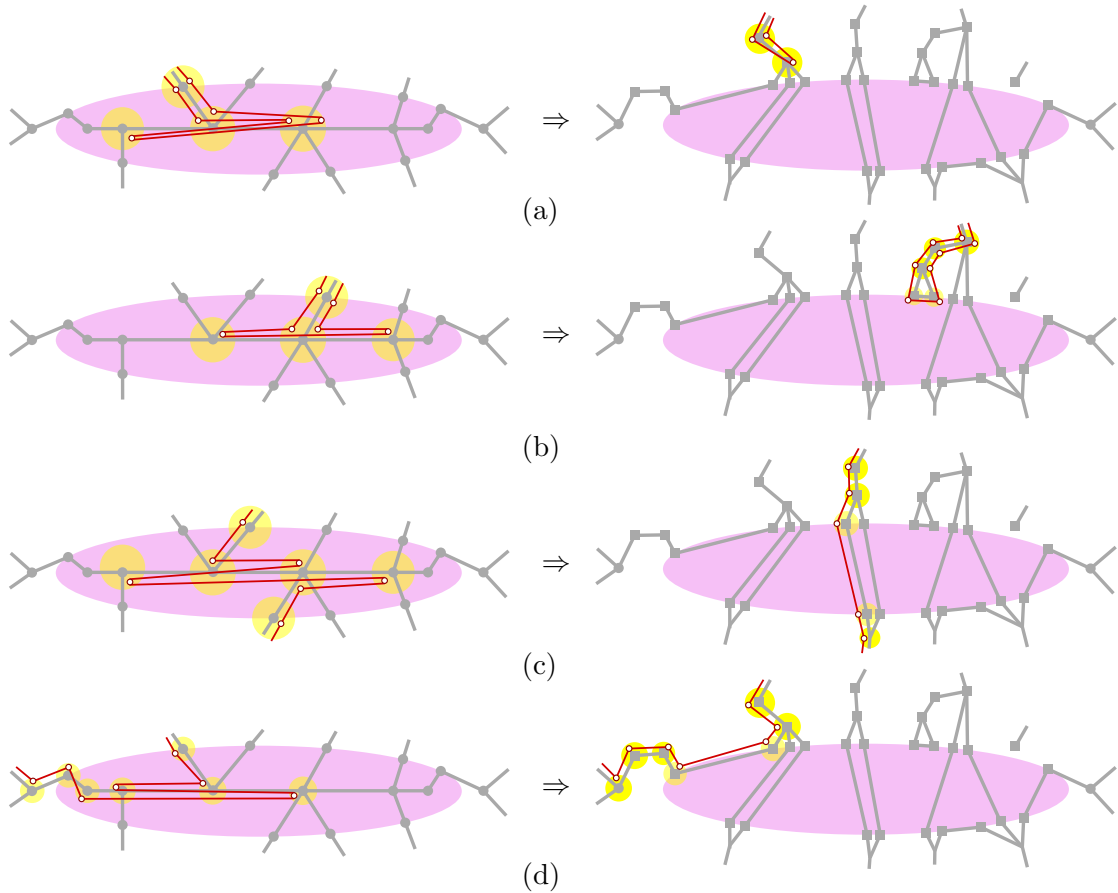

\subsection{Primitives}
\label{ssec:primitives}
The operations in Section~\ref{ssec:operations} rely on two basic steps, \spurreduction\/ and \nodesplit\/ (see Figure~\ref{fig:primitives}).
Together with \merge\ and \subdivision, these operations are called \emph{primitives}.

\begin{quote}
{\bf \spurreduction$(u,v)$.}
Assume
that every vertex at node $u$ has at least one incident edge $[u,v]$.
While there exists a path $[u,v,u]$, replace it with a single-vertex path $[u]$.
(See Figure~\ref{fig:primitives}, left.)
\end{quote}

\begin{quote}{\bf \nodesplit$(u,v,w)$.}
Assume that segments $uv$ and $vw$ are consecutive in radial order around $v$,
node $v$ is not in the interior of any edge that contains $uv$ or $vw$;
and $P$ has no spurs of the form $[u,v,u]$ or $[w,v,w]$.
Create node $v^*$ in the interior of the wedge $\angle uvw$ sufficiently close to $v$;
replace every path $[u,v,w]$ with $[u,v^*,w]$.
(See Figure~\ref{fig:primitives}, right.)
\end{quote}

\begin{figure}[h]
\centering
	\def\svgwidth{\textwidth}
\begingroup%
  \makeatletter%
  \providecommand\color[2][]{%
    \errmessage{(Inkscape) Color is used for the text in Inkscape, but the package 'color.sty' is not loaded}%
    \renewcommand\color[2][]{}%
  }%
  \providecommand\transparent[1]{%
    \errmessage{(Inkscape) Transparency is used (non-zero) for the text in Inkscape, but the package 'transparent.sty' is not loaded}%
    \renewcommand\transparent[1]{}%
  }%
  \providecommand\rotatebox[2]{#2}%
  \ifx\svgwidth\undefined%
    \setlength{\unitlength}{286.06499023bp}%
    \ifx\svgscale\undefined%
      \relax%
    \else%
      \setlength{\unitlength}{\unitlength * \real{\svgscale}}%
    \fi%
  \else%
    \setlength{\unitlength}{\svgwidth}%
  \fi%
  \global\let\svgwidth\undefined%
  \global\let\svgscale\undefined%
  \makeatother%
  \begin{picture}(1,0.13356086)%
    \put(0,0){\includegraphics[width=\unitlength,page=1]{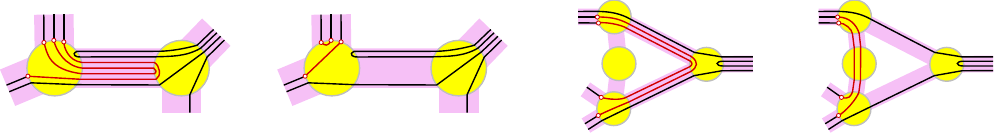}}%
    \put(0.07252801,0.02849422){\color[rgb]{0,0,0}\makebox(0,0)[lb]{\smash{$u$}}}%
    \put(0.21221591,0.04737147){\color[rgb]{0,0,0}\makebox(0,0)[lb]{\smash{$v$}}}%
    \put(0.35041521,0.02777881){\color[rgb]{0,0,0}\makebox(0,0)[lb]{\smash{$u$}}}%
    \put(0.49113513,0.04754412){\color[rgb]{0,0,0}\makebox(0,0)[lb]{\smash{$v$}}}%
    \put(0.63737132,0.00904274){\color[rgb]{0,0,0}\makebox(0,0)[lb]{\smash{$w$}}}%
    \put(0.58597128,0.06903929){\color[rgb]{0,0,0}\makebox(0,0)[lb]{\smash{$v^*$}}}%
    \put(0.63929552,0.12288175){\color[rgb]{0,0,0}\makebox(0,0)[lb]{\smash{$u$}}}%
    \put(0.70082087,0.09200199){\color[rgb]{0,0,0}\makebox(0,0)[lb]{\smash{$v$}}}%
    \put(0.82888029,0.06990255){\color[rgb]{0,0,0}\makebox(0,0)[lb]{\smash{$v^*$}}}%
    \put(0.88089689,0.12273403){\color[rgb]{0,0,0}\makebox(0,0)[lb]{\smash{$u$}}}%
    \put(0.87951564,0.00999233){\color[rgb]{0,0,0}\makebox(0,0)[lb]{\smash{$w$}}}%
    \put(0.94361818,0.09199466){\color[rgb]{0,0,0}\makebox(0,0)[lb]{\smash{$v$}}}%
    \put(0.24115352,0.06368612){\color[rgb]{0,0,0}\makebox(0,0)[lb]{\smash{$\Rightarrow$}}}%
    \put(0.78272659,0.06368612){\color[rgb]{0,0,0}\makebox(0,0)[lb]{\smash{$\Rightarrow$}}}%
  \end{picture}%
\endgroup%
\caption{Left: Spur-reduction$(u,v)$.
Right: Node-split$(u,v,w)$.}
\label{fig:primitives}
\end{figure}
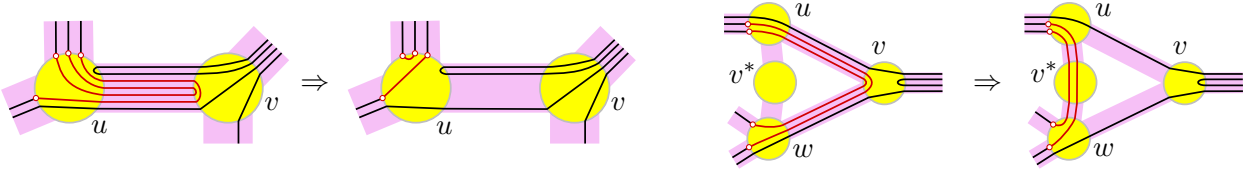

The following two lemmas are generalizations of the results in~\cite[Section~5]{CEX15}.

\begin{lemma}\label{lem:primitives1}
Operation \spurreduction\/ is ws-equivalent.
\end{lemma}
\begin{proof}
Let $P'$ be  obtained from applying \spurreduction$(u,v)$ to $P$.
First suppose that $P$ is weakly simple.
Then, there exists a simple polygon $Q\in \Phi(P)$ represented by its signature.
Successively replace any path $[u,v,u]$ by $[u]$ and delete these two edges from the ordering.
The new signature defines a polygon $Q'$ in the strip system of $P'$.
By the assumption in the operation, every edge of $Q$ in $D_u$ is adjacent to an edge in $N_{uv}$, which has another endpoint in $\partial D_v$. Since $Q$ is simple, the counterclockwise order of the endpoints of the deleted edges in $\partial D_v$ is the same as the clockwise order of the endpoints of the new edges in $\partial D_u$.
Thus, the new matching in $D_u$ produces no crossings, $Q'\in\Phi(P')$, and $P'$ is weakly simple.

Now suppose $P'$ is weakly simple.
Then, there exists a simple polygon $Q'\in \Phi(P')$ represented by its signature.
Let $H_u'$ be the set of all vertices in the node $u$ in $P'$.
Each vertex in $H_u'$ corresponds to an edge in $Q'$ that lies in the disk $D_u$; these edges are noncrossing
chords of the circle $\partial D_u$. We define a partial ordering on $H_u'$:
For two vertices $u_1,u_2\in H_u'$, let $u_1\prec u_2$ if the chord
corresponding to $u_1$ separates the chord of $u_2$ from $N_{uv}$ within the disk $D_u$.
Intuitively, we have $u_1\prec u_2$ if $u_1$ blocks $u_2$ from the corridor $N_{uv}$.
Note that if $u_1\prec u_2$, then neither endpoint of the chord corresponding
to $u_1$ is on the boundary of $N_{uv}$; consequently $u_1$ was obtained from a path $[u,v,u]$
or $[u,v,u,v,u,\ldots, u]$ in $P$ after removing one or more spurs.
We expand the paths $u_i\in H_u'$ incrementally, in an order
determined by any linear extension of the partial ordering $\prec$.
Replace the first vertex $u_1\in H_u'$ by $[p,u,v,u]$ (or $[p,u,v,u,v,u,\ldots, u,q]$ if needed),
and modify the signature by inserting consecutive new edges into the total order of
the edges along $uv$ at any position that is not separated from the
chord in $D_u$ that corresponds to $u_1$.
The resulting polygon $P''$ and the new signature define a polygon $Q''$ in the strip system of $P''$.
By construction, the new edges in $D_v$ connect consecutive endpoints in counterclockwise
order around $v$, thus the new matching in $D_v$ is noncrossing.
In the disk $D_u$, the operation replaces the chord corresponding to $u_1$ by
noncrossing new chords. Each new edge in $D_u$ has at least one endpoint in $N_{uv}$;
consequently, none of them blocks access to $N_{u,v}$.
Then, the new matching in $D_u$ has no crossing and $Q''\in \Phi(P'')$.
By repeating this procedure we obtain $P$ and a simple polygon $Q\in\Phi(P)$,
hence $P$ is weakly simple.
%
\end{proof}

\begin{lemma}\label{lem:primitives2}
Operation \nodesplit\/ is ws-equivalent.
\end{lemma}
\begin{proof}
Let $P'$ be obtained from $P$ via  \nodesplit$(u,v,w)$.
First assume  that $P$ is weakly simple.
Then there is a simple polygon $Q\in \Phi(P)$.
Consider the clockwise order of edges around $v$.
Since $Q$ is simple, the order of the edges $[u,v]$ of paths $[u,v,w]$ must be the reverse order of its adjacent edges $[v,w]$ (the paths must be nested as shown in Figure~\ref{fig:primitives}(right)).
Because $P$ has no spurs of the form $[u,v,u]$ or $[w,v,w]$, and the edges of $P$ that pass through $v$
avoid both $uv$ and $vw$,
every edge between a pair of adjacent edges $[u,v]$ and $[v,w]$ is also part of a path $[u,v,w]$. 
Replace the paths $[u,v,w]$ by $[u,v^*,w]$ and set the order of edges at segments $uv^*$ and $v^*w$ to be the same order of the removed edges at $uv$ and $vw$.
This defines a polygon $Q'\in\Phi(P')$, which is simple because the circular order of endpoints around $D_u$ and $D_w$ remains unchanged and the matching in $D_{v^*}$ is a subset of the matching in $D_{v}$.

Now, assume that $P'$ is weakly simple. Since the face in the image graph bounded by $u,v,w,v^*$ is empty,
we can change the embedding of the graph by bringing $v^*$ arbitrarily close to $v$, maintaining weak simplicity.
Let $\delta$ be the distance between $v^*$ and $v$.
Let $Q'\in \Phi(P')$ be a simple polygon defined on disks of radius $\eps$.
Then, $Q'$ is within $\eps+\delta$ Fr\'echet distance from $P$ and therefore $P$ is weakly simple.
\end{proof}

\subsection{Operations}
\label{ssec:operations}
We describe three complex operations: \pinextraction, \Vshortcut, and \Lshortcut.
In Section~\ref{ssec:phases}, we show how to use them to eliminate spurs
along any given bar $b$.
The \pinextraction\ and \Vshortcut\ operations eliminate pins and V-chains.
Chains in $\PP$ with two or more vertices in the interior of $D_b$ are simplified incrementally, removing one vertex at a time, by the \Lshortcut\/ operation.

Since the image graph is determined by the polygon, it would suffice to describe how the operations modify the polygon. However, it is sometimes more convenient to first define new nodes and segments in the image graph, and use them to describe the changes in the polygon. In the last step of these operations, we remove any node (segment) that contains no vertex (edge), to ensure that the image graph is consistent with the polygon.
\begin{quote}
{\bf \pinextraction$(u,v)$.}
Assume that $P$ satisfies \ref{inv:tree}--\ref{inv:subdiv} and contains a pin $[v,u,v]\in \Pin$.
By \ref{inv:deg}, node $v$ is adjacent to a unique node $w$ outside of $D_b$.
Perform the following three primitives:
(1) \subdivision\ of every path $[v,w]$ into $[v,w^*,w]$;
(2) \spurreduction$(v,u)$.
(3) \spurreduction$(w^*,v)$.
(4) Update the image graph.
See Figure~\ref{fig:pin-extraction} for an example.
\end{quote}

\begin{quote}
{\bf \Vshortcut$(v_1,u,v_2)$.}
Assume that $P$ satisfies \ref{inv:tree}--\ref{inv:subdiv} and $[v_1,u,v_2]\in \VV$.
Furthermore, $P$ contains no pin of the form $[v_1,u,v_1]$ or $[v_2,u,v_2]$, and no edge $[u,q]$ such that segment $uq$ is in the interior of the wedge $\angle v_1uv_2$.
By \ref{inv:deg}, nodes $v_1$ and $v_2$ are each adjacent to unique nodes $w_1$ and $w_2$ outside of $D_b$, respectively.

The operation executes the following primitives sequentially:
(1) \nodesplit$(v_1,u,v_2)$, which creates a temporary node $u^*$;
(2) \nodesplit$(u^*,v_1,w_1)$ and \nodesplit$(u^*,v_2,w_2)$;
    which create $v_1^*,v_2^*\in \partial D_b$, respectively;
(3) \merge\ every path $[v_1^*,u^*,v_2^*]$ to $[v_1^*,v_2^*]$.
(4) Update the image graph.
See Figure~\ref{fig:shortcut} for an example.
\end{quote}

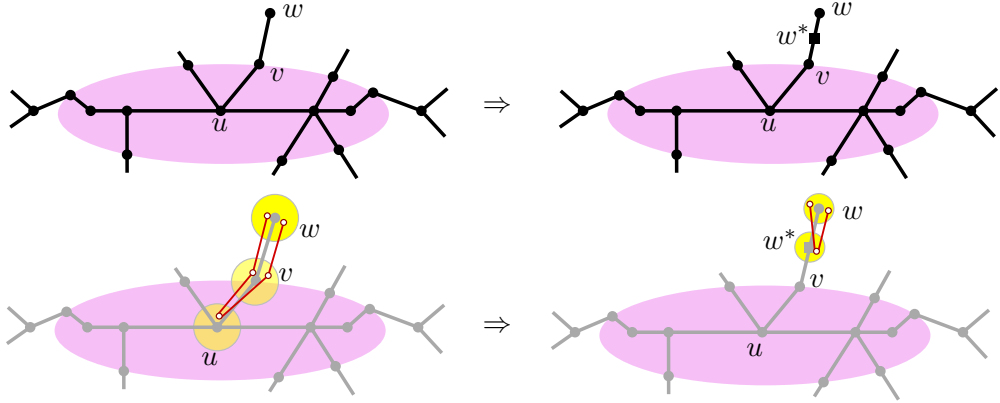
\begin{figure}[h]
\centering
	\def\svgwidth{.8\textwidth}
\begingroup%
  \makeatletter%
  \providecommand\color[2][]{%
    \errmessage{(Inkscape) Color is used for the text in Inkscape, but the package 'color.sty' is not loaded}%
    \renewcommand\color[2][]{}%
  }%
  \providecommand\transparent[1]{%
    \errmessage{(Inkscape) Transparency is used (non-zero) for the text in Inkscape, but the package 'transparent.sty' is not loaded}%
    \renewcommand\transparent[1]{}%
  }%
  \providecommand\rotatebox[2]{#2}%
  \ifx\svgwidth\undefined%
    \setlength{\unitlength}{340.13549805bp}%
    \ifx\svgscale\undefined%
      \relax%
    \else%
      \setlength{\unitlength}{\unitlength * \real{\svgscale}}%
    \fi%
  \else%
    \setlength{\unitlength}{\svgwidth}%
  \fi%
  \global\let\svgwidth\undefined%
  \global\let\svgscale\undefined%
  \makeatother%
  \begin{picture}(1,0.40440649)%
    \put(0,0){\includegraphics[width=\unitlength,page=1]{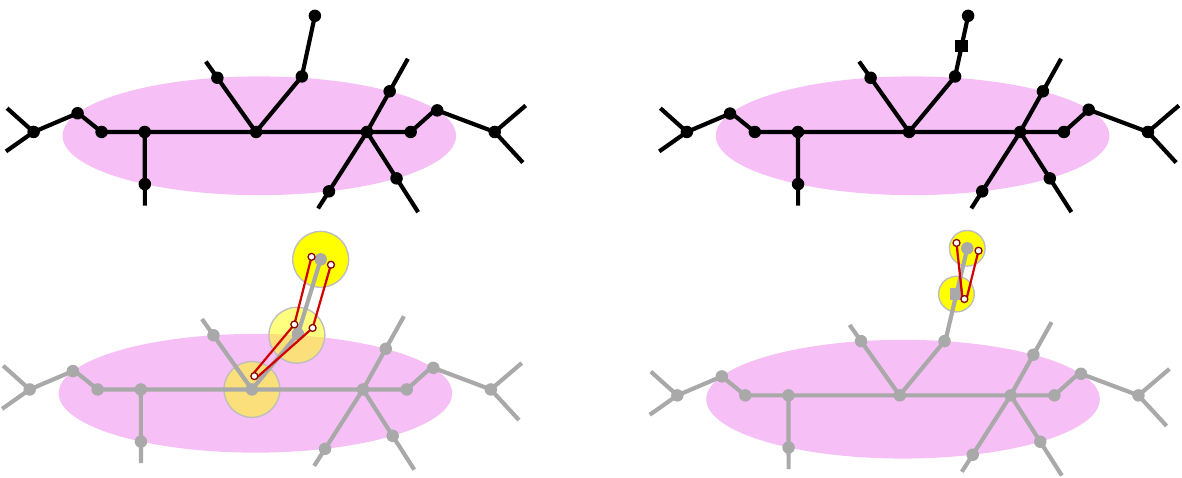}}%
    \put(0.27918366,0.38686045){\color[rgb]{0,0,0}\makebox(0,0)[lb]{\smash{$w$}}}%
    \put(0.26629627,0.32320761){\color[rgb]{0,0,0}\makebox(0,0)[lb]{\smash{$v$}}}%
    \put(0.20751802,0.27209977){\color[rgb]{0,0,0}\makebox(0,0)[lb]{\smash{$u$}}}%
    \put(0.29568539,0.16514986){\color[rgb]{0,0,0}\makebox(0,0)[lb]{\smash{$w$}}}%
    \put(0.27517419,0.12156479){\color[rgb]{0,0,0}\makebox(0,0)[lb]{\smash{$v$}}}%
    \put(0.19748922,0.03638174){\color[rgb]{0,0,0}\makebox(0,0)[lb]{\smash{$u$}}}%
    \put(0.83153794,0.38904262){\color[rgb]{0,0,0}\makebox(0,0)[lb]{\smash{$w$}}}%
    \put(0.77591292,0.35981923){\color[rgb]{0,0,0}\makebox(0,0)[lb]{\smash{$w^*$}}}%
    \put(0.81538595,0.31919424){\color[rgb]{0,0,0}\makebox(0,0)[lb]{\smash{$v$}}}%
    \put(0.7598616,0.27263173){\color[rgb]{0,0,0}\makebox(0,0)[lb]{\smash{$u$}}}%
    \put(0.74624496,0.04538366){\color[rgb]{0,0,0}\makebox(0,0)[lb]{\smash{$u$}}}%
    \put(0.76259685,0.15321744){\color[rgb]{0,0,0}\makebox(0,0)[lb]{\smash{$w^*$}}}%
    \put(0.80966286,0.11493949){\color[rgb]{0,0,0}\makebox(0,0)[lb]{\smash{$v$}}}%
    \put(0.84242183,0.18389475){\color[rgb]{0,0,0}\makebox(0,0)[lb]{\smash{$w$}}}%
    \put(0.48076988,0.29146466){\color[rgb]{0,0,0}\makebox(0,0)[lb]{\smash{$\Rightarrow$}}}%
    \put(0.48076988,0.06958966){\color[rgb]{0,0,0}\makebox(0,0)[lb]{\smash{$\Rightarrow$}}}%
  \end{picture}%
\endgroup%
\caption{\pinextraction. Changes in the image graph (top), changes in the polygon (bottom).}
\label{fig:pin-extraction}
\end{figure}
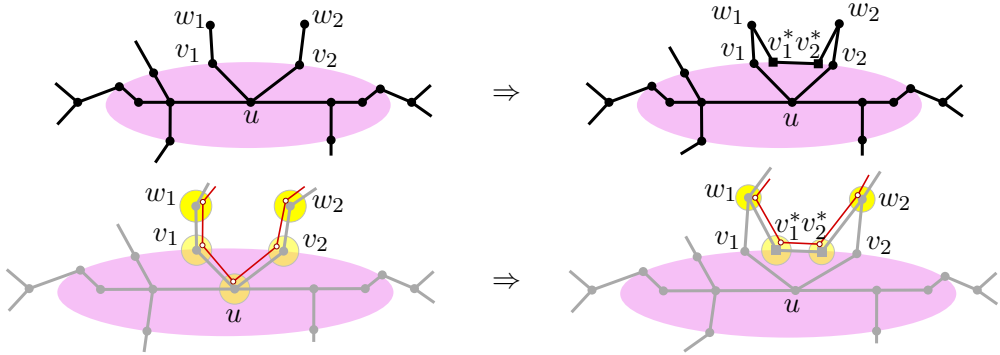
\begin{figure}[h]
\centering
	\def\svgwidth{.8\textwidth}
\begingroup%
  \makeatletter%
  \providecommand\color[2][]{%
    \errmessage{(Inkscape) Color is used for the text in Inkscape, but the package 'color.sty' is not loaded}%
    \renewcommand\color[2][]{}%
  }%
  \providecommand\transparent[1]{%
    \errmessage{(Inkscape) Transparency is used (non-zero) for the text in Inkscape, but the package 'transparent.sty' is not loaded}%
    \renewcommand\transparent[1]{}%
  }%
  \providecommand\rotatebox[2]{#2}%
  \ifx\svgwidth\undefined%
    \setlength{\unitlength}{395.15600586bp}%
    \ifx\svgscale\undefined%
      \relax%
    \else%
      \setlength{\unitlength}{\unitlength * \real{\svgscale}}%
    \fi%
  \else%
    \setlength{\unitlength}{\svgwidth}%
  \fi%
  \global\let\svgwidth\undefined%
  \global\let\svgscale\undefined%
  \makeatother%
  \begin{picture}(1,0.35018072)%
    \put(0,0){\includegraphics[width=\unitlength,page=1]{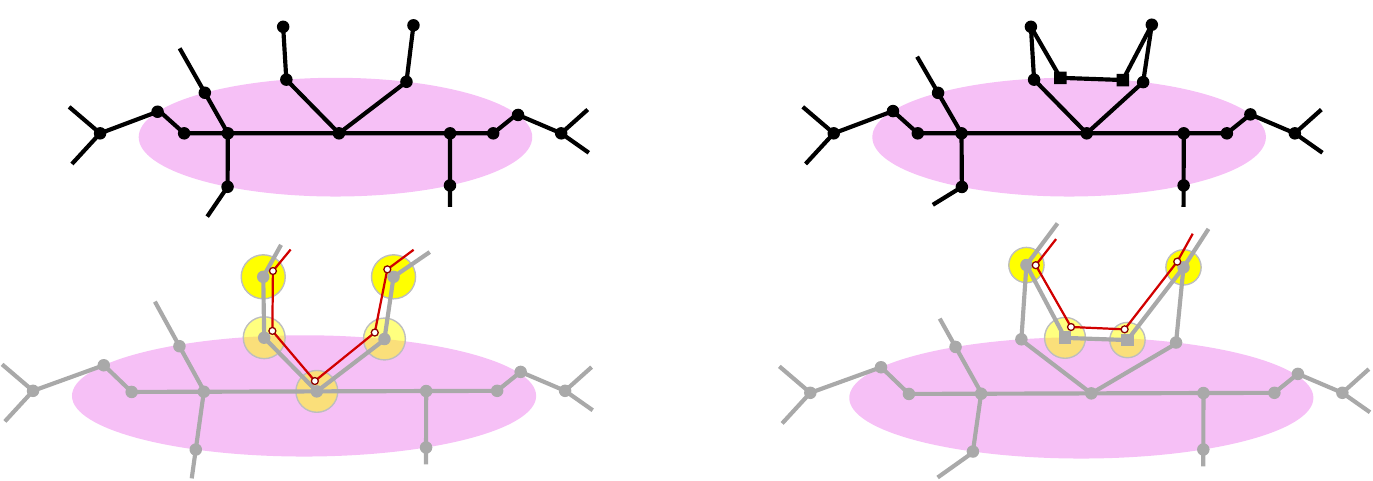}}%
    \put(0.23950538,0.23171934){\color[rgb]{0,0,0}\makebox(0,0)[lb]{\smash{$u$}}}%
    \put(0.22155035,0.03193418){\color[rgb]{0,0,0}\makebox(0,0)[lb]{\smash{$u$}}}%
    \put(0.78364235,0.04207644){\color[rgb]{0,0,0}\makebox(0,0)[lb]{\smash{$u$}}}%
    \put(0.78281298,0.22963508){\color[rgb]{0,0,0}\makebox(0,0)[lb]{\smash{$u$}}}%
    \put(0.16954477,0.33235989){\color[rgb]{0,0,0}\makebox(0,0)[lb]{\smash{$w_1$}}}%
    \put(0.14051547,0.15368746){\color[rgb]{0,0,0}\makebox(0,0)[lb]{\smash{$w_1$}}}%
    \put(0.71066076,0.33655092){\color[rgb]{0,0,0}\makebox(0,0)[lb]{\smash{$w_1$}}}%
    \put(0.7034186,0.1581071){\color[rgb]{0,0,0}\makebox(0,0)[lb]{\smash{$w_1$}}}%
    \put(0.30864648,0.33337518){\color[rgb]{0,0,0}\makebox(0,0)[lb]{\smash{$w_2$}}}%
    \put(0.30864648,0.146713){\color[rgb]{0,0,0}\makebox(0,0)[lb]{\smash{$w_2$}}}%
    \put(0.8785191,0.14889129){\color[rgb]{0,0,0}\makebox(0,0)[lb]{\smash{$w_2$}}}%
    \put(0.84802315,0.33957491){\color[rgb]{0,0,0}\makebox(0,0)[lb]{\smash{$w_2$}}}%
    \put(0.17061017,0.29835527){\color[rgb]{0,0,0}\makebox(0,0)[lb]{\smash{$v_1$}}}%
    \put(0.14869396,0.11375868){\color[rgb]{0,0,0}\makebox(0,0)[lb]{\smash{$v_1$}}}%
    \put(0.71997306,0.29961177){\color[rgb]{0,0,0}\makebox(0,0)[lb]{\smash{$v_1$}}}%
    \put(0.71377333,0.11001729){\color[rgb]{0,0,0}\makebox(0,0)[lb]{\smash{$v_1$}}}%
    \put(0.30904287,0.29343135){\color[rgb]{0,0,0}\makebox(0,0)[lb]{\smash{$v_2$}}}%
    \put(0.29904895,0.10747908){\color[rgb]{0,0,0}\makebox(0,0)[lb]{\smash{$v_2$}}}%
    \put(0.86642076,0.10482292){\color[rgb]{0,0,0}\makebox(0,0)[lb]{\smash{$v_2$}}}%
    \put(0.84193826,0.29468229){\color[rgb]{0,0,0}\makebox(0,0)[lb]{\smash{$v_2$}}}%
    \put(0.76739255,0.30589533){\color[rgb]{0,0,0}\makebox(0,0)[lb]{\smash{$v_1^*$}}}%
    \put(0.77430603,0.12205249){\color[rgb]{0,0,0}\makebox(0,0)[lb]{\smash{$v_1^*$}}}%
    \put(0.7950274,0.3054764){\color[rgb]{0,0,0}\makebox(0,0)[lb]{\smash{$v_2^*$}}}%
    \put(0.80121438,0.12314998){\color[rgb]{0,0,0}\makebox(0,0)[lb]{\smash{$v_2^*$}}}%
    \put(0.49004249,0.25550615){\color[rgb]{0,0,0}\makebox(0,0)[lb]{\smash{$\Rightarrow$}}}%
    \put(0.49004249,0.06640721){\color[rgb]{0,0,0}\makebox(0,0)[lb]{\smash{$\Rightarrow$}}}%
  \end{picture}%
\endgroup%
\caption{\Vshortcut. Changes in the image graph (top), changes in the polygon (bottom).}
\label{fig:shortcut}
\end{figure}

\begin{lemma}\label{lem:pin-and-V}
 \pinextraction\/ and \Vshortcut\/  are ws-equivalent and maintain  \ref{cond:A1}--\ref{cond:A2} in $D_b$ and \ref{inv:tree}--\ref{inv:subdiv} in adjacent clusters.
\end{lemma}
\begin{proof}
{\bf  \pinextraction. }
By construction, the operation maintains \ref{cond:A1}--\ref{cond:A2} in $D_b$ and \ref{inv:tree}--\ref{inv:subdiv} in adjacent clusters.
Also, \ref{inv:deg}--\ref{inv:subdiv} ensure that \spurreduction$(v,u)$ in step (2) satisfies its preconditions. Consequently, all three primitives are ws-equivalent.

{\bf \Vshortcut.}
By construction, the operation maintains \ref{cond:A1}--\ref{cond:A2} in $D_b$ and \ref{inv:tree}--\ref{inv:subdiv} in adjacent clusters.
The first two primitives are ws-equivalent by Lemma~\ref{lem:primitives2}.
The third step is ws-equivalent because triangle $\Delta(u^* v_1^* v_2^*)$
is empty of nodes and segments, by assumption.
\end{proof}

\medskip\noindent{\bf \Lshortcut\/ operation.}
The purpose of this operation is to eliminate a vertex of a path that has an edge along a given bar.
Before describing the operation, we introduce some notation; refer to Figure~\ref{fig:classify}.
For a node $v\in \partial D_b$, let $L_v$ be the set of paths $[v,u_1,u_2]$ in $P$ such that $u_1,u_2\in {\rm int}(D_b)$.
Each path in $\PP$ is either in $\Pin$, in $\VV$, or has two subpaths in some $L_v$.
Let $M_{cr}$ be the set of longest edges of cross-chains in $\PP$.
Denote by $\widehat{L_v}\subset L_v$ the set of paths $[v,u_1,u_2]$, where $[u_1,u_2]$ is \emph{not}
in $M_{cr}$.

\begin{figure}[h]
\centering
\includegraphics[width=0.95\linewidth]{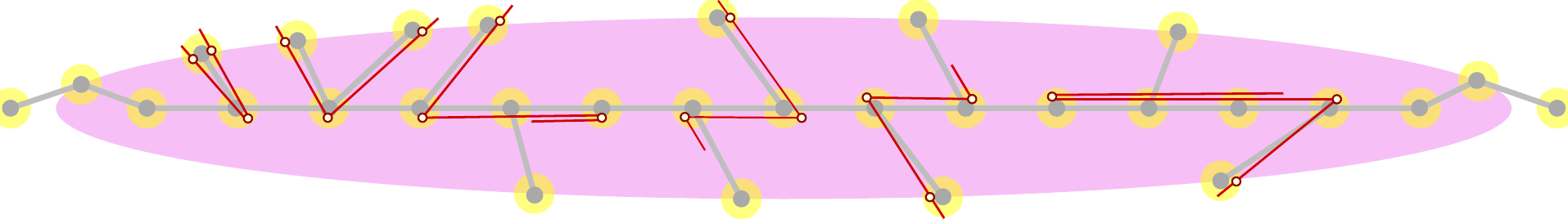}
\caption{Paths in $\Pin$, $\VV$, $L_v^{TR}$, $L_v^{TL}$, $L_v^{BR}$, and $L_v^{BL}$.}
\label{fig:classify}
\end{figure}

We partition $L_v$ into four subsets (refer to Figure~\ref{fig:classify}): a path $[v,u_1,u_2]\in L_v$ is in
\begin{enumerate}\itemsep  -1pt
\item $L_v^{TR}$ (\emph{top-right})    if $v$ is a \emph{top} vertex and $x(u_1)<x(u_2)$;
\item $L_v^{TL}$ (\emph{top-left})     if $v$ is a \emph{top} vertex and $x(u_1)>x(u_2)$;
\item $L_v^{BR}$ (\emph{bottom-right}) if $v$ is a \emph{bottom} vertex and $x(u_1)<x(u_2)$;
\item $L_v^{BL}$ (\emph{bottom-left})  if $v$ is a \emph{bottom} vertex and $x(u_1)<x(u_2)$.
\end{enumerate}
We partition $\widehat{L_v}$ into four subsets analogously.
We define the operation \Lshortcut\/ for paths in $L_v^{TR}$; the definition for the other subsets can be obtained by suitable reflections.

\begin{quote}
{\bf \Lshortcut$(v,TR)$.}
Assume that $P$ satisfies \ref{inv:tree}--\ref{inv:subdiv}, $v\in \partial D_b$ and $L_v^{TR}\neq \emptyset$.
By \ref{inv:deg}, $v$ is adjacent to a unique node $u_1\in b$ and to a unique node $w\notin D_b$. Let $U$ denote the set of all nodes $u_2$ for which $[v,u_1,u_2]\in L_v^{TR}$. Let $u_{\min}\in U$ and $u_{\max}\in U$ be the leftmost and rightmost node in $U$, respectively.
Further assume that $P$ satisfies:
\begin{enumerate}[label=(B\arabic*),start=1]\itemsep -2pt
\item \label{cond:B1}
there is no pin of the form $[v,u_1,v]$;
\item \label{cond:B2}
no edge $[p,u_1]$ such that segment $pu_1$ is in the interior of the wedge $\angle vu_1u_{\min}$;
\item \label{cond:B3}
no edge $[p,q]$ such that $p\in \partial D_b$ is a top vertex and $q\in b$, $x(u_1)<x(q)<x(u_{\max})$.
\end{enumerate}
Do the following (see Figure~\ref{fig:spur-simplification} for an example).
\begin{enumerate}[label=(\arabic*),start=0]\itemsep -2pt
\item  Create a new node $v^*\in \partial D_b$ to
 the right of $v$ sufficiently close to $v$.
\item \label{phase1} For every path $[v,u_1,u_2]\in L_v^{TR}$ in which $u_1u_2$ is the \emph{only} longest edge of a cross-chain, create a crimp by replacing $[u_1,u_2]$ with $[u_1,u_2,u_1,u_2]$.
\item \label{phase2} Replace every path $[w,v,u_1,u_{\min}]$ by $[w,v^*,u_{\min}]$.
\item \label{phase3} Replace every path $[w,v,u_1,u_2]$, where $u_2\in U$ and $u_2\neq u_{\min}$, by $[w,v^*,u_{\min},u_2]$.
\item \label{phase4} Update the image graph.
\end{enumerate}
\end{quote}

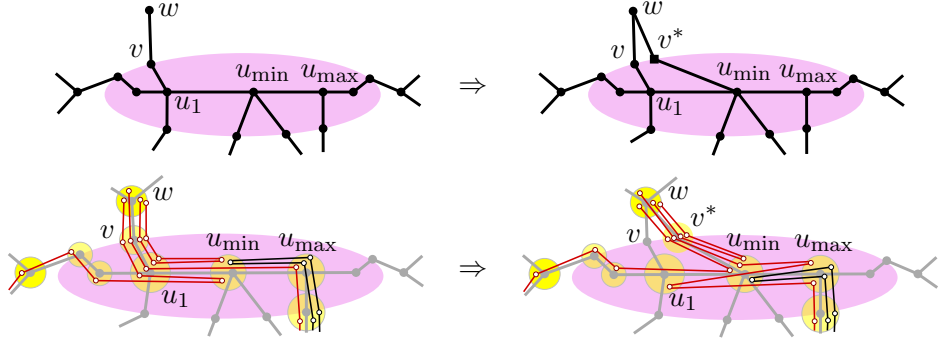
\begin{figure}[htbp]
\centering
	\def\svgwidth{.75\textwidth}
\begingroup%
  \makeatletter%
  \providecommand\color[2][]{%
    \errmessage{(Inkscape) Color is used for the text in Inkscape, but the package 'color.sty' is not loaded}%
    \renewcommand\color[2][]{}%
  }%
  \providecommand\transparent[1]{%
    \errmessage{(Inkscape) Transparency is used (non-zero) for the text in Inkscape, but the package 'transparent.sty' is not loaded}%
    \renewcommand\transparent[1]{}%
  }%
  \providecommand\rotatebox[2]{#2}%
  \ifx\svgwidth\undefined%
    \setlength{\unitlength}{382.35900879bp}%
    \ifx\svgscale\undefined%
      \relax%
    \else%
      \setlength{\unitlength}{\unitlength * \real{\svgscale}}%
    \fi%
  \else%
    \setlength{\unitlength}{\svgwidth}%
  \fi%
  \global\let\svgwidth\undefined%
  \global\let\svgscale\undefined%
  \makeatother%
  \begin{picture}(1,0.35931152)%
    \put(0,0){\includegraphics[width=\unitlength,page=1]{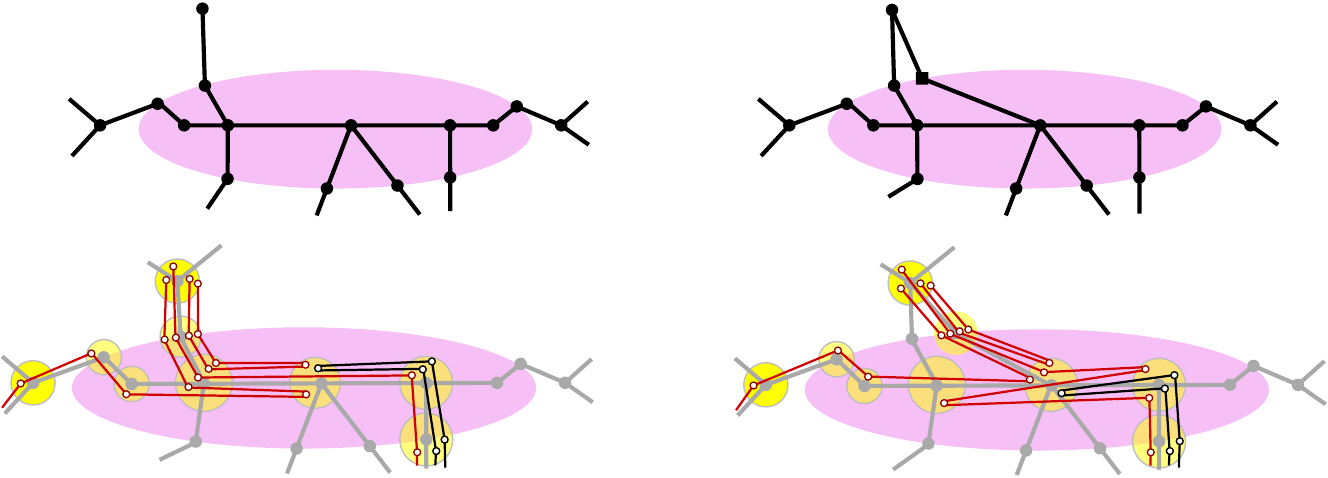}}%
    \put(0.16183673,0.34667531){\color[rgb]{0,0,0}\makebox(0,0)[lb]{\smash{$w$}}}%
    \put(0.1557278,0.1491438){\color[rgb]{0,0,0}\makebox(0,0)[lb]{\smash{$w$}}}%
    \put(0.68172301,0.34822904){\color[rgb]{0,0,0}\makebox(0,0)[lb]{\smash{$w$}}}%
    \put(0.70920726,0.149787){\color[rgb]{0,0,0}\makebox(0,0)[lb]{\smash{$w$}}}%
    \put(0.76153982,0.2830107){\color[rgb]{0,0,0}\makebox(0,0)[lb]{\smash{$u_{\min}$}}}%
    \put(0.13045644,0.30234321){\color[rgb]{0,0,0}\makebox(0,0)[lb]{\smash{$v$}}}%
    \put(0.09923589,0.1080705){\color[rgb]{0,0,0}\makebox(0,0)[lb]{\smash{$v$}}}%
    \put(0.16516711,0.03844691){\color[rgb]{0,0,0}\makebox(0,0)[lb]{\smash{$u_1$}}}%
    \put(0.64753135,0.30155136){\color[rgb]{0,0,0}\makebox(0,0)[lb]{\smash{$v$}}}%
    \put(0.66442725,0.10418293){\color[rgb]{0,0,0}\makebox(0,0)[lb]{\smash{$v$}}}%
    \put(0.17933462,0.24663845){\color[rgb]{0,0,0}\makebox(0,0)[lb]{\smash{$u_1$}}}%
    \put(0.69617619,0.24533128){\color[rgb]{0,0,0}\makebox(0,0)[lb]{\smash{$u_1$}}}%
    \put(0.70998496,0.03527508){\color[rgb]{0,0,0}\makebox(0,0)[lb]{\smash{$u_1$}}}%
    \put(0.24453446,0.28143168){\color[rgb]{0,0,0}\makebox(0,0)[lb]{\smash{$u_{\min}$}}}%
    \put(0.21330425,0.09941177){\color[rgb]{0,0,0}\makebox(0,0)[lb]{\smash{$u_{\min}$}}}%
    \put(0.77196935,0.09941177){\color[rgb]{0,0,0}\makebox(0,0)[lb]{\smash{$u_{\min}$}}}%
    \put(0.31466004,0.28053453){\color[rgb]{0,0,0}\makebox(0,0)[lb]{\smash{$u_{\max}$}}}%
    \put(0.29082744,0.09916491){\color[rgb]{0,0,0}\makebox(0,0)[lb]{\smash{$u_{\max}$}}}%
    \put(0.8285069,0.28068403){\color[rgb]{0,0,0}\makebox(0,0)[lb]{\smash{$u_{\max}$}}}%
    \put(0.84480467,0.09856683){\color[rgb]{0,0,0}\makebox(0,0)[lb]{\smash{$u_{\max}$}}}%
    \put(0.69812405,0.3119341){\color[rgb]{0,0,0}\makebox(0,0)[lb]{\smash{$v^*$}}}%
    \put(0.73331992,0.12122839){\color[rgb]{0,0,0}\makebox(0,0)[lb]{\smash{$v^*$}}}%
    \put(0.48448158,0.26434569){\color[rgb]{0,0,0}\makebox(0,0)[lb]{\smash{$\Rightarrow$}}}%
    \put(0.48448158,0.0676922){\color[rgb]{0,0,0}\makebox(0,0)[lb]{\smash{$\Rightarrow$}}}%
  \end{picture}%
\endgroup%
\caption{\Lshortcut. Changes in the image graph (top), changes in the polygon (bottom).}
\label{fig:spur-simplification}
\end{figure}

See Figure~\ref{fig:spur-simplification-diff} for an explanation of why  \Lshortcut\/
requires  conditions \ref{cond:B2}--\ref{cond:B3} and phase \ref{phase1} of the operation.
If we omit any of these conditions, \Lshortcut\/ would not be ws-equivalent.
\begin{figure}[htbp]
\centering
	\def\svgwidth{\textwidth}
\begingroup%
  \makeatletter%
  \providecommand\color[2][]{%
    \errmessage{(Inkscape) Color is used for the text in Inkscape, but the package 'color.sty' is not loaded}%
    \renewcommand\color[2][]{}%
  }%
  \providecommand\transparent[1]{%
    \errmessage{(Inkscape) Transparency is used (non-zero) for the text in Inkscape, but the package 'transparent.sty' is not loaded}%
    \renewcommand\transparent[1]{}%
  }%
  \providecommand\rotatebox[2]{#2}%
  \ifx\svgwidth\undefined%
    \setlength{\unitlength}{430.23999023bp}%
    \ifx\svgscale\undefined%
      \relax%
    \else%
      \setlength{\unitlength}{\unitlength * \real{\svgscale}}%
    \fi%
  \else%
    \setlength{\unitlength}{\svgwidth}%
  \fi%
  \global\let\svgwidth\undefined%
  \global\let\svgscale\undefined%
  \makeatother%
  \begin{picture}(1,0.33990332)%
    \put(0,0){\includegraphics[width=\unitlength,page=1]{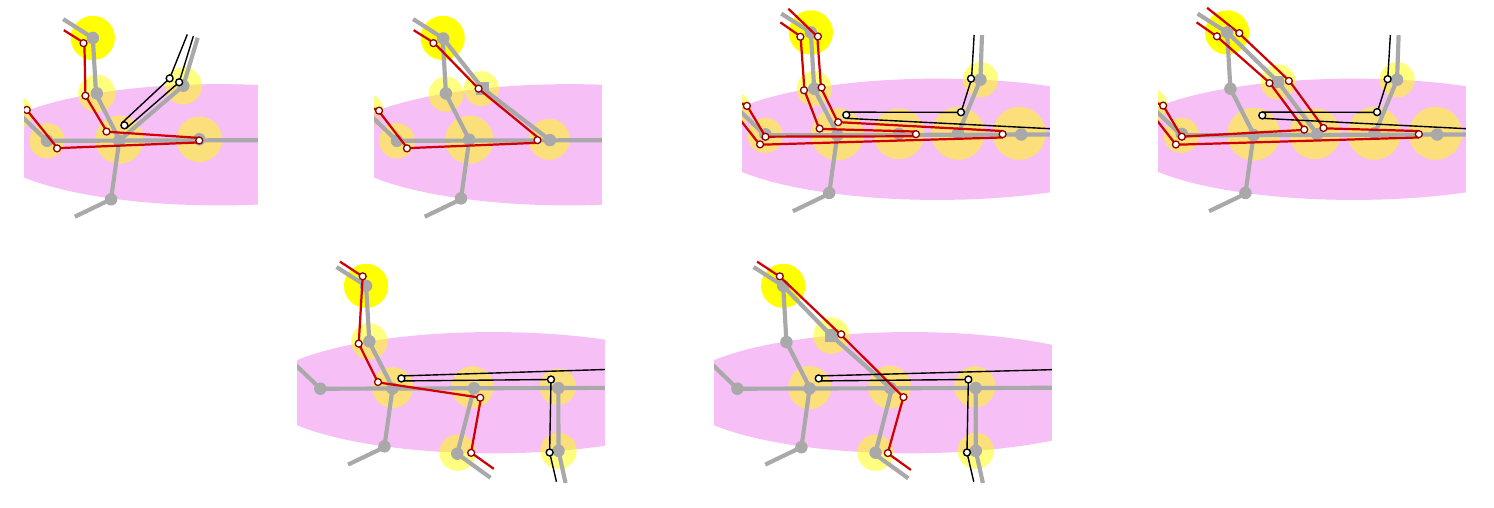}}%
    \put(0.02286202,0.31087849){\color[rgb]{0,0,0}\makebox(0,0)[lb]{\smash{$w$}}}%
    \put(0.03502536,0.27706556){\color[rgb]{0,0,0}\makebox(0,0)[lb]{\smash{$v$}}}%
    \put(0.08856851,0.22170624){\color[rgb]{0,0,0}\makebox(0,0)[lb]{\smash{$u_1$}}}%
    \put(0.13337026,0.21871759){\color[rgb]{0,0,0}\makebox(0,0)[lb]{\smash{$u_{\min}$}}}%
    \put(0.25793167,0.31164877){\color[rgb]{0,0,0}\makebox(0,0)[lb]{\smash{$w$}}}%
    \put(0.27392921,0.27702869){\color[rgb]{0,0,0}\makebox(0,0)[lb]{\smash{$v$}}}%
    \put(0.32363816,0.22247655){\color[rgb]{0,0,0}\makebox(0,0)[lb]{\smash{$u_1$}}}%
    \put(0.36157869,0.21807529){\color[rgb]{0,0,0}\makebox(0,0)[lb]{\smash{$u_{\min}$}}}%
    \put(0.31792832,0.29560103){\color[rgb]{0,0,0}\makebox(0,0)[lb]{\smash{$v^*$}}}%
    \put(0.5039422,0.31464538){\color[rgb]{0,0,0}\makebox(0,0)[lb]{\smash{$w$}}}%
    \put(0.51551624,0.28188028){\color[rgb]{0,0,0}\makebox(0,0)[lb]{\smash{$v$}}}%
    \put(0.53332473,0.22412977){\color[rgb]{0,0,0}\makebox(0,0)[lb]{\smash{$u_1$}}}%
    \put(0.58333117,0.22175716){\color[rgb]{0,0,0}\makebox(0,0)[lb]{\smash{$u_{\min}$}}}%
    \put(0.6640705,0.22291643){\color[rgb]{0,0,0}\makebox(0,0)[lb]{\smash{$u_{\max}$}}}%
    \put(0.78347989,0.31493077){\color[rgb]{0,0,0}\makebox(0,0)[lb]{\smash{$w$}}}%
    \put(0.80605201,0.28216566){\color[rgb]{0,0,0}\makebox(0,0)[lb]{\smash{$v$}}}%
    \put(0.81286242,0.22441513){\color[rgb]{0,0,0}\makebox(0,0)[lb]{\smash{$u_1$}}}%
    \put(0.86286887,0.22204252){\color[rgb]{0,0,0}\makebox(0,0)[lb]{\smash{$u_{\min}$}}}%
    \put(0.94360816,0.22320181){\color[rgb]{0,0,0}\makebox(0,0)[lb]{\smash{$u_{\max}$}}}%
    \put(0.86315444,0.29759919){\color[rgb]{0,0,0}\makebox(0,0)[lb]{\smash{$v^*$}}}%
    \put(0.21586681,0.14141477){\color[rgb]{0,0,0}\makebox(0,0)[lb]{\smash{$w$}}}%
    \put(0.22197616,0.10679468){\color[rgb]{0,0,0}\makebox(0,0)[lb]{\smash{$v$}}}%
    \put(0.26704371,0.06172713){\color[rgb]{0,0,0}\makebox(0,0)[lb]{\smash{$u_1$}}}%
    \put(0.29637785,0.10433769){\color[rgb]{0,0,0}\makebox(0,0)[lb]{\smash{$u_{\max}$}}}%
    \put(0.49828871,0.13656316){\color[rgb]{0,0,0}\makebox(0,0)[lb]{\smash{$w$}}}%
    \put(0.50439807,0.10194308){\color[rgb]{0,0,0}\makebox(0,0)[lb]{\smash{$v$}}}%
    \put(0.54946562,0.05687552){\color[rgb]{0,0,0}\makebox(0,0)[lb]{\smash{$u_1$}}}%
    \put(0.57879974,0.09948609){\color[rgb]{0,0,0}\makebox(0,0)[lb]{\smash{$u_{\max}$}}}%
    \put(0.56454389,0.12707925){\color[rgb]{0,0,0}\makebox(0,0)[lb]{\smash{$v^*$}}}%
    \put(0,0){\includegraphics[width=\unitlength,page=2]{fig-spur-simplification-diff-v2.pdf}}%
    \put(0.20251189,0.24095131){\color[rgb]{0,0,0}\makebox(0,0)[lb]{\smash{$\Rightarrow$}}}%
    \put(0.72728851,0.24352862){\color[rgb]{0,0,0}\makebox(0,0)[lb]{\smash{$\Rightarrow$}}}%
    \put(0.4295711,0.07473218){\color[rgb]{0,0,0}\makebox(0,0)[lb]{\smash{$\Rightarrow$}}}%
  \end{picture}%
\endgroup%
\caption{Cases in which \Lshortcut\/ is not ws-equivalent.
Top left: $P$  does not satisfy \ref{cond:B2}.
Top right: $P$ does not satisfy \ref{cond:B3}.
Bottom: the operation skips phase \ref{phase1}.
}
\label{fig:spur-simplification-diff}
\end{figure}


\begin{lemma}\label{lem:Lshortcut}
\Lshortcut\ is ws-equivalent and maintains \ref{cond:A1}--\ref{cond:A2} in $D_b$ and \ref{inv:tree}--\ref{inv:subdiv} in adjacent clusters.
\end{lemma}
\begin{proof}
Let $P_1$ be the polygon obtained from $P$ after phase~\ref{phase1} of \Lshortcut$(v,TR)$ and $P_2$ be the polygon obtained after phase \ref{phase3}.
Note that phase \ref{phase1} of the operation only creates crimps,
and it is ws-equivalent by Lemma~\ref{lem:crimp}.
Let $H$ be the set of edges $[u_1,u_2]$ of paths $[v,u_1,u_2]\in L_v^{TR}$.
Phases \ref{phase2}--\ref{phase3} are equivalent to the concatenation of the primitives: \subdivision, \nodesplit, and \merge.
Specifically, they are equivalent to subdividing every edge in $H$ into $[u_1,u_{\min},u_2]$ whenever $u_2\neq u_{\min}$,
and applying \nodesplit$(v,u_1,u_{\min})$ (which creates $u_1^*$) to $P_2$ followed by \nodesplit$(w,v,u_1^*)$ (which creates $v^*$), and merging every path $[v^*,u_1^*,u_{\min}]$ to $[v^*,u_{\min}]$.
The only primitive that may not satisfy its preconditions is \nodesplit$(v,u_1,u_{\min})$:
segment $u_1u_{\min}$ may be collinear with several segments of $b$,
and $P_2$ may contain spurs that overlap with $u_1u_{\min}$.
In the next paragraph, we show that the spurs that may overlap with $u_1u_{\min}$ do not pose a problem,
and we can essentially repeat the proof of Lemma~\ref{lem:primitives2}.

Assume that $P_1$ is weakly simple and consider a polygon $Q_1\in\Phi(P_1)$.
Due to \ref{cond:A1}--\ref{cond:A2} and phase~\ref{phase1}, every path in $L_v^{TR}$ is a sub-path
of some path $[v,u_1,u_2, u_3]$ where $x(u_3)\le x(u_2)$.
We show that $P_1$ has a perturbation in $\Phi(P_1)$ with the following property:
\begin{quote}
($\star$) Every edge $[u_1,u_2]\in H$ lies above all overlapping edges $e\notin H$.
\end{quote}
Let $Q_1\in \Phi(P_1)$ be a perturbation of $P_1$ into a simple polygon that has the minimum number of edges
$[u_1,u_2]\in H$ that violate ($\star$). We claim that $Q_1$ satisfies ($\star$).
Suppose the contrary, that $Q_1$ does not satisfy ($\star$). For a contradiction, we modify
$Q_1\in \Phi(P_1)$ and obtain another perturbation $Q_1'\in\Phi(P_1)$ that has strictly fewer
edges that violate ($\star$) as shown in Figure~\ref{fig:reshuffling}.
Recall that $Q_1$ yields a total order of edges in each segment of $b$  based on the above-below relationship.
Let $[u_1,u_2']\in H$ be the highest edge that violates ($\star$),
and assume that this edge is part of a path $[v,u_1,u_2', u_3']$.
Let $Z$ be the set of edges that are above $[u_1,u_2']$ within the corridors between $u_1$ and $u_2'$, and are not in $H$.
By \ref{cond:B2}--\ref{cond:B3} and Lemma~\ref{lem:irreducible}, every edge $[z_1,z_2]\in Z$ must be part of a path $[z_1,z_2,z_3]$ where $x(u_1)\le x(z_2)<x(u_2')\le x(z_1)$ and $x(u_2')\le x(z_3)$, otherwise $Q_1$ would not be simple.
We modify $\sigma(Q_1)$ by moving the edges in $Z$, maintaining their relative order, immediately below edge $[u_2', u_3']$ in all segments between $u_1$ and $u_2'$.
This results in a simple polygon $Q_1'\in\Phi(P_1)$ such that $[u_1,u_2']$ and all edges in $H$ above $[u_1,u_2']$ satisfy ($\star$), contradicting the choice of $Q_1$.

\begin{figure}[htbp]
\centering
	\def\svgwidth{0.9\textwidth}
\begingroup%
  \makeatletter%
  \providecommand\color[2][]{%
    \errmessage{(Inkscape) Color is used for the text in Inkscape, but the package 'color.sty' is not loaded}%
    \renewcommand\color[2][]{}%
  }%
  \providecommand\transparent[1]{%
    \errmessage{(Inkscape) Transparency is used (non-zero) for the text in Inkscape, but the package 'transparent.sty' is not loaded}%
    \renewcommand\transparent[1]{}%
  }%
  \providecommand\rotatebox[2]{#2}%
  \ifx\svgwidth\undefined%
    \setlength{\unitlength}{372.63422852bp}%
    \ifx\svgscale\undefined%
      \relax%
    \else%
      \setlength{\unitlength}{\unitlength * \real{\svgscale}}%
    \fi%
  \else%
    \setlength{\unitlength}{\svgwidth}%
  \fi%
  \global\let\svgwidth\undefined%
  \global\let\svgscale\undefined%
  \makeatother%
  \begin{picture}(1,0.15564823)%
    \put(0,0){\includegraphics[width=\unitlength,page=1]{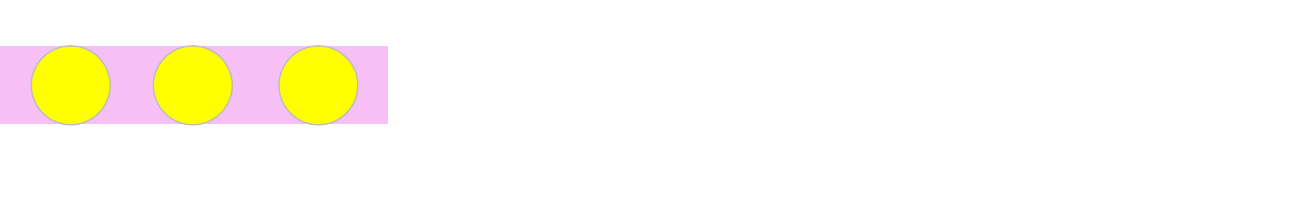}}%
    \put(0.04283601,0.04243298){\color[rgb]{0,0,0}\makebox(0,0)[lb]{\smash{$u_1$}}}%
    \put(0.22906953,0.13061408){\color[rgb]{0,0,0}\makebox(0,0)[lb]{\smash{$u_{\max}$}}}%
    \put(0,0){\includegraphics[width=\unitlength,page=2]{fig-reshuffle-v2.pdf}}%
    \put(0.39044968,0.04243298){\color[rgb]{0,0,0}\makebox(0,0)[lb]{\smash{$u_1$}}}%
    \put(0.57668321,0.13061408){\color[rgb]{0,0,0}\makebox(0,0)[lb]{\smash{$u_{\max}$}}}%
    \put(0,0){\includegraphics[width=\unitlength,page=3]{fig-reshuffle-v2.pdf}}%
    \put(0.74271802,0.04243298){\color[rgb]{0,0,0}\makebox(0,0)[lb]{\smash{$u_1$}}}%
    \put(0.92895155,0.13061408){\color[rgb]{0,0,0}\makebox(0,0)[lb]{\smash{$u_{\max}$}}}%
    \put(0,0){\includegraphics[width=\unitlength,page=4]{fig-reshuffle-v2.pdf}}%
    \put(0.15006443,0.0022426){\color[rgb]{0,0,0}\makebox(0,0)[b]{\smash{(a)}}}%
    \put(0.49767807,0.0022426){\color[rgb]{0,0,0}\makebox(0,0)[b]{\smash{(b)}}}%
    \put(0.84994646,0.0022426){\color[rgb]{0,0,0}\makebox(0,0)[b]{\smash{(c)}}}%
  \end{picture}%
\endgroup%
\caption{(a) A perturbation $Q_1$ that violates property ($\star$); the highest edge $[u_1,u_2']\in H$ that violates ($\star$) is \textcolor{red}{red}, and edges in $Z$ are \textcolor{blue}{blue}.
(b) We can modify $Q_1$ to reduce the number of edges in $H$ that violate ($\star$).
(c) There exists a perturbation $Q_1$ that satisfies ($\star$).
}
\label{fig:reshuffling}
\end{figure}

We can proceed as in the proof of Lemma~\ref{lem:primitives2}, using a perturbation $Q_1\in \Phi(P_1)$ that satisfies ($\star$)
to show that $P_2$ is weakly simple if and only if $P_1$ is weakly simple, that is, phases~\ref{phase2}--\ref{phase3} are ws-equivalent.

By construction, \ref{inv:tree}--\ref{inv:subdiv} are maintained.
Note that the intermediate polygon $P_1$ may violate condition \ref{cond:A2},
since phase~\ref{phase1} introduces crimps.
However, after phase~\ref{phase3}, conditions \ref{cond:A1} and \ref{cond:A2} are restored,
and operation \Lshortcut\/ maintains \ref{cond:A1}--\ref{cond:A2} in the ellipse $D_b$.
\end{proof}

\subsection{Bar simplification algorithm}
\label{ssec:phases}

In this section, we describe an algorithm, called {\sf bar-simplification}, that incrementally removes all spurs of the polygon $P$ from a bar $b$ using a sequence of \pinextraction, \Vshortcut, and \Lshortcut\ operations.  Informally, our algorithm ``unwinds'' each polygonal chain in the bar.
It extracts pins and V-chains whenever possible. Any other chain in $D_b$ contains edges along bar $b$, and the
sequence of these edge lengths is unimodal (cf.~Lemma~\ref{lem:irreducible}). Our algorithm ``unwinds'' these chains by a sequence of \Lshortcut\ operations. Each operation eliminates or reduces one of the shortest edges along $b$ (see Figure~\ref{fig:cross-lifecycle}). The algorithm alternates between \Lshortcut$(v,TR)$ and \Lshortcut$(v,TL)$ to unwind the chains from their top endpoints to the longest edge in $b$; and then uses \Lshortcut$(v,BR)$ and \Lshortcut$(v,BL)$ to resolve the bottom part.

When we unwind the chains in $D_b$ starting from their top vertices using \Lshortcut$(v,TR)$ and \Lshortcut$(v,TL)$, we cannot hope to remove the longest edge of a cross-chain. We stop using the operations when every path in $L_v^{TR}$ contains a longest edge of a cross-chain. This motivates the use of $\Lhat_v^{TR}$ (instead of $L_v^{TR}$) in step~\ref{step-iii} below. We continue with the algorithm and its analysis.

\begin{figure}[htbp]
\centering
\includegraphics[width=\linewidth]{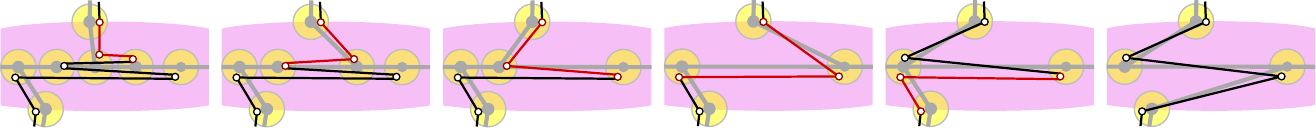}
\caption{Life cycle of a cross-chain in the while loop of {\sf bar-simplification}. The steps applied, from left to right, are: \ref{step-iii}, \ref{step-iv}, \ref{step-iii}, \ref{step-iv}, and \ref{step-vi}.}
\label{fig:cross-lifecycle}
\end{figure}

\medskip
\noindent Algorithm {\sf bar-simplification}$(P,b)$.\\
While $P$ has an edge along $b$, perform one operation as follows.
\begin{enumerate}[label=(\roman*)]\itemsep -2pt
\item\label{step-i}
	If $\Pin\neq \emptyset$, pick an arbitrary pin $[v,u,v]$ and perform \pinextraction$(u,v)$.
\item\label{step-ii}
    Else if $\VV\neq \emptyset$, then let $[v_1,u,v_2]\in \VV$ be a path where $|x(v_1)-x(v_2)|$ is minimal.
    If there is no segment $uq$ in the wedge $\angle v_1uv_2$, perform \Vshortcut$(v_1,u,v_2)$, else report that $P$ is not weakly simple and halt.
\item \label{step-iii}
    Else if there exists $v\in \partial D_b$ such that $\Lhat_v^{TR}\neq \emptyset$, do:
    \begin{enumerate}
    \item\label{step-iiiA}
    	Let $v$ be the rightmost node where $L_v^{TR}\neq \emptyset$.
    \item\label{step-iiiB}
    	If $L_{v}^{TR}$ satisfies \ref{cond:B2}--\ref{cond:B3}, do \Lshortcut$(v,TR)$.
    \item\label{step-iiiC}
    	Else let $v'$ be the leftmost node such that $x(v)<x(v')$ and $L_{v'}^{TL} \neq \emptyset$,
    or record that no such vertex $v'$ exists.
    	\begin{enumerate}
    	    \item [(c.1)]
	    	    If $v'$ does not exist,
                or $L_{v'}^{TL}$ does not satisfy \ref{cond:B2}--\ref{cond:B3},
                or any path in $L_{v'}^{TL}$ contains a longest edge of a cross-chain,
                then report that $P$ is not weakly simple and halt.	    	
    	    \item [(c.2)]
	    	    Else do \Lshortcut$(v',TL)$.	    	
		\end{enumerate}
    \end{enumerate}
\item \label{step-iv}
    Else if there exists $v\in \partial D_b$ such that $L_v^{TL}\neq \emptyset$, perform steps~\ref{step-iiiA}--\ref{step-iiiC} with left--right and $TR$--$TL$ interchanged. (Note the use of  $L_v$ instead of  $\widehat{L_v}$. The same applies to (vi) below).
\item \label{step-v}
    Else if there exists $v\in \partial D_b$ such that $\Lhat_v^{BL}\neq \emptyset$, perform steps~\ref{step-iiiA}--\ref{step-iiiC} using $BL$ and $BR$ in place of $TR$ and $TL$, respectively, and left-right interchanged.
\item \label{step-vi}
	Else if there exists $v\in \partial D_b$ such that $L_v^{BR}\neq \emptyset$, perform steps~\ref{step-iiiA}--\ref{step-iiiC} using $BR$ and $BL$ in place of $TR$ and $TL$, respectively.
\item \label{step-vii} Else invoke \oldbarexp.
\end{enumerate}
Return $P$ (end of algorithm).\\

\begin{lemma}\label{lem:simplfication-existance}
The operations performed by {\sf bar-simplification}$(P,b)$ are ws-equivalent,
and maintain properties \ref{cond:A1}--\ref{cond:A2} in $D_b$ and \ref{inv:tree}--\ref{inv:subdiv} in adjacent clusters.
The algorithm either removes all nodes from the ellipse $D_b$,
or reports that $P$ is not weakly simple.
The \Lshortcut\/ operations performed by the algorithm
create at most two crimps in each cross-chain in $\PP$.
\end{lemma}
\begin{proof}
We show that the algorithm only uses operations that satisfy their preconditions,
and  reports that $P$ is not weakly simple only when $P$ contains a forbidden configuration.

\medskip\noindent{\bf Steps~\ref{step-i}--\ref{step-ii}.}
Since every pin can be extracted from a polygon satisfying \ref{inv:tree}--\ref{inv:subdiv}, we may assume that $\Pin=\emptyset$.
Suppose that $\VV\neq \emptyset$.
Let $[v_1,u,v_2]\in \VV$ be a V-chain such that $|x(v_1)-x(v_2)|$ is minimal. Since $\Pin=\emptyset$, the only obstacle for the precondition
of \Vshortcut\ is an edge $[u,q]$ such that segment $uq$ is in the interior of the wedge $\angle v_1uv_2$ (or else the image graph would have a crossing).
If such an edge exists, it is part of a path $[p,u,q]$. The node $q$ is in $\partial D_b$ between $v_1$ and $v_2$. Note that $p\neq q$, otherwise $[p,u,q]$ would be a pin. Further, $p$ cannot be a node in the interior of the wedge $\angle v_1uv_2$,
otherwise $[p,u,q]$ would be a V-chain where $|x(p)-x(q)|<|x(v_1)-x(v_2)|$, contrary to the choice of $[v_1,u,v_2]\in \VV$.
Consequently, $p$ must be in the exterior of the wedge $\angle v_1uv_2$. In this case, the paths $[v_1,u,v_2]$ and $[p,u,q]$ form the forbidden configuration in Corollary~\ref{cor:forbidden}(1), and the algorithm correctly reports that $P$ is not weakly simple.
If no such edge $[u,q]$ exists, then \Vshortcut$(v_1,u,v_2)$ satisfies all preconditions and it is ws-equivalent by Lemma~\ref{lem:pin-and-V}.
Henceforth, we may assume that $\Pin=\emptyset$ and $\VV=\emptyset$.

\medskip\noindent{\bf Step~\ref{step-iii}--\ref{step-iv}.}
By symmetry, we consider only step~\ref{step-iii}.
Since $\Pin=\emptyset$, condition \ref{cond:B1} is met.
In step \ref{step-iiiB}, if \ref{cond:B2}-\ref{cond:B3} are also satisfied,
then \Lshortcut$(v,TR)$ is ws-equivalent by Lemma~\ref{lem:Lshortcut}.
If condition \ref{cond:B2} or \ref{cond:B3} fails, we proceed with step \ref{step-iiiC}.

\medskip\noindent{\bf Step~\ref{step-iii}(c.1).}
We show that in these cases the algorithm correctly reports that $P$ is not weakly simple.
Assume first that $v'$ does not exist. Since $L_{v}^{TR}$ does not satisfy \ref{cond:B2} or \ref{cond:B3},
there exists an edge $[p,q]$ such that $x(u_1)\le x(q)<x(u_{\max})$ and $p\in \partial D_b$ is a top node.
Edge $[p,q]$ is part of some path $[p,q,r]$. Note that $r$ cannot be a top vertex of $\partial D_b$,
since $\Pin=\emptyset$ and $\VV=\emptyset$. If $r$ is on $b$ and $x(q)<x(r)$, then $[p,q,r]\in L_{p}^{TR}$, which contradicts the choice of node $v$. If $r$ is on $b$ and $x(r)<x(q)$, then $[p,q,r]\in L_{p}^{TL}$ and $v'$ exists.
It follows that $r$ is a bottom vertex, and then the paths $[v,u_1,u_{\max}]$ and $[p,q,r]$ form a forbidden configuration
in Corollary~\ref{cor:forbidden}(1) or (3).

Assume now that $v'$ exists but $L_{v'}^{TL}$ does not satisfy \ref{cond:B2} or \ref{cond:B3}.
Let $[v',u_1',u_{\max}']$ be the path in $L_{v'}^{TL}$ with the longest edge on $b$.
By the definitions of \ref{cond:B2}--\ref{cond:B3}, $x(u_1) \le x(u'_1) < x(u_{\max})$.
If $x(u_{\max}')<x(u_1)$, then $[v,u_1,u_{\max}]$ and $[v',u_1',u_{\max}']$ form the forbidden configuration in Corollary~\ref{cor:forbidden}(2).
Else, we have $x(u_1)\leq x(u_{\max}')< x(u_1')<x(u_{\max})$.
This implies that any edge $[p,q]$ that violates \ref{cond:B2} or \ref{cond:B3} for $L_{v'}^{TL}$ must also violate \ref{cond:B2} or \ref{cond:B3} for $L_v^{TR}$.
However, this contradicts the choice of $v$ (rightmost where $L_v^{TR}\neq \emptyset$) and $v'$ (leftmost, $x(v)<x(v')$, where $L_{v'}^{TL}\neq\emptyset$).

Next assume that there is a path $[v',u_1',u_{2}']\in L_{v'}^{TL}$ such that $[u_1',u_{2}']$
is the longest edge of a cross-chain. Then this cross-chain is of the form
$[v',u_1',u_{2}',\ldots ,p']$, where all interior vertices lie on the line segment $u_1' u_2'$,
and $p'$ is a bottom vertex. Now $[v,u_1,u_{\max}]$ and this cross-chain form the forbidden configuration in Corollary~\ref{cor:forbidden}(3).
In all three cases in step \ref{step-iii}(c.1), the algorithm correctly reports that $P$ is not weakly simple.

\medskip\noindent{\bf Step~\ref{step-iii}(c.2).}
Let the path $[v',u_1',u_{\max}']\in L_{v'}^{TL}$ be selected in \Lshortcut$(v',TL)$ by the algorithm.
Since conditions \ref{cond:B1}--\ref{cond:B3} are satisfied, \Lshortcut$(v',TL)$ is ws-equivalent by Lemma~\ref{lem:Lshortcut}.

\noindent{\bf Steps~\ref{step-v}--\ref{step-vii}.}
If steps~\ref{step-i}--\ref{step-iv} do not apply, then $\Lhat_v^{TR}\cup L_v^{TL}=\emptyset$.
That is, for every path $[v,u_1,u_2]\in L^{TR}$, we have $[u_1,u_2]\in M_{cr}$.
In particular, there are no top chains. The operations in \ref{step-v}--\ref{step-vi}
do not change these properties.
Consequently, once steps~\ref{step-v}--\ref{step-vi} are executed for the first time,
steps~\ref{step-iii}--\ref{step-iv} are never executed again.
By a symmetric argument, steps~\ref{step-v}--\ref{step-vi} eliminate all paths in $\Lhat_v^{BL}\cup L_v^{BR}$.
When the algorithm reaches step~\ref{step-vii}, every edge in $b$ is necessarily in $M_{cr}$ and $L_v^{TL}\cup L_v^{BR}=\emptyset$. Consequently, by Lemma~\ref{lem:irreducible}, $b$ contains no spurs and \oldbarexp\ is ws-equivalent.
This operation eliminates all nodes in the interior of $D_b$.

\medskip\noindent {\bf Termination.}
Each \pinextraction\/ and \Vshortcut\/ operation reduces the number of vertices of $P$ within $D_b$.
Operation \Lshortcut$(v,X)$, $X\in \{TR,TL,BR,BL\}$, either reduces the number of interior vertices,
or produces a crimp if edge $[u_1,u_2]$ is a longest edge of a cross-chain.
For termination, it is enough to show that, for each cross-chain $c\in \PP$, the algorithm introduces
a crimp at most once in steps~\ref{step-iii}--\ref{step-iv}, and at most once in steps~\ref{step-v}--\ref{step-vi}.
Without loss of generality, consider step~\ref{step-iii}.

Note that step~\ref{step-iii} may apply an \Lshortcut\/ operation in two possible cases: \ref{step-iiiB} and \ref{step-iiiC}. However, an \Lshortcut\/ operation in \ref{step-iiiC} does not create crimps: \Lshortcut\/ is performed when all three conditions in \ref{step-iii}(c.1) fail.
In this case, $L_{v'}^{TR}$ does not contain any edge in $M_{cr}$, and \Lshortcut\/ does not create crimps.
We may assume that step~\ref{step-iii} creates crimps in case \ref{step-iiiB} only.

Every cross-chain remains a cross-chain in algorithm {\sf bar-simplification}: operations \pinextraction\/ and \Vshortcut\/ do not modify cross-chains; and operations \Lshortcut\/ and \oldbarexp\/ modify only the first or last few edges of a cross-chain. A longest edge of a cross-chain $c$ always connects the same two nodes in $b$ until step~\ref{step-vii} (\oldbarexp), although the \emph{number} of longest edges in $c$ may change. When \Lshortcut$(v,X)$ modifies a cross-chain, it moves its endpoint from $v\in \partial D_b$ to a nearby new node $v^*\in \partial D$. Consequently, if $L_v^X$, $X\in \{TR,TL\}$ contains the first two edges of two chains in $\PP$, then they have been modified by the same sequence of previous \Lshortcut\/ operations.

Suppose, for contradiction, that two invocations of step \ref{step-iiiB} create crimps in a cross-chain $c$, say, in operations \Lshortcut$(v_0,TR)$ and \Lshortcut$(v_2,TR)$ (see Figure~\ref{fig:termination}). The first invocation replaces $[v_0,u_1,u_2]$ with $[v_0^*,u_{\min},u_2,u_1,u_2]$ (where the edge $[u_{\min},u_2]$ may vanish if $u_{\min}=u_2$). The resulting cross-chain has two maximal longest edges, $[u_2,u_1]$ and $[u_1,u_2]$.
Since \Lshortcut\/ creates crimps only if the longest edge is unique, there must be an intermediate operation \Lshortcut$(v_1,TL)$ that removes or shortens the edge $[u_2,u_1]$, so that $[u_1,u_2]$ becomes the unique longest edge again.
When \Lshortcut$(v_1,TL)$ is performed in a step \ref{step-iv}, we have $\Lhat_v^{TR}=\emptyset$ for all top nodes $v$, and $L_{v'}^{TL}=\emptyset$ for all top nodes $v'$, $x(v')<x(v_1)$.
The steps between \Lshortcut$(v_1,TL)$ and \Lshortcut$(v_2,TR)$ modify only cross-chains whose top node is at or to the right of the top node of $c$ (\Lshortcut\/ operations move the top vertex of $c$ to the left, from $v_1$ to $v_2$ in one or more steps).
Consequently, when \Lshortcut$(v_2,TR)$ is performed in a step \ref{step-iii}, we still have $\Lhat_{v'}^{TR}=L_{v'}^{TL}=\emptyset$ for all top nodes $v'$, $x(v')<x(v_2)$.

\begin{figure}[h!tbp]
\centering
	\def\svgwidth{.8\textwidth}
\begingroup%
  \makeatletter%
  \providecommand\color[2][]{%
    \errmessage{(Inkscape) Color is used for the text in Inkscape, but the package 'color.sty' is not loaded}%
    \renewcommand\color[2][]{}%
  }%
  \providecommand\transparent[1]{%
    \errmessage{(Inkscape) Transparency is used (non-zero) for the text in Inkscape, but the package 'transparent.sty' is not loaded}%
    \renewcommand\transparent[1]{}%
  }%
  \providecommand\rotatebox[2]{#2}%
  \ifx\svgwidth\undefined%
    \setlength{\unitlength}{245.51352539bp}%
    \ifx\svgscale\undefined%
      \relax%
    \else%
      \setlength{\unitlength}{\unitlength * \real{\svgscale}}%
    \fi%
  \else%
    \setlength{\unitlength}{\svgwidth}%
  \fi%
  \global\let\svgwidth\undefined%
  \global\let\svgscale\undefined%
  \makeatother%
  \begin{picture}(1,0.15152281)%
    \put(0,0){\includegraphics[width=\unitlength,page=1]{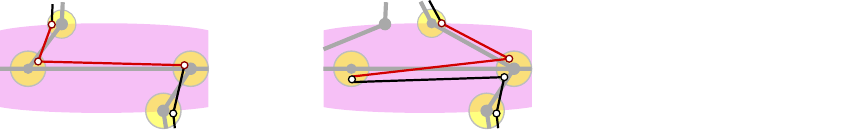}}%
    \put(0.02254864,0.1261191){\color[rgb]{0,0,0}\makebox(0,0)[lb]{\smash{$v_0$}}}%
    \put(0.52094839,0.13080363){\color[rgb]{0,0,0}\makebox(0,0)[lb]{\smash{$v_1$}}}%
    \put(0,0){\includegraphics[width=\unitlength,page=2]{fig-termination.pdf}}%
    \put(0.8564743,0.08997056){\color[rgb]{0,0,0}\makebox(0,0)[lb]{\smash{$v_2$}}}%
    \put(0.01405167,0.03437309){\color[rgb]{0,0,0}\makebox(0,0)[lb]{\smash{$u_1$}}}%
    \put(0.20337662,0.09760289){\color[rgb]{0,0,0}\makebox(0,0)[lb]{\smash{$u_2$}}}%
    \put(0.26169635,0.06848557){\color[rgb]{0,0,0}\makebox(0,0)[lb]{\smash{$\Rightarrow\ldots\Rightarrow$}}}%
    \put(0.63980371,0.06848557){\color[rgb]{0,0,0}\makebox(0,0)[lb]{\smash{$\Rightarrow\ldots\Rightarrow$}}}%
  \end{picture}%
\endgroup%
\caption{At most one crimp can be created in a cross-chain by steps~\ref{step-iiiB}.
}
\label{fig:termination}
\end{figure}

When \Lshortcut$(v_2,TR)$ is performed, we have $[v_2,u_1,u_2]\in L_{v_2}^{TR}$ but
$[v_2,u_1,u_2]\notin \Lhat_{v_2}^{TR}$ (since $u_1u_2$ is the longest edge of $c$).
Step~\ref{step-iii} is performed only if $\Lhat_p^{TR}\neq\emptyset$ for some top vertex $p$.
Since the rightmost top vertex where $L_v^{TR}\neq \emptyset$ is $v=v_2$, we have $x(p)\leq x(v_2)$.
This implies $p=v_2$. Consequently there exists a chain $c'\in \PP$ that contains a
subpath $[v_2,u_1,u_3]\in L_{v_2}^{TR}$, such that $[u_1,u_3]$ is not the longest edge of $c'$.
Since $L_{v_2}^{TR}$ contains the first two edges of both $c$ and $c'$, they have been modified by
the same sequence of \Lshortcut\/ operations. Therefore $c'$ contained $[v,u_1,u_2,u_1]$ initially.
By Lemma~\ref{lem:irreducible}, only the longest edge can repeat,
hence $[u_1,u_2]$ is the longest edge of $c'$. This implies that
$u_3=u_2$ and $\Lhat_{v_2}^{TR}=\emptyset$, contradicting the condition in Step~\ref{step-iii}.

We conclude that {\sf bar-simplification}$(P,b)$ introduces a crimp at most once in steps~\ref{step-iii}--\ref{step-iv}, and at most once in steps~\ref{step-v}--\ref{step-vi} in each cross-chain. Since all other steps decrease the number of vertices in $D_b$, the algorithm terminates, as claimed.
\end{proof}

\begin{lemma}\label{lem:simplfication-time}
Algorithm {\sf bar-simplification}$(P,b)$ takes $O(m\log m)$ time using suitable data structures,
where $m$ is the number of vertices in $b$.
\end{lemma}
\begin{proof}
Operations \pinextraction, \Vshortcut, and \Lshortcut\/  each make $O(1)$ changes in the image graph.
Operations \pinextraction\/ and \Vshortcut\/  decrease the number of vertices inside $D_b$.
Each \Lshortcut\/  does as well, except for the steps that create crimps. By Lemma~\ref{lem:Lshortcut}, \Lshortcut\/
operations may create at most $2|\PP|=O(m)$ crimps. So the total number of operations is $O(m)$.

When $[v,u_1,u_2]\in L_v^{TR}$ and $u_2\neq u_{\min}$,  \Lshortcut\, replaces $[v,u_1,u_2]$ by $[v^*,u_{\min},u_2]$: vertex $[u_1]$ shifts to $[u_2]$, but no vertex is eliminated.
In the worst case, one \Lshortcut\ modifies $\Theta(m)$ paths, so in $\Theta(m)$ operations
the total number of vertex shifts is  $\Theta(m^2)$.

\medskip\noindent {\bf Data structures.}
We maintain a cyclic list of nodes in $\partial D_b$ given by the combinatorial embedding of the image graph.
Since each operation adds a constant number of nodes to $\partial D_b$ at positions adjacent to the nodes to which the operation was applied, such a list can be maintained using $O(1)$ time per operation.
Our implementation does not maintain the paths in $\PP$ explicitly.
Instead, we use set operations. We maintain the sets $\Pin$, $\VV$,
and $L_v^X$, with $v\in \partial D_b$ and $X\in\{TR,TL,BR,BL\}$, in sorted lists.
The pins $[v,u,v]\in \Pin$ are sorted by $x(v)$;
the wedges $[v_1,u,v_2]\in \VV$ are sorted by $|x(v_1)-x(v_2)|$.
In every set $L_v^X$, the first two nodes in the paths $[v,u_1,u_2]\in L_v^X$  are the same by \ref{inv:no-spur}, and so it is enough to store vertex $[u_2]$; these vertices are stored in a list sorted by $x(u_2)$.
We also maintain binary variables to indicate for each path $[v,u_1,u_2]\in L_v^X$ whether it is part of a cross-chain, and whether $[u_1,u_2]$ is the only longest edge of that chain.

\medskip\noindent {\bf Running time analysis.}
The condition in step~\ref{step-ii} can be tested in $O(1)$ time by checking whether $uv_1$ and $uv_2$ are consecutive segments in the rotation of node $u$ in the image graph.
Steps~\ref{step-i}-\ref{step-ii} remove pins and V-chains, taking linear time in the number of removed vertices, without introducing any path in any set.

Consider \Lshortcut$(v,TR)$, executed in step~\ref{step-iii}, which can be generalized to other occurrences of the \Lshortcut\ operation performed in one of steps~\ref{step-iii}--\ref{step-vi}. Recall that $\Pin=\VV=\emptyset$.
Let $p'$ be the leftmost top vertex in $\partial D_b$ to the right of $v$, which can be found in $O(1)$ time using the cyclic list of nodes in $\partial D_b$.
By \ref{inv:deg-b}, every path  $[p',q',r']\in L_{p'}^X, X\in\{TR,TL\}$, must contain the edge $[p',q']$.
If $x(q')< x(u_{\max})$, then \ref{cond:B2} or \ref{cond:B3} are not satisfied.
Assume that $x(q')\ge x(u_{\max})$ and \ref{cond:B2} (resp., \ref{cond:B3}) is not satisfied.
Then there must exist an edge $[p,u_1]$ (resp., $[p,q]$ where $x(q)< x(u_{\max})$) such that $p$ is to the right of $p'$.
Then, segments $p'q'$ and $pu_1$ (resp., $pq$) properly cross.
This is a contradiction since no operation introduces crossings in the image graph.
Hence \ref{cond:B2}--\ref{cond:B3} are satisfied if and only if either $p'$ does not exist (i.e., $v$ is the rightmost top vertex), or $x(u_{\max})\le x(q')$; this can be tested in $O(1)$ time.
The elements $[v,u_1,u_{\min}]\in L_v^{TR}$ are simplified to $[v^*,u_{\min}]$.
Consider one of these paths, and assume that the next edge along $P$ is $[u_{\min},u_3]$.
Then, the path $[v^*,u_{\min},u_3]$ is inserted into either $\Pin\cup \VV$ if $u_3\in \partial D_b$ is a top vertex,
or $L_{v^*}^{TL}$ if $u_3\in b$. We can find each chain $[v,u_1,u_{\min}]\in L_v^{TR}$ in $O(1)$ time since $L_v^{TR}$ is sorted by $x(u_2)$. Finally, all other paths of the form $[v,u_1,u_2]\in L_v^{TR}$, where $u_2\neq u_{\min}$, become $[v^*,u_{\min},u_2]$ and they form the new set $L_{v^*}^{TR}$.
Since we store only the last vertex $[u_2]$, which is unchanged, we create $L_{v^*}^{TR}$ at no cost.

This representation allows the manipulation of $O(m)$ vertices with one set operation.
The number of insert and delete operations in the sorted lists is proportional to the number of
vertices that are removed from the interior of $D_b$, which is $O(m)$.
Each insertion and deletion takes $O(\log m)$ time, and the overall time complexity is $O(m\log m)$.
\end{proof}

\section{Spur elimination algorithm}
\label{sec:tree-exp}

After bar-simplification (Section~\ref{sec:bars}), we obtain a polygon that has no forks and every spur is at an interior node of some cluster (formed on the boundary of some ellipse $D_b)$. In the absence of forks, we can decide weak simplicity using~\cite[Theorem 5.1]{CEX15}, but a na\"{\i}ve implementation runs in $O(n^2\log n)$ time: successive applications of \spurreduction\/ would perform an operation at each dummy vertex. In this section, we show how to eliminate spurs in $O(n\log n)$ time.

\medskip\noindent{\bf Formation of Groups.} We create \emph{groups} by gluing pairs of clusters with adjacent roots together.
Recall that by \ref{inv:tree} each cluster induces a tree.
We modify the image graph, transforming each tree in a cluster into a binary tree using ws-equivalent primitives. For each node $s$ with more than two children, let $s_1$ and $s_2$ be the first two children in counterclockwise order. Create new nodes $s_1'$ and $s_2'$ by \subdivision\ in $ss_1$ and $ss_2$, respectively, and create a segment $s_1's_2$. Use the inverse of \nodesplit\ to merge nodes $s_1'$ and $s_2'$ into a node $s'$, reducing the number of children of $s$ by one.

In the course of our algorithm, an analogue of the \pinextraction\/ operation extracts a spur from one group into an ``adjacent'' group. This requires a well-defined adjacency relation between groups.
By construction, if a segment $uv$ connects nodes in different clusters, both $u$ and $v$ are leaves or both are root nodes. For every pair of clusters, $C(u)$ and $C(v)$, with adjacent roots,  $u$ and $v$, create a \emph{group} $G_{uv}=C(u)\cup C(v)$; see Figure \ref{fig:groups}. By construction, the groups are pairwise disjoint. Two groups are called \emph{adjacent} if they have two adjacent leaves in the image graph.

\begin{figure}[htbp]
\centering
	\def\svgwidth{.4\textwidth}
\begingroup%
  \makeatletter%
  \providecommand\color[2][]{%
    \errmessage{(Inkscape) Color is used for the text in Inkscape, but the package 'color.sty' is not loaded}%
    \renewcommand\color[2][]{}%
  }%
  \providecommand\transparent[1]{%
    \errmessage{(Inkscape) Transparency is used (non-zero) for the text in Inkscape, but the package 'transparent.sty' is not loaded}%
    \renewcommand\transparent[1]{}%
  }%
  \providecommand\rotatebox[2]{#2}%
  \ifx\svgwidth\undefined%
    \setlength{\unitlength}{177.11726074bp}%
    \ifx\svgscale\undefined%
      \relax%
    \else%
      \setlength{\unitlength}{\unitlength * \real{\svgscale}}%
    \fi%
  \else%
    \setlength{\unitlength}{\svgwidth}%
  \fi%
  \global\let\svgwidth\undefined%
  \global\let\svgscale\undefined%
  \makeatother%
  \begin{picture}(1,0.53029271)%
    \put(0,0){\includegraphics[width=\unitlength,page=1]{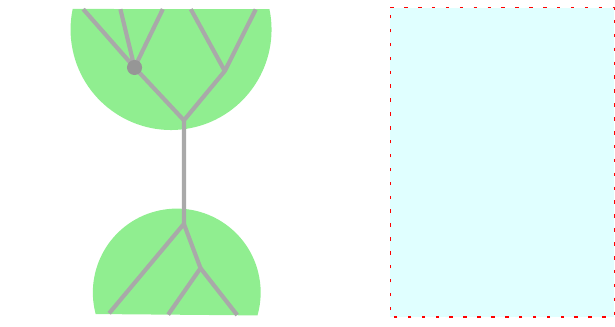}}%
    \put(0.47476046,0.24881804){\color[rgb]{0,0,0}\makebox(0,0)[lb]{\smash{⇒}}}%
    \put(0,0){\includegraphics[width=\unitlength,page=2]{fig-groups-A-v3.pdf}}%
    \put(-0.00050523,0.36624684){\color[rgb]{0,0,0}\makebox(0,0)[lb]{\smash{$C(u)$}}}%
    \put(-0.00146401,0.14609837){\color[rgb]{0,0,0}\makebox(0,0)[lb]{\smash{$C(v)$}}}%
    \put(0.66021391,0.24503706){\color[rgb]{0,0,0}\makebox(0,0)[lb]{\smash{$G_{uv}$}}}%
    \put(0.85226454,0.31865494){\color[rgb]{0,0,0}\makebox(0,0)[lb]{\smash{$u$}}}%
    \put(0.85100001,0.1541031){\color[rgb]{0,0,0}\makebox(0,0)[lb]{\smash{$v$}}}%
  \end{picture}%
\endgroup%
\caption{The formation of a group $G_{uv}$, containing clusters $C(u)$ and $C(v)$. Leaf nodes are shown as black dots.}
\label{fig:groups}
\end{figure}

Recall that a maximal path in each cluster is represented by benchmark vertices (leaves and spurs).
We denote by $[u_1;\ldots;u_k]$ (using semicolons) a maximal path inside a group defined by the benchmark vertices $u_1,\ldots,u_k$. For a given group $G_{uv}$, let $\PP$ denote the set of maximal paths with vertices in $G_{uv}$; and let $\BB$ be the set of subpaths in $\PP$ between consecutive benchmark vertices.

\begin{remark}\label{rem:groups}
By invariants \ref{inv:tree}--\ref{inv:deg}, a path in $\PP$ of a group $G_{uv}$ has alternating benchmark vertices between $C(u)$ and $C(v)$.
Consequently, every path in $\BB$ has one endpoint in $C(u)$ and one in $C(v)$,
and each spur in $G_{uv}$ is incident to two paths in $\BB$.
\end{remark}

\medskip\noindent{\bf Spur-elimination algorithm.}
Assume that $\GG$ is a partition of the nodes of the image graph into groups satisfying \ref{inv:tree}--\ref{inv:subdiv}.
We consider one group at a time, and eliminate all spurs from one cluster of that group.
When we process one group, we may split it into two groups, create a new group,
or create a new spur in an adjacent group (similar to \pinextraction\/ in Section~\ref{sec:bars}).
The latter operation implies that we may need to process a group several times.
Termination is established by showing that each operation reduces a weighted sum of the
number of benchmark vertices (i.e., spurs and boundary vertices).
Initially, the number of benchmarks is $O(n)$.

\begin{quote} Algorithm {\sf spur-elimination}$(P,\GG)$.\\
While $P$ contains a spur, do:
\begin{enumerate}\itemsep -1pt
\item\label{step:setup} Choose a group $G_{uv}\in \GG$ that contains a spur, w.l.o.g.~contained in cluster $C(u)$, and create its supporting data structures (described in Section~\ref{ssec:data} below).
\item While $T[u]$ contains an interior node, do:
\begin{enumerate}\itemsep -1pt
\item\label{step:merge} If $u$ contains no spurs and is incident to only two edges $uv$ and $uw$, eliminate $u$ with a merge operation. Rename node $w$ to $u$ which becomes the new root of the tree $T[u]$.
\item\label{step:low} If $u$ contains spurs, eliminate them as described in Section~\ref{ssec:low}.
\item\label{step:split} If $u$ contains no spurs, split $G_{uv}$ into two groups along a chain of segments that contains $uv$ as described in Section~\ref{ssec:split}. Rename a largest resulting group to $G_{uv}$.
\end{enumerate}
\end{enumerate}
\end{quote}

The detailed description of steps~\ref{step:low} and \ref{step:split} are in Sections~\ref{ssec:low} and \ref{ssec:split}, respectively. We first present supporting data structures in Section~\ref{ssec:data}, and then analyze the algorithm in Section~\ref{ssec:time}.

\subsection{Data structures}
\label{ssec:data}

In this section, we describe the data structures that we maintain for a group $G_{uv}$. We start with reviewing and introducing some notation. Consider a group $G_{uv}$ composed of two binary trees $T[u]$ and $T[v]$ rooted at $u$ and $v$, respectively. Recall that $\BB$ denotes the set of benchmark-to-benchmark paths, each with one benchmark in $T[u]$ and one in $T[v]$. In the algorithm {\sf spur-elimination}, we dynamically maintain the image trees $T[u]\cup T[v]$, and the set of paths $\BB$.
In each group $G_{uv}$, we maintain only $O(|\BB|)$ nodes that contain benchmark vertices or have degree higher than 2.
Dummy nodes of degree two that contain no benchmark vertices are redundant for the combinatorial representation, and
will be eliminated with \merge\/ operations. However, a polyline formed by a chain of dummy nodes of degree two cannot always be replaced by a straight-line segment (this might introduce unnecessary crossings). By Remark~\ref{rem:combin},
it suffices to maintain the combinatorial embeddings of the trees $T[u]$ and $T[v]$ (i.e., the counterclockwise order
of the incident segments around each node).

The partition of a group into two groups is driven by the partition of the paths in $\BB$.
For a set $\BB'\subset \BB$ of benchmark-to-benchmark paths, we define a subtree $T(\BB')$ induced by $\BB'$ as follows. Let $N=N(\BB')$ be the set of nodes that contain endpoints of some path in $\BB'$. The tree $T(\BB')$ is obtained in two steps: take the minimum subtree of $T[u]\cup T[v]$ that contains all nodes in $N$, and then merge all nodes of degree two that are not in $N$. In particular, the  nodes of $T(\BB')$ include $N$ and the lowest common ancestor of any two nodes in $N\cap C(u)$ and in $N\cap C(v)$, respectively. Denote by $\lca(r,s)$ the \emph{lowest common ancestor} of nodes $r$ and $s$ in $T[u]$ (resp., $T[v]$).

\paragraph{Description of data structures.}
For the image graph of $G_{uv}$, we maintain the following data structures.
\begin{itemize}\itemsep -1pt
\item We store trees $T[u]$ and $T[v]$ each using the dynamic data structure of~\cite{CH05}, which supports $O(1)$-time insertion and deletion of leaves, merging interior nodes of degree 2,
subdivision of edges, and lowest common ancestor queries.
\item 
Imagine that $G_{uv}$ is inside an axis-aligned rectangle with the leaves of $T[u]$ along the top edge and leaves of $T[v]$ along the bottom edge (see Figure~\ref{F:tree-cross-geom}(a)).
For each tree, we maintain a left-to-right Euler tour in an order-maintenance data structure \cite{BCD+02,DS88}, which supports insertions immediately before or after an existing item, deletions, and precedence queries, each in $O(1)$ amortized time.  For any node $w$, let $w^\flat$ and $w^\sharp$ respectively denote the first and last occurrences of $w$ in the Euler tour. Note that we have $w^\sharp=w^\flat$ for a leaf $w$. We refer to the elements of the Euler tour as \emph{tokens}.  We write $x<y$ to denote that some token $x$ occurs before (``to the left of'') another token $y$ in their common Euler tour.
\item
We also maintain the cyclic list of all leaves of the tree $T[u]\cup T[v]$ (in the order determined by the Euler tour above).
\end{itemize}

\begin{figure}[htb]
\centering
	\def\svgwidth{.7\textwidth}
\begingroup%
  \makeatletter%
  \providecommand\color[2][]{%
    \errmessage{(Inkscape) Color is used for the text in Inkscape, but the package 'color.sty' is not loaded}%
    \renewcommand\color[2][]{}%
  }%
  \providecommand\transparent[1]{%
    \errmessage{(Inkscape) Transparency is used (non-zero) for the text in Inkscape, but the package 'transparent.sty' is not loaded}%
    \renewcommand\transparent[1]{}%
  }%
  \providecommand\rotatebox[2]{#2}%
  \ifx\svgwidth\undefined%
    \setlength{\unitlength}{531.74995117bp}%
    \ifx\svgscale\undefined%
      \relax%
    \else%
      \setlength{\unitlength}{\unitlength * \real{\svgscale}}%
    \fi%
  \else%
    \setlength{\unitlength}{\svgwidth}%
  \fi%
  \global\let\svgwidth\undefined%
  \global\let\svgscale\undefined%
  \makeatother%
  \begin{picture}(1,0.72035908)%
    \put(0,0){\includegraphics[width=\unitlength,page=1]{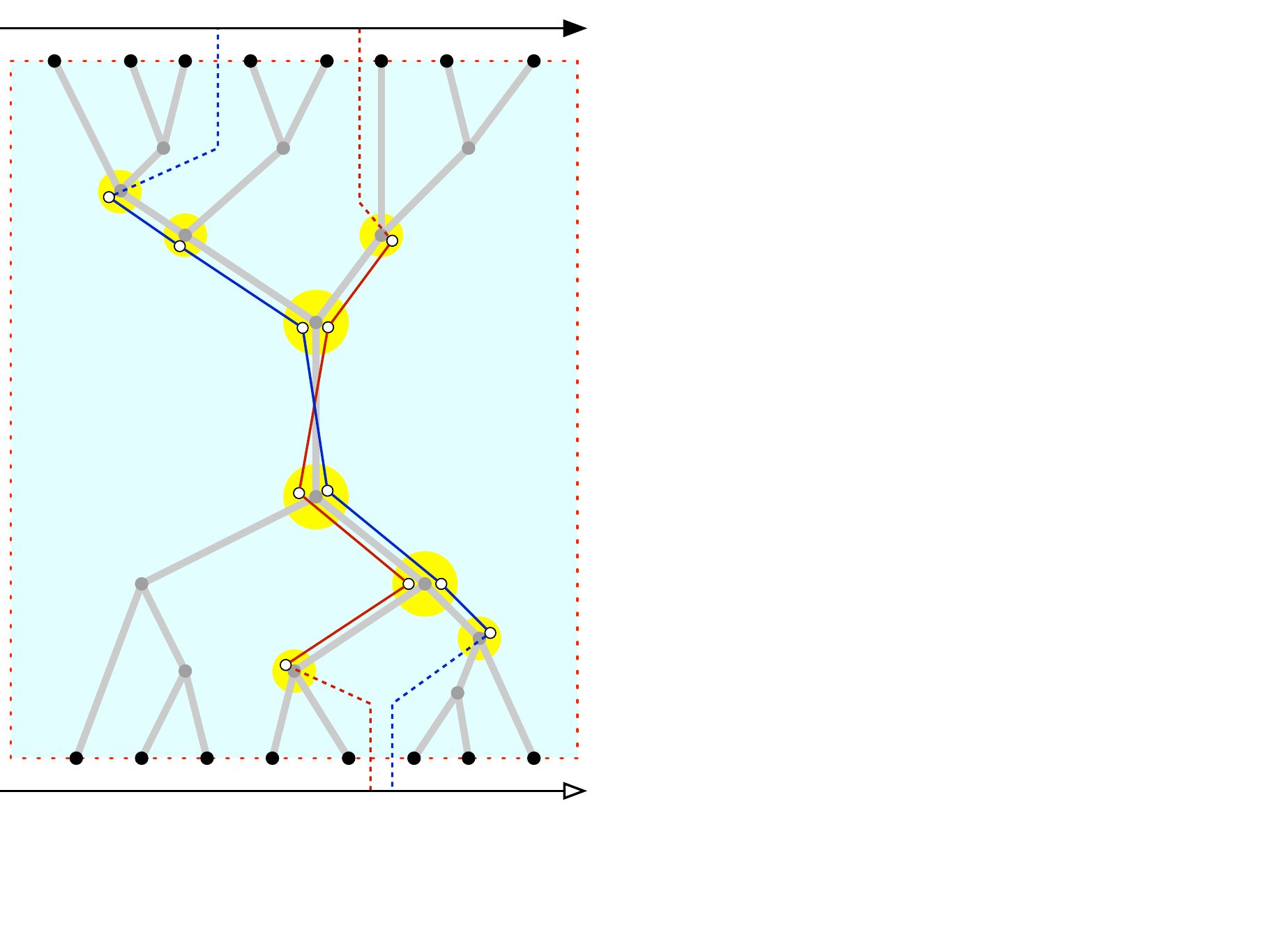}}%
    \put(0.1610408,0.71152469){\color[rgb]{0,0,0}\makebox(0,0)[lb]{\smash{$s_1^\sharp$}}}%
    \put(0.26579358,0.71139562){\color[rgb]{0,0,0}\makebox(0,0)[lb]{\smash{$s_2^\flat$}}}%
    \put(0.26149116,0.06369101){\color[rgb]{0,0,0}\makebox(0,0)[lb]{\smash{$t_2^\sharp$}}}%
    \put(0.30380727,0.06369101){\color[rgb]{0,0,0}\makebox(0,0)[lb]{\smash{$t_1^\flat$}}}%
    \put(0.05555216,0.54868515){\color[rgb]{0,0,0}\makebox(0,0)[lb]{\smash{$s_1$}}}%
    \put(0.31209519,0.5209169){\color[rgb]{0,0,0}\makebox(0,0)[lb]{\smash{$s_2$}}}%
    \put(0.20416376,0.21819207){\color[rgb]{0,0,0}\makebox(0,0)[lb]{\smash{$t_2$}}}%
    \put(0.38732331,0.23846114){\color[rgb]{0,0,0}\makebox(0,0)[lb]{\smash{$t_1$}}}%
    \put(0,0){\includegraphics[width=\unitlength,page=2]{fig-tree-paths-cross-v4.pdf}}%
    \put(0.80385162,0.06369103){\color[rgb]{0,0,0}\makebox(0,0)[lb]{\smash{$t^\sharp$}}}%
    \put(0.84616772,0.06369103){\color[rgb]{0,0,0}\makebox(0,0)[lb]{\smash{$t_{\max}^\flat$}}}%
    \put(0.59791261,0.54868518){\color[rgb]{0,0,0}\makebox(0,0)[lb]{\smash{$s$}}}%
    \put(0.74652419,0.2181921){\color[rgb]{0,0,0}\makebox(0,0)[lb]{\smash{$t$}}}%
    \put(0.92968374,0.23846117){\color[rgb]{0,0,0}\makebox(0,0)[lb]{\smash{$t_{\max}$}}}%
    \put(0,0){\includegraphics[width=\unitlength,page=3]{fig-tree-paths-cross-v4.pdf}}%
    \put(0.66098733,0.22305777){\color[rgb]{0,0,0}\makebox(0,0)[lb]{\smash{$t_{\min}$}}}%
    \put(0.73389935,0.06369103){\color[rgb]{0,0,0}\makebox(0,0)[lb]{\smash{$t^\flat$}}}%
    \put(0.65108804,0.06369103){\color[rgb]{0,0,0}\makebox(0,0)[lb]{\smash{$t_{\min}^\sharp$}}}%
    \put(0.22849646,0.00230949){\color[rgb]{0,0,0}\makebox(0,0)[b]{\smash{(a)}}}%
    \put(0.77085688,0.00230949){\color[rgb]{0,0,0}\makebox(0,0)[b]{\smash{(b)}}}%
  \end{picture}%
\endgroup%
\caption{
The geometry of crossing benchmark-to-benchmark paths.
(a) Paths $[s_1;t_1]$ and $[s_2;t_2]$ cross.
(b) If $t_{\min}^\sharp< t^\flat \leq  t^\sharp < t_{\max}^\flat$,
    then any benchmark-to-benchmark path $[s;t]$ crosses path $[t_{\min};t_{\max}]$.
}
\label{F:tree-cross-geom}
\end{figure}
We now describe data structures for $\PP$ and $\BB$. For every benchmark-to-benchmark path $[s;t]\in \BB$, we assume that $s$ is in $T[u]$ and $t$ is in $T[v]$.
A path $[s;t]$ is associated with the intervals $[s^\flat,s^\sharp]$ and $[t^\flat,t^\sharp]$.
For two consecutive benchmark-to-benchmark paths $[s_1;t;s_2]$, where $t$ is in $T[v]$,
we define the interval $I[s_1;t;s_2]=[s_1^\flat,s_2^\flat]$.

\begin{itemize}\itemsep -1pt
\item The set of benchmark-to-benchmark paths $[s;t]\in \BB$ is stored in four lists,
      sorted by $s^\flat$, $s^\sharp$, $t^\flat$, and $t^\sharp$, respectively, with ties broken arbitrarily.
      The sorted lists can be computed in $O(|\BB|)$ time by an Eulerian traversal of the tree.
\item For each node $s$ of $T[u]$, let $\BB_s$ denote the set of paths $[s;t] \in \BB$.
     We store $\BB_s$ in two lists, sorted by $t^\flat$ and $t^\sharp$, respectively.
\item We use a centered \emph{interval tree}~\cite{DutchBook} for all $O(n)$ intervals $I[s_1;t;s_2]$ that can report, for a query node $q$, all intervals containing $q$ in output-sensitive $O(\log n+k)$ time, where $k$ is the number of intervals that contain $q$. Since the interval endpoints $s^\flat$ are already sorted, the interval tree can be constructed in $O(|\BB|)$ time. The interval tree can handle the deletion of an interval in $O(1)$ time (without re-balancing, hence maintaining the $O(\log n+k)$ query time).
\end{itemize}

All data structures described in this section can be constructed in $O(|\BB|)$ preprocessing time.

\medskip\noindent{\bf Crossing paths.}
The data structure described above can determine in $O(1)$ time whether two paths in $\BB$ cross. Straightforward case analysis implies the following characterization of path crossings (refer to Figure~\ref{F:tree-cross-geom}(a)).

\begin{lemma}\label{lem:path-crossing}
Let $s_1$ and $s_2$ be arbitrary nodes in tree $T[u]$, and let $t_1$ and $t_2$ be arbitrary nodes in $T[v]$.
Paths $[s_1;t_1]$ and $[s_2;t_2]$ cross if and only if either
(1) $s_1^\sharp < s_2^\flat$ and $t_1^\flat > t_2^\sharp$, or
(2) $s_2^\sharp < s_1^\flat$ and $t_2^\flat > t_1^\sharp$.
\end{lemma}

\subsection{Eliminating spurs from a root}
\label{ssec:low}

We describe step~\ref{step:low} of Algorithm~{\sl spur-elimination}. Suppose that the root node $u$ contains a spur. The following operation eliminates all spurs from $u$, but the resulting cluster $C(v)$ need not satisfy \ref{inv:max-path} and \ref{inv:deg}, and we need to perform other operations to restore these properties.
Refer to Figure~\ref{fig:spur-shortcut}(a)--(b) for an example.
\begin{quote}{\spurshortcut}$(u)$.
Assume that $G_{uv}$ satisfies invariants \ref{inv:tree}--\ref{inv:subdiv}, and $u$ contains a spur.
Replace every path $[t_1;u;t_2]$ by $[t_1;t_2]$.
Let $\calS$ be the set of all such modified paths.
\end{quote}

\begin{lemma}
\spurshortcut\  is ws-equivalent and maintains properties \ref{inv:tree} and \ref{inv:subdiv}.
\end{lemma}
\begin{proof}
The operation is equivalent to a sequence of \spurreduction\/ operations:
First perform \spurreduction$(v,u)$. In a BFS traversal of all nodes $x$ of $T[v]$, except for the root, perform  \spurreduction$(x,{\rm parent}(x))$.
All these operations satisfy \spurreduction's constraints.
Initially, every path through the node $x$ has an edge in the segment $x\,{\rm parent}(x)$, by \ref{inv:max-path}.
The BFS traversal ensures that this property still holds when the algorithm performs \spurreduction$(x,{\rm parent}(x))$.
\end{proof}

Note that for every path $[t_1;u;t_2]$, both $t_1$ and $t_2$ are in $T[v]$ (cf.~Remark~\ref{rem:groups}) and path $[t_1;t_2]$ is uniquely defined by \ref{inv:tree}.
However, a maximal path in $C(v)$ that contains $[t_1;t_2]$ violates \ref{inv:max-path},
and if $t_1=t_2$ is a leaf in $C(v)$, then it forms a spur that may violate \ref{inv:deg}.
We proceed with a sequence of ``repair'' steps to restore them,
after which the total number of benchmark vertices decreases by at least $|\cal S|$.
The following three steps restore \ref{inv:max-path} and \ref{inv:deg} when $t_1$ and $t_2$ are in ancestor-descendent relation, that is, $\lca(t_1,t_2)\in \{t_1,t_2\}$. Let $\min(t_1,t_2)$ denote the node in $\{t_1,t_2\}$ farther from the root.

\begin{quote}
For every path $[t_1;t_2]\in \calS$, do
\begin{enumerate}
\item\label{step:trivi1} If $\lca(t_1,t_2)\in \{t_1,t_2\}$  and $t_1\neq t_2$,
    then replace $[t_1;t_2]$ with $[\min(t_1,t_2)]$.
\item\label{step:trivi2} If $t_1=t_2$ and $t_1$ is not a leaf of $T[v]$ that has degree two in the image graph, then replace $[t_1;t_2]$ with $[t_1]$.
\item\label{step:pin} If $t_1=t_2$ and $t_1$ is a leaf of $T[v]$ that has degree two in the image graph, then do: by \ref{inv:deg}, node $t_1$ is adjacent to a unique node $z\notin G_{uv}$ and $z$ is incident to a single segment in the cluster containing $z$. Subdivide such segment creating a new node $z^*$ (added to the cluster containing $z$), 
and replace every path $[z,t_1,z]$ with $[z^*]$.
    See Figure~\ref{fig:spur-shortcut}(b)--(c) for an example.
\end{enumerate}
\end{quote}

\begin{figure}[htbp]
\centering
	\def\svgwidth{.95\textwidth}
\begingroup%
  \makeatletter%
  \providecommand\color[2][]{%
    \errmessage{(Inkscape) Color is used for the text in Inkscape, but the package 'color.sty' is not loaded}%
    \renewcommand\color[2][]{}%
  }%
  \providecommand\transparent[1]{%
    \errmessage{(Inkscape) Transparency is used (non-zero) for the text in Inkscape, but the package 'transparent.sty' is not loaded}%
    \renewcommand\transparent[1]{}%
  }%
  \providecommand\rotatebox[2]{#2}%
  \ifx\svgwidth\undefined%
    \setlength{\unitlength}{502.22407227bp}%
    \ifx\svgscale\undefined%
      \relax%
    \else%
      \setlength{\unitlength}{\unitlength * \real{\svgscale}}%
    \fi%
  \else%
    \setlength{\unitlength}{\svgwidth}%
  \fi%
  \global\let\svgwidth\undefined%
  \global\let\svgscale\undefined%
  \makeatother%
  \begin{picture}(1,0.42798625)%
    \put(0,0){\includegraphics[width=\unitlength,page=1]{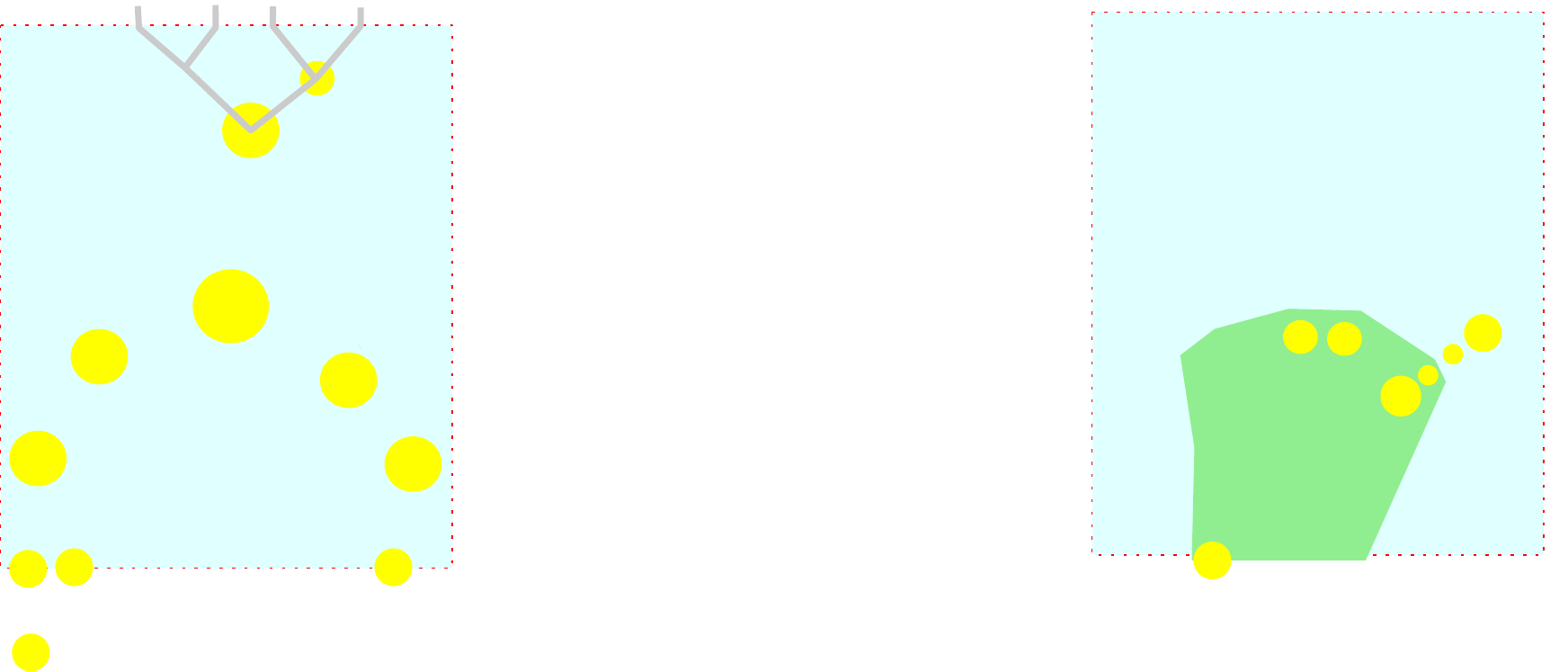}}%
    \put(0.30354396,0.24779177){\color[rgb]{0,0,0}\makebox(0,0)[lb]{\smash{⇒}}}%
    \put(0,0){\includegraphics[width=\unitlength,page=2]{fig-spur-shortcut-v4.pdf}}%
    \put(0.65485311,0.24905416){\color[rgb]{0,0,0}\makebox(0,0)[lb]{\smash{⇒}}}%
    \put(0,0){\includegraphics[width=\unitlength,page=3]{fig-spur-shortcut-v4.pdf}}%
    \put(0.1300881,0.02001139){\color[rgb]{0,0,0}\makebox(0,0)[lb]{\smash{(a)}}}%
    \put(0.49189714,0.02001139){\color[rgb]{0,0,0}\makebox(0,0)[lb]{\smash{(b)}}}%
    \put(0.83863079,0.02001139){\color[rgb]{0,0,0}\makebox(0,0)[lb]{\smash{(c)}}}%
    \put(0.03563383,0.00935967){\color[rgb]{0,0,0}\makebox(0,0)[lb]{\smash{$z$}}}%
    \put(0.38550616,0.00915192){\color[rgb]{0,0,0}\makebox(0,0)[lb]{\smash{$z$}}}%
    \put(0.73400998,0.03507422){\color[rgb]{0,0,0}\makebox(0,0)[lb]{\smash{$z$}}}%
    \put(0.73223334,0.00886862){\color[rgb]{0,0,0}\makebox(0,0)[lb]{\smash{$z^*$}}}%
    \put(0.44141069,0.17011298){\color[rgb]{0,0,0}\makebox(0,0)[lb]{\smash{$\color{violet}T[\ell]$}}}%
    \put(0.01804386,0.29040192){\color[rgb]{0,0,0}\makebox(0,0)[lb]{\smash{$G_{uv}$}}}%
    \put(0.37692719,0.29040192){\color[rgb]{0,0,0}\makebox(0,0)[lb]{\smash{$G_{uv}$}}}%
    \put(0.72476528,0.29040192){\color[rgb]{0,0,0}\makebox(0,0)[lb]{\smash{$G_{uv}$}}}%
    \put(0.12087177,0.34105828){\color[rgb]{0,0,0}\makebox(0,0)[lb]{\smash{$u$}}}%
    \put(0.47011801,0.34105828){\color[rgb]{0,0,0}\makebox(0,0)[lb]{\smash{$u$}}}%
    \put(0.81936426,0.34105828){\color[rgb]{0,0,0}\makebox(0,0)[lb]{\smash{$u$}}}%
    \put(0.17380485,0.24203508){\color[rgb]{0,0,0}\makebox(0,0)[lb]{\smash{$v$}}}%
    \put(0.52574254,0.24049938){\color[rgb]{0,0,0}\makebox(0,0)[lb]{\smash{$v=\ell$}}}%
    \put(0.87067448,0.26748229){\color[rgb]{0,0,0}\makebox(0,0)[lb]{\smash{$v=\ell$}}}%
    \put(0.38589004,0.03837291){\color[rgb]{0,0,0}\makebox(0,0)[lb]{\smash{$t_{\min}$}}}%
    \put(0.58239411,0.20465139){\color[rgb]{0,0,0}\makebox(0,0)[lb]{\smash{$t_{\max}$}}}%
    \put(0.81767931,0.19037051){\color[rgb]{0,0,0}\makebox(0,0)[lb]{\smash{$\ell^-$}}}%
    \put(0.84504052,0.19037051){\color[rgb]{0,0,0}\makebox(0,0)[lb]{\smash{$\ell^+$}}}%
    \put(0.84095821,0.1322675){\color[rgb]{0,0,0}\makebox(0,0)[lb]{\smash{$G_{\ell^-\ell^+}$}}}%
  \end{picture}%
\endgroup%
\caption{
(a) Node $u$ contains spurs.
(b) After eliminating spurs, $T[v]$ does not satisfy \ref{inv:max-path}.
(c) The analogues of \pinextraction\/ and  \Vshortcut\/.}
\label{fig:spur-shortcut}
\end{figure}
These steps restore \ref{inv:deg} at all leaves, and \ref{inv:max-path} for the affected paths $[t_1;t_2]\in \calS$. Note that these steps are ws-equivalent:
Steps~\ref{step:trivi1}--\ref{step:trivi2} do not modify the polygon (they change only the benchmarks); and step~\ref{step:pin} is analogous to  \pinextraction$(t_1,z)$.

We are left with paths $[t_1;t_2]\in \calS$ where $t_1$ and $t_2$ are in different branches of $T[v]$. In this case, we perform an elaborate version of the \Vshortcut\/ operation, that creates a new group. For every node $\ell$ of $T[v]$, let $\calS_\ell$ be the set of paths $[t_1;t_2]\in \calS$ such that $\lca(t_1,t_2)=\ell$. Consider every node $\ell$ of $T[v]$ where $\calS_\ell\neq \emptyset$ in a bottom-up traversal of $T[v]$; and create a new group $G_{\ell^-\ell^+}$ as follows
(refer to Figure~\ref{fig:spur-shortcut}).

Let $N^-$ (resp., $N^+$) be the set of nodes $t_1$ (resp., $t_2$) such that there is a path
$[t_1;t_2]\in \calS_\ell$, and $t_1$ is in the left (resp., right) subtree of $\ell$.
Let $N=N^-\cup N^+$.
Sort the nodes $t_1\in N^-$ by $t_1^\sharp$, and let $t_{\min}$ be the minimum node;
and similarly sort the nodes $t_2\in N^+$ by $t_2^\flat$, and let $t_{\max}$ be the maximum node.
The following lemma shows that interior nodes of the path from $t_{\min}$ to $\ell$ in $T[v]$ have no right branches, and the interior nodes of the path from $t_{\max}$ to $\ell$ have no left branches.

\begin{lemma}\label{lem:group-crossing}
If there is a path $[s;t]\in \BB \setminus \calS_\ell$ such that
$t_{\min}^\sharp< t^\flat \leq t^\sharp< t_{\max}^\flat$,
then it crosses some path in $\calS_\ell$, hence $P$ is not weakly simple.
\end{lemma}
\begin{proof}
Let $C$ be the path between $t_{\min}$ and $t_{\max}$ in $T[v]$.
Refer to Figure~\ref{F:tree-cross-geom}(b).
By the choice of $\ell$ (in a bottom-up traversal of $T[v]$),
we have $\calS_{\ell'}=\emptyset$ for all descendants of $\ell$.
Path $[s;t]$ reaches $C$ at some interior node $t^*\in C$,
and then continues to $\ell$, and farther to ${\rm parent}(\ell)$.
If $t^*$ is in a left (resp., right) subtree of $\ell$, then $[s;t]$ crosses every
path in $\calS_\ell$ that starts at $t_{\min}$ (resp., ends at $t_{\max})$.
\end{proof}
We can find the set $N'$ of nodes $t$ such that $t_{\min}^\sharp< t^\flat \leq t^\sharp< t_{\max}^\flat$, in $O(|N'|+\log n)$ time, by binary search in the list of leaves to find all the leaves between $t_{\min}$ and $t_{\max}$, and by lowest common ancestor queries to find nodes in $N'$. The algorithm reports that the input polygon is not weakly simple and halts if some node in $N'$ has a path satisfying the condition in Lemma~\ref{lem:group-crossing}. We can now assume that $N'\subset N$.
The nodes in $N$ induce a binary tree, denoted $T[\ell]$, of size at most $2|N|$:
its nodes are all nodes in $N$ and the lowest common ancestors of consecutive nodes
in $N^-$ and $N^+$ respectively.
Note that a segment of $T[\ell]$ might not correspond to a segment of $T[v]$ (see Figure~\ref{fig:spur-shortcut}(b)).
Denote by $C^*$ the path between $t_{\min}$ and $t_{\max}$ in $T[\ell]$.

We now define the changes in the image graph.
Every node $t\in N\setminus C^*$ is deleted from $G_{uv}$, and added to the new group.
Create two nodes, $\ell^-$ and $\ell^+$, in $G_{\ell^- \ell^+}$ sufficiently close to $\ell$ in the wedge between the two children of $\ell$, and connect them by a segment $\ell^- \ell^+$. Duplicate each node $t\in C^*\setminus \{\ell\}$, by creating a node $t'$ (added to $G_{\ell^- \ell^+}$) sufficiently close to $t$, and add a segment $tt'$.
Subdivide every segment $tt'$ with two new boundary nodes, $t_{\rm leaf}$ (added to $T[v]$) and $t_{\rm leaf}'$ (added to $G_{\ell^- \ell^+}$).
The nodes $t$ or $t'$ might now have degree 4.
Adjust the image graph so that the group trees are binary.
Finally partition the nodes in $G_{\ell^- \ell^+}$ into two trees, $T[\ell^-]$ and $T[\ell^+]$, rooted at $ \ell^-$ and $\ell^+$, respectively.

We now define the changes in the polygon.
Replace every path $[t;t_1]\in \calS_\ell$, where $t\in C^*$,
by $[t';t_1]$ if it is adjacent to a path $[t;t_2]\in \calS_\ell$, i.e., replacing the path $[t_1;t;t_2]$ by $[t_1;t';t_2]$.
Otherwise, replace $[t;t_1]$ by $[t_{\rm leaf};t'_{\rm leaf}; t_1]$.
Now we can build $\BB'$ as the set of benchmark-to-benchmark paths $[t_1';t_2']$ where $t_1',t_2'\in G_{\ell^- \ell^+}$ in $O(|\BB'|)$ time.

To prove ws-equivalence, we consider the changes in the polygon.
These amount to a sequence of ws-equivalent primitives:
a \nodesplit\/ at $\ell$, a sequence of \nodesplit s
along the chain $C$ from $\ell$ to $t_{\min}$ and $t_{\max}$, respectively,
\subdivision\ operations that create the new leaf nodes,
and \merge\ operations at degree two nodes that no longer contain spurs.
The creation of new groups takes $O(|\calS_\ell|+\log n)$ time
and $O(|\calS_\ell|)$ paths in $\BB$ are removed or modified in $G_{uv}$. Thus the data structures for $G_{uv}$ are updated in $O(|\calS_\ell|\log n)$ time. Overall, operation {\spurshortcut}$(u)$ and the repair steps that follow take $O(|\calS|\log n)$ time.

\subsection{Splitting a group in two}
\label{ssec:split}

In this section we describe step~\ref{step:split} of Algorithm {\sf spur-elimination}$(P,\GG)$.
Assume that $G_{uv}$ satisfies invariants \ref{inv:tree}--\ref{inv:subdiv}
and $u$ contains no spur.

Denote the left and right child of $u$ by $u^-$ and $u^+$, respectively.
Let $\BB^-,\BB^+\subset \BB$, resp., be the set of benchmark-to-benchmark paths that contain $u^-$ and $u^+$.
We split $G_{uv}$ into two groups induced by $\BB^-$ and $\BB^+$, respectively.
Refer to Figure~\ref{fig:tree-splitting}.

It would be easy to compute the groups induced by $\BB^-$ and $\BB^+$ in $O(|\BB|)$ time.
However, for an overall $O(n\log n)$-time algorithm, we can afford $O(\min(|\BB^-|,|\BB^+|))$
time for the split operation, and an additional $O(\log n)$ time for each eliminated spur
and each node that we split into two nonempty nodes.
Without loss of generality, we may assume $|\BB^-|\leq |\BB^+|$.
The group induced by $|\BB^-|$ can be computed from scratch in $O(|\BB^-|)$ time,
and we construct the group for $\BB^+$ by modifying $G_{uv}$, and updating
the corresponding data structures.

\begin{figure}[htbp]
\centering
	\def\svgwidth{\textwidth}
\begingroup%
  \makeatletter%
  \providecommand\color[2][]{%
    \errmessage{(Inkscape) Color is used for the text in Inkscape, but the package 'color.sty' is not loaded}%
    \renewcommand\color[2][]{}%
  }%
  \providecommand\transparent[1]{%
    \errmessage{(Inkscape) Transparency is used (non-zero) for the text in Inkscape, but the package 'transparent.sty' is not loaded}%
    \renewcommand\transparent[1]{}%
  }%
  \providecommand\rotatebox[2]{#2}%
  \ifx\svgwidth\undefined%
    \setlength{\unitlength}{695.91513672bp}%
    \ifx\svgscale\undefined%
      \relax%
    \else%
      \setlength{\unitlength}{\unitlength * \real{\svgscale}}%
    \fi%
  \else%
    \setlength{\unitlength}{\svgwidth}%
  \fi%
  \global\let\svgwidth\undefined%
  \global\let\svgscale\undefined%
  \makeatother%
  \begin{picture}(1,0.29385472)%
    \put(0,0){\includegraphics[width=\unitlength,page=1]{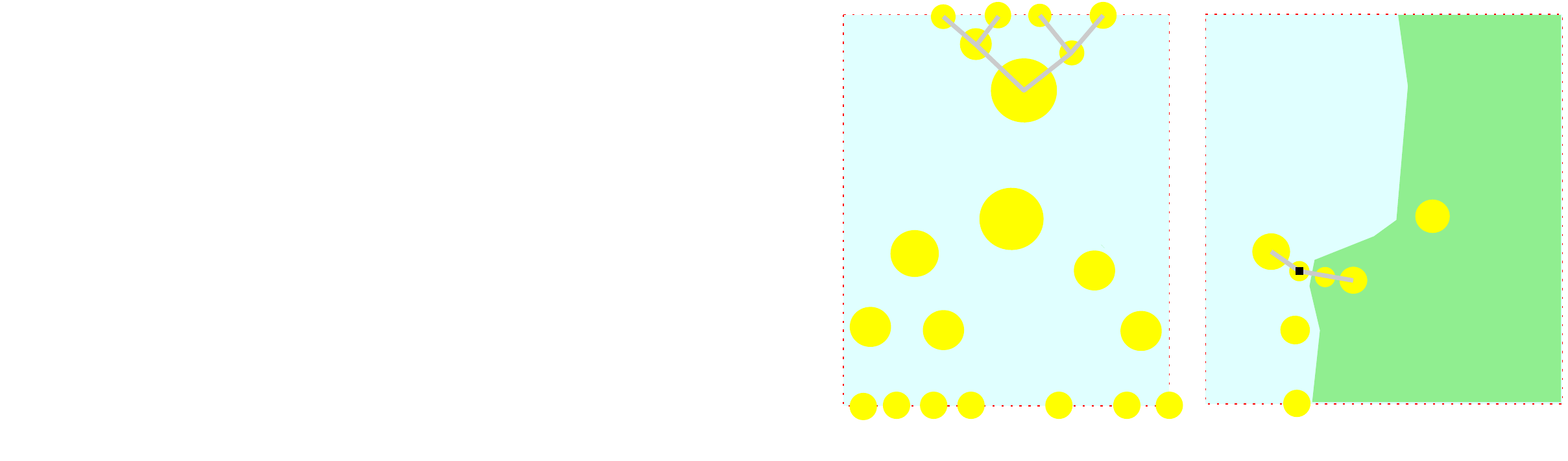}}%
    \put(0.74565438,0.16607432){\color[rgb]{0,0,0}\makebox(0,0)[lb]{\smash{⇒}}}%
    \put(0,0){\includegraphics[width=\unitlength,page=2]{fig-tree-splitting-v4.pdf}}%
    \put(0.20886988,0.16716439){\color[rgb]{0,0,0}\makebox(0,0)[lb]{\smash{⇒}}}%
    \put(0,0){\includegraphics[width=\unitlength,page=3]{fig-tree-splitting-v4.pdf}}%
    \put(0.20661289,0.00515217){\color[rgb]{0,0,0}\makebox(0,0)[lb]{\smash{(a)}}}%
    \put(0.74707625,0.00594656){\color[rgb]{0,0,0}\makebox(0,0)[lb]{\smash{(b)}}}%
    \put(0.05611281,0.24780407){\color[rgb]{0,0,0}\makebox(0,0)[lb]{\smash{$u^-$}}}%
    \put(0.08817504,0.22819802){\color[rgb]{0,0,0}\makebox(0,0)[lb]{\smash{$u$}}}%
    \put(0.06488288,0.10828481){\color[rgb]{0,0,0}\makebox(0,0)[lb]{\smash{$\color{violet}C_0$}}}%
    \put(0.09216714,0.15841171){\color[rgb]{0,0,0}\makebox(0,0)[lb]{\smash{$v$}}}%
    \put(0.15339811,0.24811665){\color[rgb]{0,0,0}\makebox(0,0)[lb]{\smash{$u^+$}}}%
    \put(0.4408261,0.24289201){\color[rgb]{0,0,0}\makebox(0,0)[lb]{\smash{$u^+$}}}%
    \put(0.28704883,0.2481639){\color[rgb]{0,0,0}\makebox(0,0)[lb]{\smash{$u^-$}}}%
    \put(0.30969976,0.15753221){\color[rgb]{0,0,0}\makebox(0,0)[lb]{\smash{$v^-$}}}%
    \put(0.40874405,0.15381146){\color[rgb]{0,0,0}\makebox(0,0)[lb]{\smash{$v^+$}}}%
    \put(0.55381074,0.18014395){\color[rgb]{0,0,0}\makebox(0,0)[lb]{\smash{$G_{uv}$}}}%
    \put(0.9308036,0.1639601){\color[rgb]{0,0,0}\makebox(0,0)[lb]{\smash{$G_{u^+v^+}$}}}%
    \put(0.78027088,0.1754372){\color[rgb]{0,0,0}\makebox(0,0)[lb]{\smash{$G_{u^-v^-}$}}}%
  \end{picture}%
\endgroup%
\caption{
Splitting group $G_{uv}$.
(a) Changes in the image graph.
(b) Changes in the polygon.}
\label{fig:tree-splitting}
\end{figure}

First, we find $\BB^-$ and $\BB^+$. Compute $\BB^-$ using the list of paths $[s;t]\in \BB$ sorted by $s^\sharp$ or $s^\flat$. Since both lists naturally split into corresponding lists for $\BB^-$ and $\BB^+$, we can split these lists in $O(\min(|\BB^-|,|\BB^+|))=O(|\BB^-|)$ time. To construct the list of $\BB^+$ sorted by $t^\sharp$ and $t^\flat$,
we start with the corresponding lists for $\BB$, and delete all elements of $\BB^-$ in $O(|\BB^-|)$ time. To compute the lists sorted by $t^\sharp$ and $t^\flat$ for $\BB^-$, we shall first compute the subtree $T[v^-]$ induced by $\BB^-$.
However, we can already find the maximum $t^\sharp$ of a path $[s;t]\in \BB^-$ in $O(|\BB^-|)$ time.

Next, we test for crossings between the paths in $\BB^-$ and the paths in $\BB^+$.
Let $t_-^\sharp$ be the maximum $t^\sharp$ of a path $[s;t]\in \BB^-$, and $t_+^\flat$ the minimum $t^\flat$ of a path $[s;t]\in \BB^+$. By Lemma~\ref{lem:path-crossing}, there is such a crossing if and only if $t_+^\flat<t_-^\sharp$, which can be determined in $O(1)$ time using our order-maintenance structures. If a crossing is detected, the algorithm reports that $P$ is not weakly simple and halts.

Trees $T[u^-]$ and $T[u^+]$ are simple subtrees of $T[u]$; but splitting $T[v]$ is nontrivial.
We use binary search in the Eulerian cycle of all leaves to find the rightmost leaf $\ell_0$ in $T[v]$ such that $\BB_{\ell_0}\cap \BB^- \neq \emptyset$, if such a leaf exists, otherwise the leftmost leaf $\ell_0$ in $T[v]$.
Let $C_0=[\ell_0;u]$. We do not compute the path $C_0$ explicitly, as it may contain more than $O(|\BB^-|)$ nodes,
but we can test whether a query node $t$ of $T[v]$ is in $C_0$ in $O(1)$ time by checking whether $\lca(\ell_0,t)=t$.
Since the paths in $\BB^-$ and $\BB^+$ do not cross, all nodes of $T[v^-]$ are
in or to the left of the chain $C_0$, and all common nodes of $T[v^-]$ and $T[v^+]$ are in $C_0$.
The image graph of $T[v^-]$ can be computed from scratch using $\BB^-$ in $O(\BB^-)$ time.
Replace each node $t$ of $T[v^-]$ that is in $C_0$ by a duplicate copy $t^-$ located sufficiently
close to $t$, to the right of $t$.
The tree $T[v^+]$ is computed from $T[v]$ by node deletion and merge operations as follows.
First delete all nodes that are in $T[v^-]$ but not in $C_0$.
For every node $t$ of $T[v^-]$ that lies in $C_0$, if $t$ has degree two in $T[v^-]$ and $\BB^+_t=\emptyset$,
then it would be a degree two node in $T[v^+]$ with no spurs, and so we can delete $t$ by merging its two
incident segments.
Let $v^+$ be the node not in $T[u^+]$ adjacent to $u^+$.
The resulting $T[v^+]$ becomes a tree induced by $\BB^+$.
It remains to resolve the connections between trees.

Let $\VV^0$ denote the set of chains $[s_1;t;s_2]$ such that $[s_1;t]\in \BB^-$ and $[t;s_2]\in \BB^+$.
The spurs at $t$ on all chains $[s_1;t;s_2]\in \VV^0$ will be eliminated
(they will become adjacent leaves in the two resulting groups).
$\VV^0$ can be found with a query for $u$ in the interval tree.
Let $N^0$ be the set of all nodes $t$ such that $[s_1;t;s_2]\in\VV^0$.
Each node $t\in N^0$ is in $C_0$ and, therefore, has a copy $t^-$ in $T[v^-]$.
Create a segment between $t$ and $t^-$, and subdivide the segment $t^- t$ with
two new nodes $t_{\rm leaf}^-$ and $t_{\rm leaf}$ in $T[v^-]$ and $T[v^+]$, respectively.
The degree of nodes $t$ or $t^-$ might increase to 4; and
so we adjust the image graphs so that both trees are binary.
The image graph is now split into groups $G_{u^-v^-}$ and $G_{u^+v^+}$.

We next define the changes in the polygon.
Replace every chain $[s_1;t;s_2]\in\VV^0$ with a new chain $[s_1;t^-_{\rm leaf};t_{\rm leaf};s_2]$,
while also replacing the corresponding paths in the lists $\BB^-$ and $\BB^+$ in $O(|\VV^0|)$ time.
In the sorted lists for $\BB^-$ and $\BB^+$, this is done by deletions and reinsertions.
Note that all leaves $t_{\rm leaf}^-$ (resp., $t_{\rm leaf}$) are at the end (resp., beginning)
of the Euler tour of $T[v^-]$ (resp., $T[v^+]$),
so deletions can be performed in $O(|\VV^0|)$ time;
and insertions take $O(|\VV^0|\log n)$ time.

The changes in the polygon are equivalent to a sequence of ws-equivalent primitives:
a \nodesplit\/ operation at $u$, followed by a sequence of \nodesplit s
along the chain $C_0$ from $\ell_0$ to $u$, and subdivision operations
that create the new leaf nodes between the two groups.
The interval tree is updated by deleting the intervals that contain $u$,
and the query time remains the same output-sensitive $O(\log n+k)$.
Consequently, we can split $G_{uv}$ in $O(\min(|\BB^-|,|\BB^+|)+|\VV^0|\log n+\log n)$ time.

\subsection{Analysis of the spur-elimination algorithm}
\label{ssec:time}

\begin{lemma}\label{lem:group-simplfication-time}
Given $m$ benchmark vertices, {\sf spur-elimination}$(P,\GG)$ takes $O(m\log m)$ time.
\end{lemma}
\begin{proof}
Let $\sigma$ be the number of spurs, $\beta$ the number of benchmark vertices at the leaves of clusters, and let $\phi=2\sigma+\beta$. Initially, $\phi=O(m)$ by \ref{inv:tree}. All operations in {\sf spur-elimination} monotonically decrease both $\sigma$ and $\phi$. Step~\ref{step:low} decreases $\phi$ by the number of spurs at $u$, and steps~\ref{step:merge} and \ref{step:split} both maintain $\phi$. In particular, Step~\ref{step:split} converts some spurs into pairs of adjacent benchmark vertices at leaves. Consequently, the number of benchmark vertices remains $O(m)$ throughout the algorithm.

Step~\ref{step:setup} creates data structures for new groups: For a group containing $m$ benchmarks, all supporting data structures can be computed in $O(m)$ time, that is, in $O(1)$ time per benchmark. A new benchmark $v$ appears in a group when (i) a benchmark is extracted into an adjacent group, or (ii) a group of size $m$ is split and $v$ is part of the smaller group of size at most $m/2$. Extraction strictly decreases $\phi$, so it occurs $O(m)$ times. The total number of benchmarks that are either present initially or created by extraction is $O(m)$. Each of these benchmarks can move into a group of half-size $O(\log m)$ times. Consequently, there are $O(m\log m)$ new benchmarks overall, and the time spent on all instances of Steps~\ref{step:setup} is $O(m\log m)$.

Step~\ref{step:merge} removes an interior node of degree two; the update of supporting data structures takes $O(\log m)$ time. Interior nodes are created only when they contain a spur, so at most $O(m)$ interior nodes are ever created, and all instances of Step~\ref{step:merge} take $O(m \log m)$ time.
Step~\ref{step:low} eliminates $|\calS|$ spurs in $O(|\calS|\log m)$ time.
Eventually, all spurs are eliminated, thus all instances of Step~\ref{step:low} take $O(m\log m)$ time.
Step~\ref{step:split} takes $O(\min(|\BB^-|,|\BB^+|)+|\VV^0|\log m+ \log m)$ time.
By a standard heavy-path decomposition argument, the terms $\min(|\BB^-|,|\BB^+|)$
contribute $O(m\log m)$ time.
Every chain in $\VV^0$ corresponds to a spur that is destroyed in a step~\ref{step:split}
(and no new spurs are created in step~\ref{step:split}),
therefore the terms $O(|\VV^0|\log m)$ sum to $O(m\log m)$
over the course of the algorithm. Since every execution of step~\ref{step:split} increases
the number of groups by one, and this step is repeated $O(m)$ times,
the $\log m$ terms sum to $O(m\log m)$ in the entire algorithm.
\end{proof}

Algorithm {\sf spur-elimination}$(P,\GG)$ returns a polygon $P'$, a set $\GG'$ of groups, and a set $\BB'$ of benchmark-to-benchmark paths, each of which connects two leaves in two different clusters of a group.
In each group $G_{uv}$, the trees $T[u]$ and $T[v]$ have no interior nodes,
thus $G_{uv}$ consists of two single-node clusters $C(u)=\{u\}$ and $C(v)=\{v\}$, connected by a single edge $uv$.
Consequently, the image graph is 2-regular.
We can now decide whether $P'$ is weakly simple in $O(n)$ time similarly to~\cite[Section~3.3]{CEX15}.
The polygon $P'$ is weakly simple if and only if the image graph is connected and each group contains precisely one benchmark-to-benchmark path.
These properties can be verified by a simple traversal of the image graph and $P'$ in $O(n)$ time.
This completes the proof of Theorem~\ref{thm:main}.

\section{Perturbing weakly simple polygons into simple polygons}
\label{sec:perturb}

In Sections~\ref{sec:preprocess}--\ref{sec:tree-exp}, we have presented an algorithm that
decides, in $O(n\log n)$ time, whether a given $n$-gon $P$ is weakly simple.
If $P$ is weakly simple, then for every $\varepsilon>0$ it can be perturbed into a simple
polygon by moving each vertex a distance at most $\varepsilon$. In this section we show how to
find, for any $\eps>0$, a simple polygon $Q$ with $2n$ vertices such that ${\rm dist}_F(P,Q)<\eps$.
Let $P'$ and $P''$ be the polygons obtained after the bar-simplification and spur-elimination phases
of the algorithm, respectively. $P''$ has $O(n)$ vertices, none of which is a fork or a spur.
Using the results in \cite[Section~3]{CEX15}, we can construct a simple polygon $Q''\in \Phi(P'')$ in $O(n)$ time.
In this section, we show that we can reverse the sequence of operations in $O(n \log n)$
time and perturb $P$ as well into a simple polygon $Q\in\Phi(P)$.

\medskip\noindent{\bf Combinatorial representation by bar-signatures.}
A perturbation of a weakly simple polygon has a combinatorial representation, called a signature, which consists of total orders of the overlapping edges in all segments of the image graph (cf. Section~\ref{sec:preliminaries}). In the absence of forks, every edge lies in a segment, and the size of such a signature is $O(n)$. However, the signature may have size $\Theta(n^2)$ in the presence of forks. When our algorithm eliminates forks from a polygon, it may create $\Theta(n^2)$ dummy vertices and edges, which would again lead to a signature of size $\Theta(n^2)$. For reversing the operations of the algorithm in Sections~\ref{sec:preprocess}--\ref{sec:tree-exp}, we introduce a new combinatorial representation of size $O(n)$ that maintains the total order of the edges in each bar that are outside of clusters.

For $n\geq 3$, let $P=(p_0,\ldots , p_{n-1})$ be a weakly simple polygon with image graph $G$.
Assume that the sober nodes of $G$ are partitioned into a set $\mathcal{C}$ of disjoint clusters satisfying invariants \ref{inv:tree}--\ref{inv:subdiv} such that every bar is either entirely in a cluster or outside of all clusters.
Let $Q=(p'_0,\ldots , p'_{n-1})$ be a simple polygon such that $|p_i,p'_i|<\eps_0=\eps_0(P)$ for all $i=0,\ldots , n-1$.
We may assume that $G$ has no vertical segments (so that the above-below relationship is defined between disjoint segments parallel to a bar).
In each segment $uv$ of $G$ outside of clusters, the above-below relationship yields a total ordering over the edges of $Q$ that contain $uv$. For each bar $b$ outside of clusters, the total orders of the segments along $b$ are consistent (since the above-below relationship between two edges is the same in every corridor). Consequently, the transitive closure of these total orders is a partial order over all edges in $b$. Consider a linear extension of such a partial order. The collection of these total orders for all bars in $P$ is a \emph{bar-signature} of $Q$. Since the linear extensions need not be unique, a polygon $Q\in \Phi(P)$ may have several bar-signatures.

Given a bar-signature of a perturbation of $P$, we can (re)construct an approximate simple polygon $Q'$ as follows; refer to Figure~\ref{fig:comb}. For every bar $b=uv$ of $G$ outside of clusters, let the \emph{volume} $\vol(uv)$ be the number of edges of $P$ that lie on $b$. Place $\vol(uv)$ parallel line segments, called \emph{lanes}, between $\partial D_u$ and $\partial D_v$ in the region $U_{\eps}$, ordered from bottom to top (the lanes contain the edges of $Q'$). For the $i$-th edge $pq$ in the total order of $b$, let the corresponding edge in $Q'$ be the shortest edge connecting $\partial D_p$ and $\partial D_q$ in the $i$-th lane. For each cluster $C(u)$, denote by $R(u)$ the union of all disks $D_v$, $v\in C(u)$, and all corridors between nodes in $C(u)$. If $C(u)$ contains only the node $u$, then $R(u)=D_u$, but $R(u)$ is always simply connected since $C(u)$ induces a tree $T[u]$. For each cluster $C(u)$, construct a noncrossing polyline matching, between the endpoints of the edges in $\partial R(u)$, that connects the endpoints corresponding to a maximal subpath in $T[u]$. The edges in the lanes and the perfect matchings in the regions $R(u)$ produce a polygon $Q'$. If the Euclidean diameter of each region $R(u)$ is at most $\delta$, then the Fr\'echet distance between $P$ and $Q'$ is at most $\eps+\delta$.
Denote by $\Psi(P)$ the set of all simple polygons that can be constructed in this manner from a bar-signature for some $\eps$, $0<\eps<\eps_0$.

\begin{figure}[h!tbp]
\centering
		\def\svgwidth{\textwidth}
\begingroup%
  \makeatletter%
  \providecommand\color[2][]{%
    \errmessage{(Inkscape) Color is used for the text in Inkscape, but the package 'color.sty' is not loaded}%
    \renewcommand\color[2][]{}%
  }%
  \providecommand\transparent[1]{%
    \errmessage{(Inkscape) Transparency is used (non-zero) for the text in Inkscape, but the package 'transparent.sty' is not loaded}%
    \renewcommand\transparent[1]{}%
  }%
  \providecommand\rotatebox[2]{#2}%
  \ifx\svgwidth\undefined%
    \setlength{\unitlength}{439.49633789bp}%
    \ifx\svgscale\undefined%
      \relax%
    \else%
      \setlength{\unitlength}{\unitlength * \real{\svgscale}}%
    \fi%
  \else%
    \setlength{\unitlength}{\svgwidth}%
  \fi%
  \global\let\svgwidth\undefined%
  \global\let\svgscale\undefined%
  \makeatother%
  \begin{picture}(1,0.29124992)%
    \put(0,0){\includegraphics[width=\unitlength,page=1]{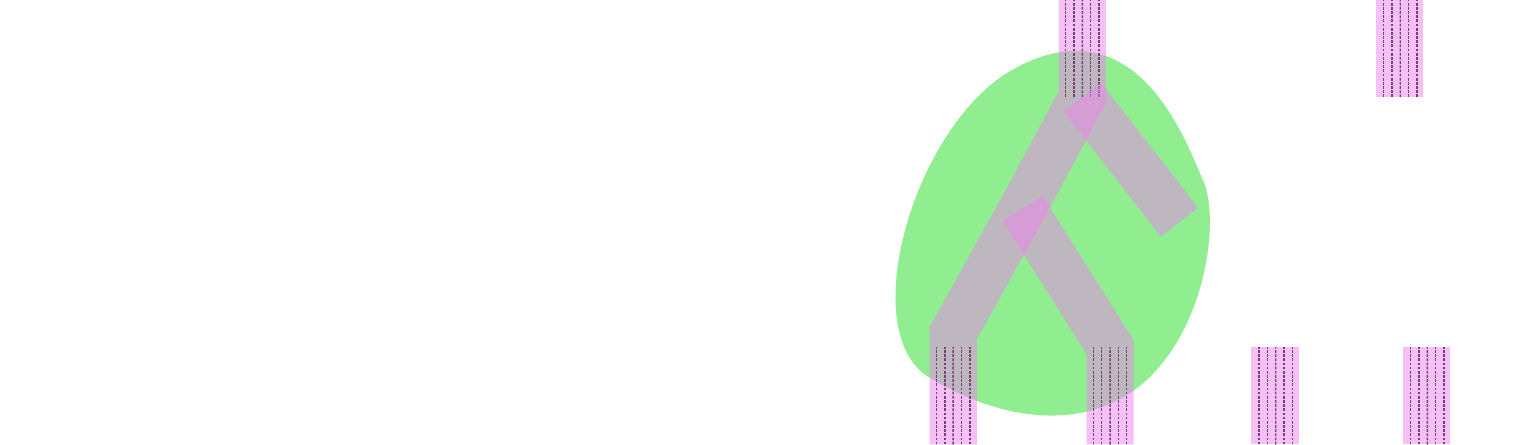}}%
    \put(-0.0238531,0.16343463){\color[rgb]{0,0,0}\makebox(0,0)[lb]{\smash{$u$}}}%
    \put(0.53015159,0.05989555){\color[rgb]{0,0,0}\makebox(0,0)[lb]{\smash{$v$}}}%
    \put(0,0){\includegraphics[width=\unitlength,page=2]{fig-combinatorial.pdf}}%
    \put(0.74248043,0.22326319){\color[rgb]{0,0,0}\makebox(0,0)[lb]{\smash{$u$}}}%
    \put(0.56266881,0.2008106){\color[rgb]{0,0,0}\makebox(0,0)[lb]{\smash{$C(u)$}}}%
    \put(0,0){\includegraphics[width=\unitlength,page=3]{fig-combinatorial.pdf}}%
    \put(0.96708557,0.09159466){\color[rgb]{0,0,0}\makebox(0,0)[lb]{\smash{$R(u)$}}}%
    \put(0,0){\includegraphics[width=\unitlength,page=4]{fig-combinatorial.pdf}}%
  \end{picture}%
\endgroup%
\caption{\small
Construction of a simple polygon $Q'\in \Psi(P)$ from a bar-signature. Left: Bar $uv$ of a simple polygon obtained from an order compatible with the polygon shown in Figure~\ref{fig:bars}(c).
Right: maximal paths of $Q$ and $Q'$ inside clusters.
\label{fig:comb}}
\end{figure}

\medskip\noindent{\bf Spur elimination.}
If a given $n$-gon is weakly simple, our decision algorithm computes a polygon $P''$, which is ws-equivalent to $P$ and represented implicitly by a cyclic sequence of benchmark nodes. Specifically, $P''$ is represented by an image graph $G''$, a set $\GG''$ of groups, a set $\BB''$ of benchmark-to-benchmark paths, and for every group $G_{uv}\in \GG''$, a linear order of the paths in $\BB''$ that cross the corridor $N_{uv}$ between $D_u$ and $D_v$. Consequently, the decision algorithm provides a bar-signature for the weakly simple polygon $P''$.

We show that, by reversing the steps of Algorithm {\sf spur-elimination}$(P',\GG')$, we can compute a bar-signature of $P'$ in $O(n\log n)$ time.
If a group $G_{uv}$ has been split in some step~\ref{step:split} (cf. Section~\ref{ssec:split}), we can construct an ordering of the benchmark-to-benchmark paths of $G_{uv}$ by concatenating the orders of $\BB^-$ and $\BB^+$ (the sets of benchmark-to-benchmark paths of the resulting two groups).

If $G_{uv}$ had spurs eliminated from $u$ in some step~\ref{step:low} (cf.~Section~\ref{ssec:low}), we reverse each of the steps in the following manner. Recall that if a new group $G_{\ell^- \ell^+}$ was created, then every path $[t_1';t_2']\in\BB'$ in that group was created from a concatenation of two paths $[t_1;u]$ and $[t_2;u]$.
Use the ordering of the paths in $\BB'$ to insert the paths $[t_1;u]$ and $[t_2;u]$ into the ordering of $\BB$ so that they form nested spurs, i.e., if $[t_1';t_2']$ is the topmost edge in $\BB'$, $[t_1;u]$ (resp. $[t_2;u]$) should be the leftmost (resp., rightmost) path (without loss of generality, we use the orientation of Figure~\ref{fig:spur-shortcut}).
Identify the leftmost path in the segment that connects $\ell$ and its right child and place all nested paths that created $G_{\ell^- \ell^+}$ immediately to its left.

If one or more spurs were created at a node $z$ in an adjacent group, we can find the position of the edges incident to each spur in the ordering of the adjacent group. Using this order, we can identify the first path in $G_{uv}$ to the right of the edges incident to $[z]$. Then, immediately to the left of such a path, we can place the paths $[t_1;u;t_1]$ that generated the spurs at $z$. The relative order of these paths is the same as the one obtained by reversing a \spurreduction, described in the proof of Lemma~\ref{lem:primitives1}, and therefore produces a simple polygon. If a path $[t_1;u;t_2]$ is simplified to  $[t_1]$ (Step~\ref{step:trivi1} with $\min(t_1,t_2)=t_1$ without loss of generality; or Step~\ref{step:trivi2}), we can proceed analogously to the reversal of a \crimpreduction\ (cf.~Lemma~\ref{lem:crimp}) from a path $[t_1;u;t_2;s]$ to $[t_1;s]$. Identify the path $[t_1;s]$ in the ordering of $\BB$ and replace it with the paths $[t_1;u]$, $[u;t_2]$, and $[t_2;s]$ in this order.

\medskip\noindent{\bf Bar simplification.}
The bar-signature determines all segments between adjacent clusters. Using these orders, we can reverse the operation \pinextraction$(u,v)$\ assigning the same order for the edges in $uv$ as the order of its adjacent benchmark-to-benchmark paths. \Vshortcut\ is also trivially reversible by concatenating the order of segments that get merged.

Updating the bar-signature when we reverse an \Lshortcut\ operation is a bit more challenging.
Determining the edge order in segments $vw$ and $vu_1$ can be trivially done by just concatenating the order of merged segments. But phase~\ref{phase1} introduces a crimp in some cross-chains, and the reverse operation, \crimpreduction, may require nontrivial reordering in the bar-signature. Suppose that $P'$ is obtained from $P$ after a \crimpreduction.
The proof of Lemma~\ref{lem:crimp} shows a straightforward way to obtain a bar-signature of a polygon in $\Psi(P)$ given a polygon in $\Psi(P')$.
However, obtaining a bar-signature of $Q'\in\Psi(P')$ given $Q\in\Psi(P)$ requires identifying $W_{top}$ and $W_{bot}$, which takes $O(n)$ time.

In order to handle the reversal of phase~\ref{phase1} in $O(1)$ time, we divide the signature of each bar into pieces.
Recall that the {\sf bar-simplification} algorithm does not eliminate any cross-chains from $D_b$, and when {\sf bar-simplification} terminates, only one-edge cross-chains remain in the interior of $D_b$. Let $K$ denote the set of cross-chains of $D_b$. The segments of the image graph that cross the ellipse $D_b$, and the bar-signatures of these edges yield a linear order (from left to right) of $K$; and the cross-chains subdivide $D_b$ into $|K|+1$ regions. We maintain a linear order for the edges along the bar in each such region (including the boundary of the region), and denote the set of these edges in $b$ by $E_1,\ldots, E_{|K|+1}$.

We reverse phases~\ref{phase2} and \ref{phase3} of \Lshortcut$(v,TR)$ as follows (applying reflections for other \Lshortcut\ operations if necessary). Assign the new edges $[u_1,u_2]$ the highest lanes in the ordering of the appropriate $E_i$, maintaining the relative order of affected paths. To reverse phase~\ref{phase1}, first notice that the three edges in the crimp $[u_1,u_2,u_1,u_2]$ are part of a cross-chain, consequently they appear in two consecutive subsets $E_i$ and $E_{i+1}$. In the ordering of the left (resp., right) subset, assign the new edge $[u_1,u_2]$ to the highest (resp., lowest) position among the positions of the three edges $[u_1,u_2]$.

When all operations in the bar simplification algorithm have been reversed, we have to combine the linear orders of $E_1,\ldots, E_{|K|+1}$ into a total order, a common linear extension of these orders. The intersection of two edge sets, $E_i\cap E_j$ with $i<j$, is either disjoint or contains the edges of the $i$-th cross-chain. The above-below relationship between the edges of each cross-chain is uniquely determined by Lemma~\ref{lem:irreducible}, and must be the same in each total order. Therefore, the union of the total orders is a partial order for all edges in the bar. Since the ordering of each subset guarantees that its paths can be realized without crossing, any linear extension of this partial order produces a bar-signature of a simple polygon.

\medskip\noindent{\bf Preprocessing.}
The cluster formation and \newbarexp\ consist of \subdivision\ operations that do not influence the order of edges that define the bar-signature. If an edge $[v,w]$ in a bar $b$ is subdivided into $[v,v',w]$, where $[v',w]$ is in $D_b$, we can assign $[v,w]$ to the same lane of $[v',w]$ in the ordering of edges in $b$.
The \crimpreduction\ operations can be reversed by making the three edges that form a new crimp consecutive in the ordering, as in the proof of Lemma~\ref{lem:crimp}.

We have shown how to maintain bar-signatures while reversing the operations of our algorithms, in time proportional to those operations. For every $\varepsilon>0$, the bar-signatures yield a perturbation of a weakly simple polygon $P$ into a simple polygon $Q\in \Phi(P)$ with $2n$ vertices, where each vertex $[u]$ of $P$ corresponds to two vertices of $Q$  on the circle $\partial D_u$. This completes the proof of Theorem~\ref{thm:perturb}.

\section{Conclusion}

We presented an $O(n\log n)$-time algorithm for deciding whether a polygon with $n$ vertices is weakly simple.
Weak simplicity of polygons has a natural generalization for planar graphs~\cite[Appendix D]{CEX15}.
We can define the {\em weak embedding} for graphs in terms of Fr\'echet distance. A graph $H=(V,E)$ can be considered a 1-dimensional simplicial complex. A \emph{drawing} of $H$ is a continuous map of $H$ to $\mathbb{R}^2$. The Fr\'echet distance between two drawings, $P$ and $Q$, of $H$ is defined as ${\rm dist}_F(P,Q)=\inf_{\phi:H\rightarrow H}\max_{x\in H}{\rm dist}(P(\phi(x)),Q(x))$, where $\phi$ is an automorphism of $H$ (a homeomorphism from $H$ to itself).
Very recently, Fulek and Kyn\v{c}l~\cite{FK17} gave a polynomial-time algorithm for deciding
whether a given drawing of a graph $H$ is weakly simple, i.e.,
whether a straight-line drawing $P$ of $H$ is within $\eps$ Fr\'echet distance from some embedding $Q$ of $H$, for all $\eps>0$.
Earlier, efficient algorithms were known only in special cases:
when the embedding is restricted to a given isotopy class (i.e., given combinatorial embedding)~\cite{Ful16};
and 
when all $n$ vertices are collinear and the isotopy class is given~\cite{ADD+14}.

We can also generalize the problem to higher dimensions.
A polyhedron can be described as a map $\gamma: M\rightarrow \mathbb{R}^3$, where $M$ is a 2-manifold without boundary. A simple polyhedron is an injective function.
A polyhedron $P$ is weakly simple if there exists a simple polyhedron within $\eps$ Fr\'echet distance from $P$ for all $\eps>0$.
This problem can be reduced to origami flat foldability.
The results of~\cite{BH96} imply that, given a convex polygon $P$ and a piecewise isometric function $f:P\rightarrow \mathbb{R}^2$ (called \textit{crease pattern}), it is NP-hard to decide if there exists an injective embedding of $P$ in three dimensions $\lambda:P\rightarrow \mathbb{R}^3$ within $\eps$ Fr\'echet distance from $f$ for all $\eps>0$, i.e., if $f$ is \textit{flat foldable}.
Given $P$ and $f$, we can construct a continuous function $g:\mathbb{S}^2\rightarrow P$ mapping each hemisphere of $\mathbb{S}^2$ to $P$ (for a point $x\in P$, the inverse image $g^{-1}(x)$ is a set of two points in opposite hemispheres of $\mathbb{S}^2$).
Then, the polyhedron $\gamma = g\circ f$ is weakly simple if and only if $f$ is flat foldable.
Therefore, it is also NP-hard to decide whether a polyhedron is weakly simple.

Finally it is an open problem to find a linear-time algorithm for recognizing weakly simple polygons.
Chang et al.~\cite{CEX15} conjectured that this is possible in the absence of  spurs and forks.

\paragraph{Acknowledgements.}
Research by Akitaya, Aloupis, and T\'oth was supported in part by the NSF awards CCF-1422311 and CCF-1423615.
Akitaya was supported by the Science Without Borders program.
Research by Erickson was supported in part by the NSF award CCF-1408763.
We thank Anika Rounds and Diane Souvaine for many helpful conversations that contributed to the
completion of this project. We thank the anonymous referees for many useful comments and suggestions.

\end{document}